\documentclass[12pt,titlepage, twoside, openright]{book}
\usepackage{Preamble}

\begin{document}

\theoremstyle{definition}
\newtheorem{definition}{Definition}[chapter]
\newtheorem{assumption}[definition]{Assumption}
\newtheorem{notation}[definition]{Notation}
\newtheorem{remark}[definition]{Remark}

\theoremstyle{remark}
\newtheorem{example}{Example}[chapter]

\theoremstyle{plain}
\newtheorem{theorem}{Theorem}[chapter]
\newtheorem{corollary}[theorem]{Corollary}
\newtheorem{proposition}[theorem]{Proposition}
\newtheorem{lemma}[theorem]{Lemma}

\begin{titlepage}
    \newgeometry{margin=3cm}
    \noindent\makebox[\textwidth]{%
    \begin{minipage}[l]{0.48\textwidth}
        \centering
        \includegraphics[width=0.85\linewidth]{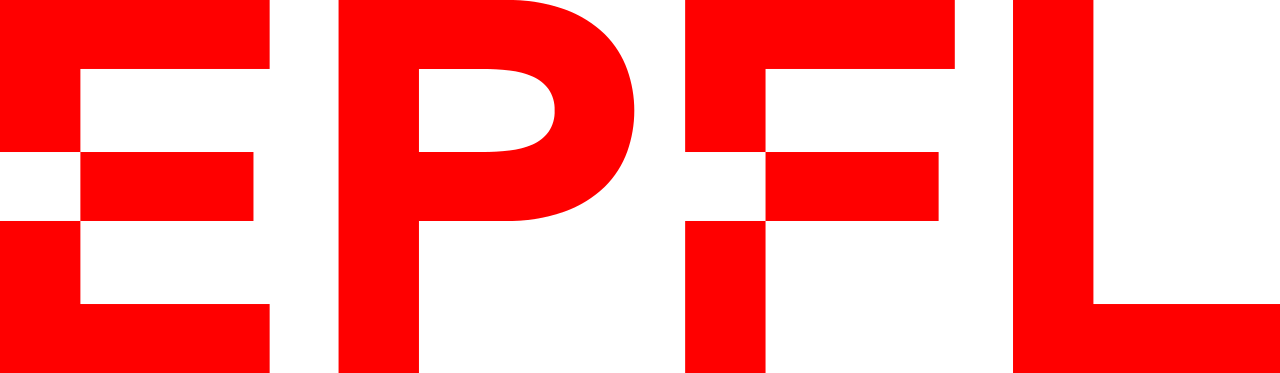}
    \end{minipage}\hfill
    \begin{minipage}[r]{0.48\textwidth}
        \centering
        \includegraphics[width=0.51\linewidth]{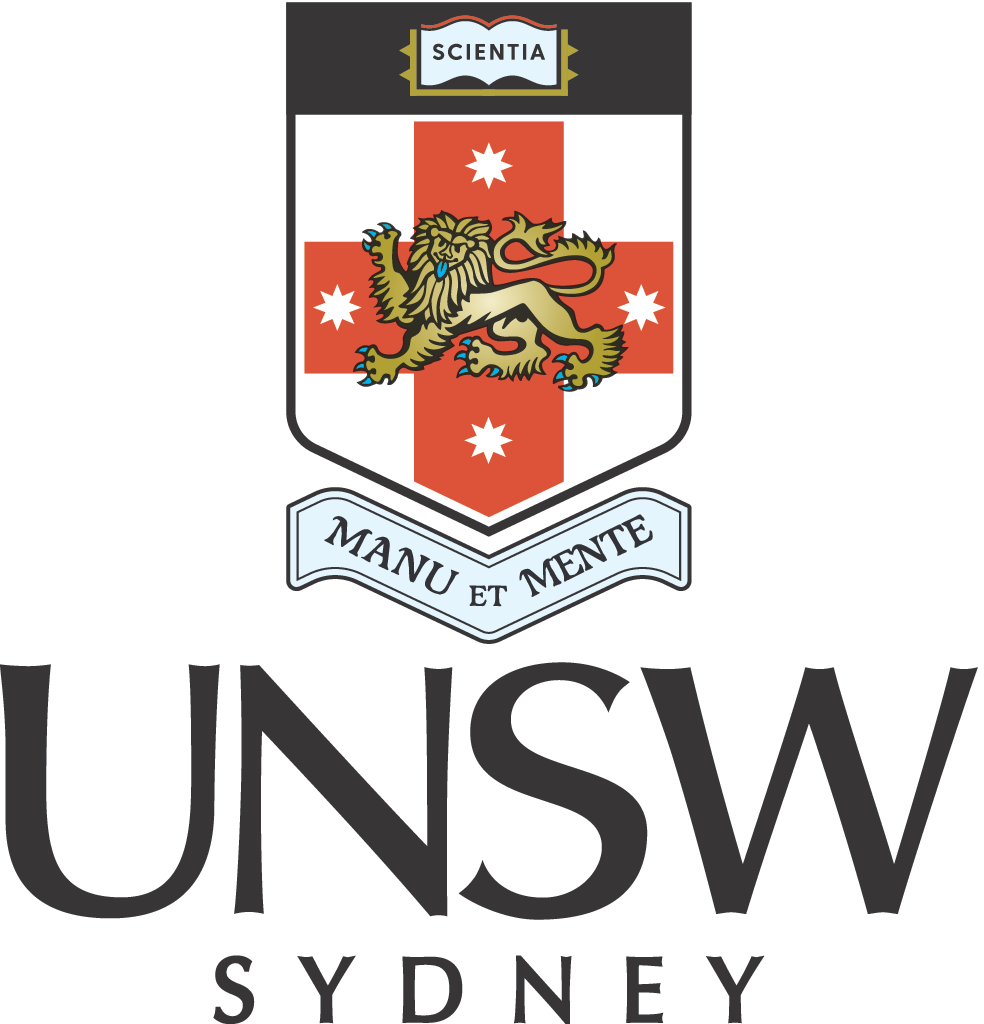}
    \end{minipage}%
    }%
	\centering
    \vspace{\fill}
    \textbf{\color{darkgreen} \scshape \parbox[t]{\textwidth}{\centering\fontsize{35}{45}\selectfont Regenerative Rejection Sampling}}\\ \vspace{\fill}		
	{\large Master's project, September 2025 -- March 2026}\\[0.4cm]
	\rule{\linewidth}{0.2 mm} \\[0.5 cm]

    \begin{multicols}{2}
    \noindent\textbf{Author:} \\
    Tommaso Bozzi \\[1em]

    \columnbreak

    \noindent\textbf{Supervisors:} \\
    Prof. Fabio Nobile,\\
    \textit{\small EPFL, Insitute of Mathematics -- CSQI Chair}\\
    Dr. Zdravko Botev,\\
    \textit{\small UNSW Sydney, School of Mathematics and Statistics}
    \end{multicols}

    \vspace{2cm}
	{\small 6th March 2026}
\end{titlepage}
\restoregeometry

\pagenumbering{Roman}
{
  \pagestyle{numberonly}

\clearpage

\begin{abstract}
This thesis presents Regenerative Rejection Sampling (RRS), a novel approximate sampling algorithm inspired by classical Rejection Sampling and Markov Chain Monte Carlo methods. The method constructs a continuous-time regenerative process whose stationary distribution coincides with a target density known only up to a normalizing constant. Unlike standard Rejection Sampling, RRS does not require the existence of a finite constant that upper-bounds the likelihood ratio. As a result, its total variation convergence rate remains exponential for a larger class of scenarios compared to, for example, the Independent Metropolis-Hastings sampler, which requires a finite bounding constant.

To explain the workings of the method, we first present a detailed review of renewal and regenerative processes, including their limit theorems, stationary versions, and convergence properties under standard conditions. We explain a coupling proof for exponential convergence of regenerative processes, under the assumption of a spread-out cycle length distribution.

We then introduce the RRS algorithm, and derive its convergence rate. Its performance is compared theoretically and empirically with classical MCMC methods. Numerical experiments demonstrate that RRS can exhibit lower autocorrelations and faster effective mixing, both in synthetic examples and in a  Bayesian probit regression model applied to a real medical dataset. Moreover, if the algorithm is run until time $t$, we show that the usual order $ O(1/t)$ results for the bias of the time-average estimators, is improved to a bias of $ O(1/t^2)$ for the estimator constructed from the RRS method, and provide easy-to-estimate non-asymptotic bounds for this bias.
\end{abstract}
\clearpage

\setcounter{page}{1}
  \tableofcontents
  \listoffigures
  \listoftables
  \cleardoublepage
}
\pagestyle{main}

\newpage
\pagenumbering{arabic}
\setcounter{page}{1}
\chapter{Introduction}
\label{chap:1}
One of the basic, but nonetheless important, concepts in Monte Carlo methods \cite{Metropolis:MC} is sampling. It is used to obtain realizations of random variables (or vectors) following a given distribution, i.e. the \emph{samples}, which are then exploited to compute estimators of given Quantities of Interest. The fundamental sampling methods are those tailored to generate uniform random variables on $(0,1)$, called \emph{Uniform Random Number Generators} (URNGs) \cite{Kroese:HandbookMC}. The majority of the URNGs used nowadays, are based on methods that can be implemented on computers. This implies that such algorithms only produce a \emph{deterministic} stream of uniform numbers, that is, however, indistinguishable from a random one according to common statistical tests \cite{Lecuyer:RNGs}. These type of generators are called \emph{Uniform Pseudo-Random Number Generators}.\\
To obtain (pseudo-) random numbers according to a general distribution, one can then apply specific transformations to the uniform samples (e.g. the \emph{Inverse Transform Method} \cite{Kroese:HandbookMC,Lecuyer:RNGs}), or other methods that make use of the uniform (pseudo-) random numbers.

The \emph{Rejection Sampling method} falls under the latter category. It is one of the most common methods to generate exact random numbers from any distribution with probability density function $f(x)$, possibly known only up to a multiplicative constant (i.e. only $f_\propto(x)=\kappa f(x)$ is known) \cite{Martino:Sampling}. The scheme relies on sampling random variables from a different distribution, called the \emph{proposal distribution}, which is assumed to have density $g(x)$. The method only works if there exists a known constant $C>0$, such that $f_\propto(x)\leq Cg(x)$ for all $x$ (or $g$-almost everywhere). Unfortunately, this assumption can be quite restrictive, since, in various cases, calculating the constant $C$ is computationally too expensive, or even impossible (e.g. if such a finite constant does not exist).

One could resort to approximate methods, such as the Independent Metropolis-Hastings (IMH) Markov Chain Monte Carlo (MCMC) method \cite{Hastings:MCMC,Tierney:IndependenceSampler}, which do not require the constant $C$ to be known, but, in turn, generate random numbers that are only approximately distributed according to the given pdf $f(x)$. However, the IMH MCMC method has fast (exponential) convergence only when a finite constant $C>0$, such that $f_\propto(x)\leq Cg(x)$, exists \cite{MergensenTweedie:MHConvergence}.

In this thesis, we take inspiration from both MCMC methods, which construct a Markov Chain that takes $f(x)$ as stationary distribution, and Rejection Sampling, and propose a novel approximate sampling method based on regenerative processes, which we call \emph{Regenerative Rejection Sampling} (RRS). Regenerative processes have been widely used in simulation, for example, in \cite{Glynn:GSMP, Glynn:RegenerativeSimulation, Glynn:SimulationRegenerative, AwadGlynn:SteadyStateEstimators, Asmussen:StationarityDetection,MeketonHeidelberger:BiasReduction,CraneIglehart:SimulatingStableStochSystIII}, and also in the context of Markov Chain Samplers, e.g. in \cite{Mykland:RegenerationMCMC} and \cite[Truncated Multivariate Student Computations via Exponential Tilting]{Botev:AdvancesModeling}. However, to our knowledge, the method that we propose in this report has not been theoretically studied yet, as it was only recently introduced in \cite{Botev:MachineLearning}.

The algorithm is a generalization of Rejection Sampling, as it relies on samples generated from a proposal distribution with density $g(x)$. Instead of a Markov Chain, RRS constructs a regenerative process with stationary distribution $f(x)$. In this thesis, we will show that the method can be applied in many of the cases where the Rejection Sampling algorithm cannot be used (i.e. with intractable or non-existent constant $C$), and that it converges exponentially fast, in total variation, to its stationary distribution for a larger class of instances compared to the IMH MCMC method. Additionally, we will analyze its performance on some examples, and show that it can behave significantly better than alternative and widely-used MCMC methods.

\section{General Outline}
The thesis is outlined as follows:
\begin{itemize}
    \item In Chapters \ref{chap:2} and \ref{chap:3}, we present the theory of renewal and regenerative processes. For the exposition, we closely followed the book by Asmussen \cite{Asmussen:Applied}, together with other additional sources. However, we provided additional simple examples and expanded on Asmussen's arguments, filling in the missing steps, to ease the understanding of the concepts.
    \item Chapter \ref{chap:4} is devoted to the coupling proof of the exponential convergence of regenerative processes. Once again, we closely followed \cite{Asmussen:Applied}, expanding some steps of the proofs and providing simple examples. The last Section of the Chapter (Section \ref{sec:CouplingProofDetailed}) explores the construction of a coupling proof through the lense of \cite{Thorisson:RegProcCoupling}.
    \item In Chapter \ref{chap:5} we first recall the theory on which Rejection Sampling is based, and then, we present the \emph{Regenerative Rejection Sampling} method.
    \item In Chapter \ref{chap:6} we concentrate on the bias and variance properties of the time average estimator constructed from a run of the Regenerative Rejection Sampling method for the purpose of estimating an expected value.
    \item Chapter \ref{chap:7} shows two practical applications of the RRS method: first, by sampling from a toy bi-dimensional distribution, and then by performing a Probit Bayesian Regression on a real medical dataset containing  $55$ patients, $18$ of which have been diagnosed with \emph{Latent membranous lupus nephritis} \cite{AlbertChib:BayesBinary}. We also include a comparison with common MCMC methods.
    \item Lastly, in Chapter \ref{chap:8}, we summarize the results of our work, and provide commentary on them.
\end{itemize}

\chapter{Continuous-time Renewal Processes: Background}
\label{chap:2}
The study of regenerative processes lays its foundations in the analysis of renewal processes. As will be explained in the dedicated chapter, we can embed a renewal process in every regenerative process. Hence, it is natural to start our survey from the theory of renewal process, also called \emph{Renewal Theory}. We will concentrate on continuous-time processes, since the RRS method develops in continuous time, as we will see in Chapter \ref{chap:5}.

\section{Basic definitions and results}
Let us start the analysis with the classical definition of renewal process \cite{Asmussen:Applied}.
\begin{definition}[Renewal Process] 
    Let $0 \leq T_0 < T_1 < T_2 < \cdots$ be the (random) times of occurrences of some event and define $X_n \colon \!\!\! = T_n - T_{n-1}$, $X_0 = T_0$. Then $\{ T_n\}_{n \in \mathbb{N}}$ is called a \emph{Renewal Process} if $X_0,X_1,X_2, \dots$ are independent and $X_1,X_2,X_3,\dots$ (not necessarily $X_0$) have the same distribution.

    The $T_n$'s are the \emph{renewals}, or \emph{epochs}, of the process. The common distribution $F$ of $X_1,X_2, \dots$ is called \emph{interarrival distribution} or \emph{waiting-time distribution}.

    A renewal process is said to be \emph{pure} or \emph{zero-delayed} if $X_0 = T_0 = 0 \: \text{a.s.}$. Otherwise it is called \emph{delayed}, and the \emph{delay distribution} is the distribution of $X_0$.
\end{definition}

Let us make a few remarks. Since we want to avoid having more than one renewal at a time, we always assume that the $X_n$'s, $n\geq 1$, have zero mass at $0$, meaning $F(0)=0$ \cite{Asmussen:Applied}.
Moreover, in the case of a zero-delayed process, we choose to count $T_0=X_0=0$ as a proper renewal of the process.

Another important definition related to each renewal process, is the counting process of the epochs, which we call $\{N(t)\}_{t\geq 0}$.

\begin{definition}[Renewal Counting Process]
    A \emph{Renewal Counting Process} $\{ N(t)\}_{t \geq 0}$, associated to a renewal process $\{T_n\}_{n\in\mathbb{N}}$, is defined as 
    \[
    N(t) = \inf \{n \in \mathbb{N}: T_n > t\}.
    \]
    Notice that $N(0) = 1$ in the case of a zero-delayed renewal process and $N(0)=0$ for a delayed one.
\end{definition}

In the related literature, one can find an alternative definition of the renewal counting process (e.g. in \cite{Grimmett:RandomProcesses}), i.e.
\[
\tilde{N}(t) = \sup\{n \in \mathbb{N}: T_n \leq t \}.
\]
However, we use the definition with the $\inf$ for the proposed one because it is more intuitive, since $N(t)$ represents the number of renewals up to time $t$ (as previously remarked, the starting renewal at $0$ counts as a proper renewal, in the case of a zero-delayed renewal process), and is always well defined. On the other hand, $\tilde{N}(t)$, in the case of a delayed renewal process, for all $t < T_0$ would yield $N(t) = \sup \varnothing$, which requires us to specify a value for such a quantity. For the specific choice of $\tilde{N}(t) = 0$ for all $t \leq T_0$, we get the useful relation $N(t) = \tilde{N}(t) + 1$.

Figure \ref{fig:NtANDTn} shows a realization of a zero-delayed renewal process $\{T_n\}_{n\in\mathbb{N}}$ with a $Gamma(2,1)$ interarrival distribution, together with a comparison of $\{N(t)\}$ and $\{\tilde{N}(t)\}$.
\begin{figure}[htpb]
    \centering
    \includegraphics[width=0.55\linewidth]{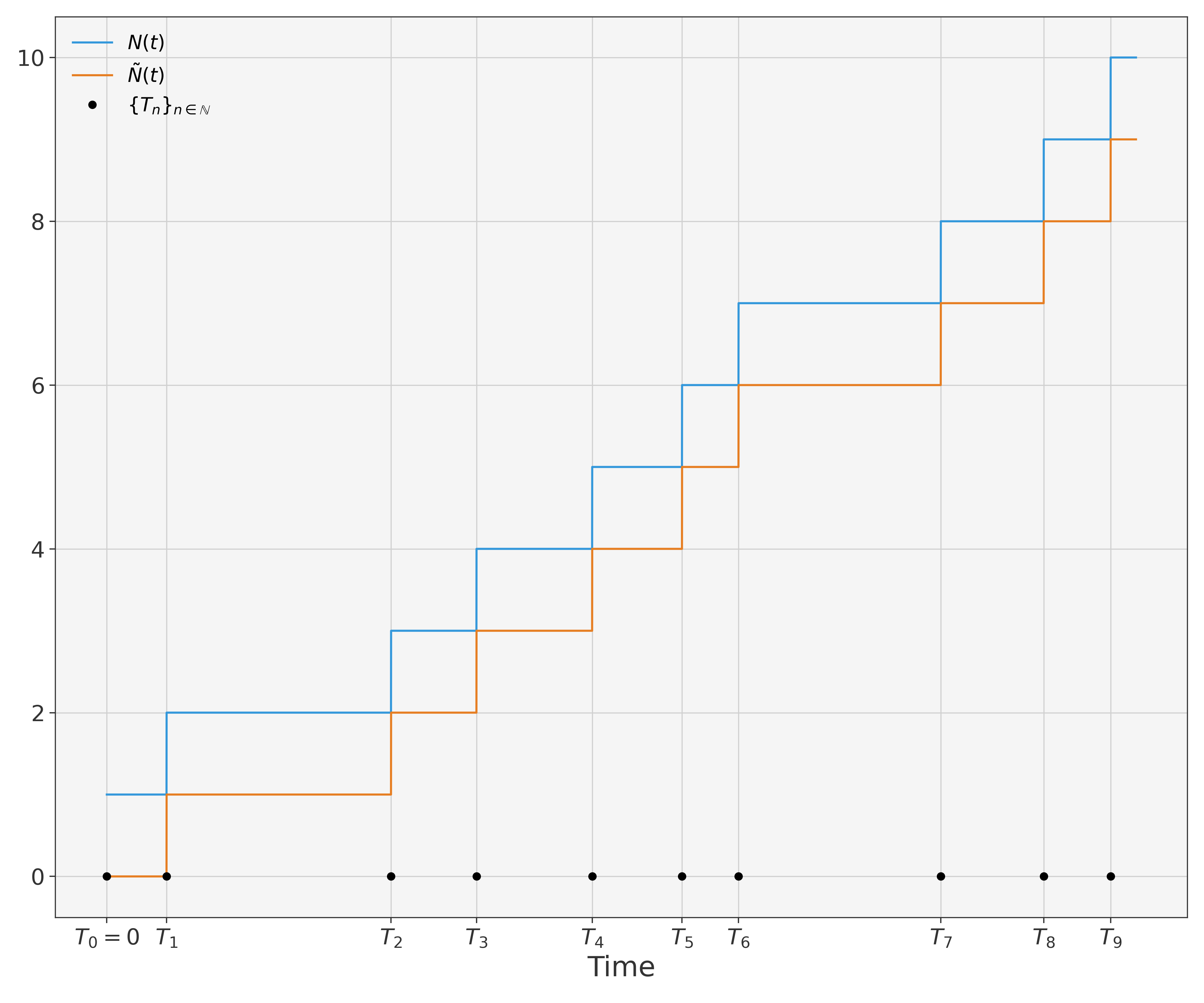}
    \caption{Realization of a renewal- and renewal counting process with $Gamma(2,1)$ interarrival distribution}
    \label{fig:NtANDTn}
\end{figure}

The following remark helps to better comprehend the roles of the various processes \cite{Grimmett:RandomProcesses}.
\begin{remark} 
\label{rmk:lifecycle}
    It is clear from the definition that, for any $t$, $t \in [T_{N(t)-1}, T_{N(t)})$. We can think of this interval as the lifetime of an item. Moreover, 
    \begin{itemize}
        \item The renewal process $\{T_n\}$ takes as values the times at which each item starts its life-cycle;
        \item The renewal counting process $\{N(t)\}$ counts the number of items that have started (and possibly concluded) their life-cycle.
    \end{itemize}
\end{remark}

Any renewal process, as is easily shown from the definitions, is tightly linked to the related renewal counting process. Indeed, we can see that \cite{Asmussen:Applied}
\begin{equation} 
\label{eq:NtTn1}
\{N(t) \leq n\} = \{T_n > t\},
\end{equation}
and together with Remark \ref{rmk:lifecycle} we obtain
\begin{equation} 
\label{eq:NtTn2}
\{ N(t)=n \}=\{T_{n-1}\leq t < T_n\}.
\end{equation}

At this point, we make a simple example:
\begin{example}
\label{ex:PP1}
    The most common type of renewal process/renewal counting process is obtained by defining a zero-delayed process with interarrival distribution exponential of parameter $\lambda$, i.e $F \sim Exp(\lambda)$, with density $f(x) = \lambda e^{-\lambda x}, \: x \in [0,\infty)$. Such a process is nothing more than a modified version of a Poisson Process with rate $\lambda$ ($PP(\lambda)$). The difference arises from the fact that, as explained before, $N(0) = 1$ in our case, while for the common $PP(\lambda)$, $N(0)=0$. The usual definition corresponds to choosing $\tilde{N}(t)$ instead of $N(t)$ as counting process.

    We can show this:
    \[
    \begin{split}
        \mathbb{P}(N(t)>n) = \mathbb{P}(T_n \leq t) &\stackrel{(*)}{=} \int_0^t \lambda e^{-\lambda x} \frac{(\lambda x)^{n-1}}{(n-1)!} \text{d}x\\
        &\stackrel{(**)}{=} 1- \sum_{i=0}^{n-1}\frac{(\lambda t)^i e^{-\lambda t}}{i!}, 
    \end{split}
    \]
    where $(*)$ will be shown in Example \ref{ex:PP2}, and $(**)$ follows by recursively integrating by parts.
    Now, since
    \[
    \mathbb{P}(N(t)>n) = 1- \sum_{i=0}^{n-1}\frac{(\lambda t)^i e^{-\lambda t}}{i!} \Leftrightarrow \mathbb{P}(N(t) \leq n) = \sum_{i=0}^{n-1}\frac{(\lambda t)^i e^{-\lambda t}}{i!},
    \]
    we can conclude that 
    \[
    \mathbb{P}(N(t) = n) = \frac{(\lambda t)^{n-1} e^{-\lambda t}}{(n-1)!}, \quad n=1,2,3,\dots,
    \]
    which is exactly the desired "shifted" Poisson distribution.
\end{example}

Equations \eqref{eq:NtTn1} and \eqref{eq:NtTn2} suggest that we can convert classical results for $\{T_n\}_{n\in\mathbb{N}}$, i.e. a sum of i.i.d. random variables, to $\{N(t)\}_{t\geq 0}$. The first one is a Law of Large Numbers (LLN)-type result \cite{Asmussen:Applied,Grimmett:RandomProcesses}:

\begin{theorem}[LLN for renewal counting process] 
\label{thm:LLNCountingProcess}
    Let $\mu = \mathbb{E}[X_1] = \int_0^\infty x F(\text{\emph{d}}x)$ be the mean of the interarrival distribution. Then, irrespective of the distribution of $X_0$ or whether $\mu < \infty$ or $\mu = \infty$ (but assuming that $1/\infty=0$),
    \begin{equation}
        \frac{N(t)}{t} \longrightarrow \frac{1}{\mu}, \quad t\longrightarrow\infty \ \text{\emph{ almost surely}}.
    \end{equation}
\end{theorem}

the counting process also satisfies a Central Limit Theorem (CLT)-type result for the counting process \cite{Asmussen:Applied,Grimmett:RandomProcesses}:
\begin{theorem} 
\label{thm:RenCLT}
If $\sigma = \mathbb{V}ar(X_1)$ satisfies $0 < \sigma < \infty$, we have
\begin{equation}
    \frac{N(t)-t/\mu}{\sqrt{t\sigma^2/\mu^3}} \stackrel{\mathscr{D}}{\longrightarrow} \mathcal{N}(0,1), \quad t\longrightarrow\infty.
\end{equation}
\end{theorem}

The proofs of these two theorems are quite straightforward, as they rely on the LLN and CLT for the renewal process $\{T_n\}_{n\in\mathbb{N}}$, and can be seen, e.g. in \cite{Asmussen:Applied}.

A further result, which is an important theorem of renewal theory, is called \emph{Elementary Renewal Theorem}, and serves as the starting point for developing a complete theory of renewal processes.

\begin{theorem}[Elementary Renewal Theorem]
    In the same setting of Theorem \ref{thm:LLNCountingProcess}, we have
    \begin{equation}
        \frac{\mathbb{E}[N(t)]}{t} \longrightarrow \frac{1}{\mu}, \quad t\longrightarrow\infty.
    \end{equation}
\end{theorem}

Its proof is more convoluted than the previous ones, for references one can see \cite{Asmussen:Applied,Grimmett:RandomProcesses}.

To conclude the first section, we present the definitions of three random processes of interest that are associated with a renewal process \cite{Asmussen:Applied,Grimmett:RandomProcesses}.

\begin{definition} 
    For a given renewal process $\{ T_n\}_{n \in \mathbb{N}}$, we define
    \begin{itemize}
        \item The \emph{Backward Recurrence Time Process} $\{E(t)\}_{t\geq 0}$, which represents the \emph{Elapsed time} of the current item: $E(t) = t - T_{N(t)-1}$;
        \item The \emph{Forward Recurrence Time Process} $\{R(t)\}_{t\geq 0}$, which represents the \emph{Residual lifetime} of the current item: $R(t) = T_{N(t)} - t$;
        \item The \emph{Total lifetime} of the current item, $\{C(t)\}_{t\geq 0}$: $C(t) = E(t) + R(t) = T_{N(t)} - T_{N(t)-1}$. 
    \end{itemize}

    Note that $R(0) = X_0$ when $X_0>0$ and $R(0) = X_1$ when $X_0=0$.
\end{definition}

Figure \ref{fig:EtANDRt} shows a realization of the two processes $\{E(t)\}_{t\geq0}$ and $\{R(t)\}_{t\geq0}$ for a $Gamma(2,1)$ interarrival distribution.
\begin{figure}[htpb]
    \centering
    \includegraphics[width=0.55\linewidth]{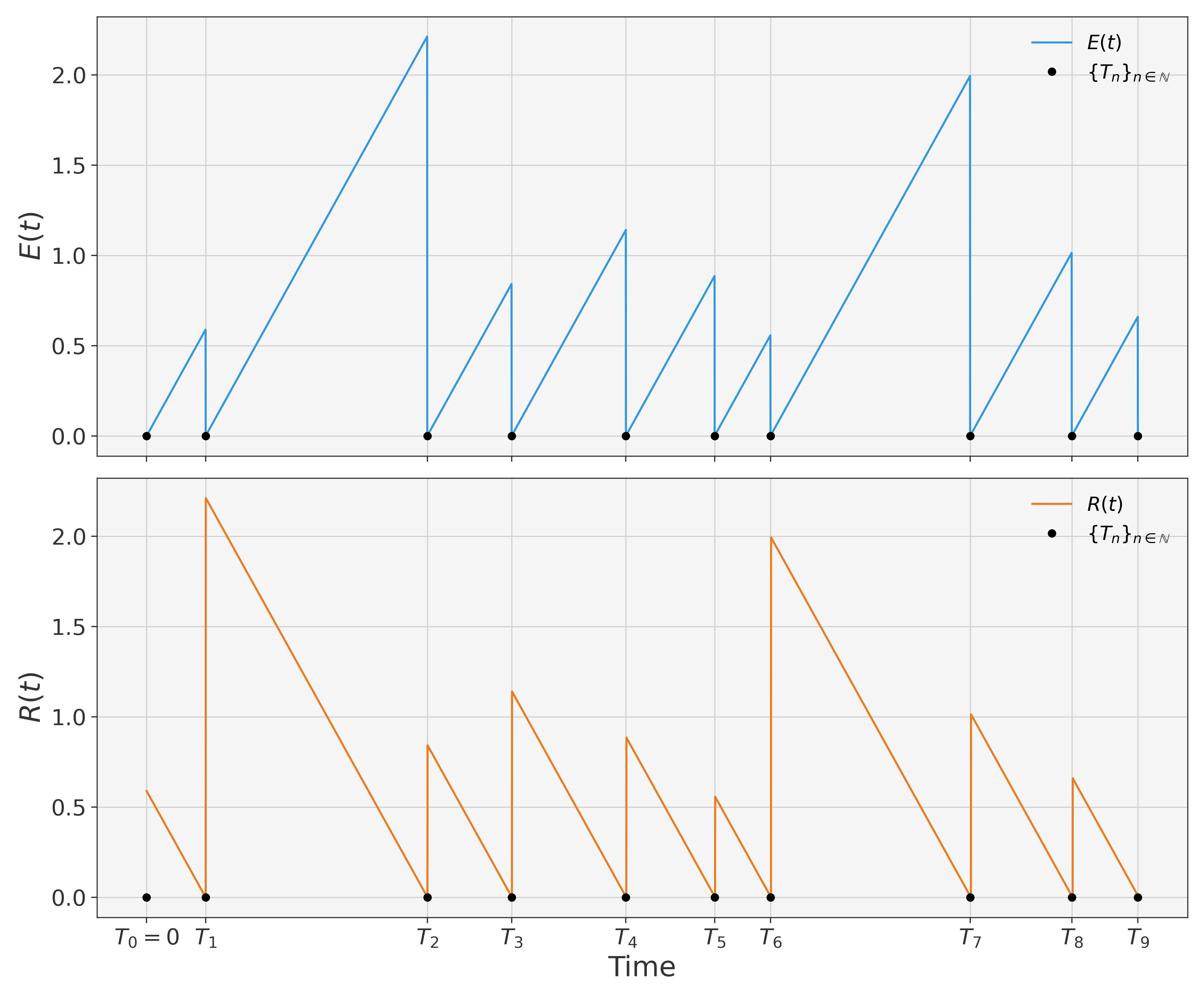}
    \caption{Realization of the forward and backward recurrence time processes with $Gamma(2,1)$ interarrival distribution}
    \label{fig:EtANDRt}
\end{figure}

\begin{remark} 
    The definition of $E(t)$ only makes sense when $t \geq X_0$. However, for any $a\geq 0$ such that $\bar{F}(a) \colon \!\!\! = 1 - F(a) = \mathbb{P}(X_1>a) > 0$, we can assign to $X_0$ the conditional distribution of $X_1$ given that $X_1>a$, i.e.
    \[
    \mathbb{P}(X_0>a) = \mathbb{P}(X_1>x+a|X_1>a)=\frac{\bar{F}(x+a)}{\bar{F}(a)}.
    \]
    Then, by letting $E(t)=t+a$ for every $t<X_0$ we have a version of $\{E(t)\}_{t\geq 0}$ that is well-defined for all $t$, and has $E(0)=a$. In other words, we have defined a renewal process by "starting with a renewal at $-a$" \cite{Asmussen:Applied}.
\end{remark}

\section{Renewal Equations, Renewal Measure and Renewal Function}
A renewal equation is a specific type of convolution equation that has a special role in renewal theory. As we will see in the following pages, many quantities of interest (QoIs) can be proven to be a solution of this type of equation. This enables the study of the QoIs by using the results concerning the solutions of renewal equations.

\begin{definition}[Renewal equation] 
    A \emph{Renewal Equation} is a convolution equation of the type $Z = z + F \ast Z$, i.e.
    \begin{equation}
    \label{Eq:RenEq}
    Z(t) = z(t) + \int_0^t Z(t-u)F(\text{d}u), \quad t\geq 0.
    \end{equation}
    Here $Z(t)$ is an unknown function on $[0,\infty)$, $z(t)$ is a known function on $[0,\infty)$, and $F$ is a known Radon measure on $[0,\infty)$.

    If $F$ is a probability, i.e. $\|F\| \colon\!\!\!= \lim_{t\rightarrow \infty} F(t) = 1$, the renewal equation \eqref{Eq:RenEq} is \emph{proper}.
\end{definition}

Let us see a first example, also shown in \cite{Asmussen:Applied}.
\begin{example} 
\label{ex:TimeProcessesRenEq}
    Consider a zero-delayed renewal process with interarrival distribution $F$ and the related processes $E(t),R(t)$. Let $\xi \geq 0$ be fixed and define $Z_E(t)= \mathbb{P}(E(t) \leq \xi)$, $Z_R(t) = \mathbb{P}(R(t) \leq \xi)$. Then $Z_E, Z_R$ satisfy the renewal equations
    \begin{align}
        Z_E &= z_E + F\ast Z_E, z_E = \mathbb{P}(E(t) \leq \xi, X_1 > t) \stackrel{\text{\emph{(*)}}}{=} \mathbb{I}(t \leq \xi)\bar{F}(t),\\
        Z_R &= z_R + F\ast Z_R, z_R = \mathbb{P}(R(t) \leq \xi, X_1 > t) \stackrel{\text{\emph{(**)}}}{=} F(t+\xi) - F(t).
    \end{align}
\end{example}
\begin{proof}
    Let us start with the process $E(t)$. The proof is carried by the so-called \emph{renewal argument}, which consists in conditioning on the value of $X_1$:
    \[
    \begin{split}
        Z_E(t) = \mathbb{P}(E(t) \leq \xi) &= \mathbb{P}(E(t) \leq \xi, X_1 > t) + \mathbb{P}(E(t) \leq \xi, X_1 \leq t)\\
        &= \mathbb{P}(E(t) \leq \xi, X_1 > t) + \int_0^t\mathbb{P}(E(t) \leq \xi | X_1 = u)F(\text{d}u)\\
        &= z_E(t) + \int_0^t\mathbb{P}(E(t-u) \leq \xi)F(\text{d}u),
    \end{split}
    \]
    because the process starts from scratch at $X_1$, hence $\mathbb{P}(E(t) \leq \xi | X_1 = u) = \mathbb{P}(E(t-u) \leq \xi)$ for $u \leq t$.
    An equivalent argument can be carried out for the process $R(t)$.
    
    Now we only need to show the equalities marked with $(*),(**)$. For $(*)$, notice that, when $X_1 > t$, for a zero-delayed renewal process we have $E(t) = t$ (i.e. $T_{N(t)-1} = T_0 = 0$). Moreover, $\mathbb{P}(X_1>t) = 1 - F(t) = \bar{F}(t)$. Hence,
    \[
    \mathbb{P}(E(t) \leq \xi, X_1 > t) = \mathbb{P}(t\leq\xi, X_1>t) = \mathbb{I}(t \leq \xi)\bar{F}(t).
    \]
    For $(**)$, we can work similarly. When $X_1 >t$, $R(t) = X_1 - t$. Hence,
    \[
    \begin{split}
        \mathbb{P}(R(t) \leq \xi ,X_1 > t) &= \mathbb{P}(X_1 - t \leq \xi ,X_1 >t)\\
        &= \mathbb{P}(X_1 \leq t+\xi,X_1 >t)\\
        &= \mathbb{P}(t < X_1 \leq t+\xi)\\
        &= F(t+\xi) - F(t),
    \end{split}
    \]
    which concludes the proof.
\end{proof}

Such example is crucial for our study, not only because it shows the useful calculations for the renewal equations of the laws of $E(t)$ and $R(t)$, but because it introduces the \emph{renewal argument}. This technique is widely used in renewal theory to compute renewal equations for specific QoIs, such as the law of a given process. 

Two definitions that are closely related to renewal equations are those of \emph{renewal measure} and \emph{renewal function} \cite{Asmussen:Applied}.
\begin{definition}[Renewal Measure and Renewal Function] 
    For a given interarrival distribution $F$, we define the \emph{renewal measure} (on $[0,\infty)$) by $U(\text{d}x) = \sum_{n=0}^\infty F^{*n}(\text{d}x)$, where $F^{*n}$ represents the $n$-th convolution of $F$ with itself (also called $n$-th convolution power of $F$). As a convention, we set $F^{*0}(\text{d}x)= \delta_0(\text{d}x)$ to be the Dirac measure concentrated at $0$ (i.e. $U$ has an atom of mass $1$ at $0$).

    Moreover, we define the \emph{renewal function} as $U(t) = U([0,t]) = 1 + \sum_{n=1}^\infty F^{*n}(t)$, with a slight abuse of notation.

    For a pure renewal process, $F^{*n}$ is the distribution of $T_n = X_1 + \dots + X_n$.
\end{definition}

\begin{definition}[Renewal density] 
    If $U$ is absolutely continuous on $(0,\infty)$ with respect to the Lebesgue measure, we call the density $u(x) = \text{\emph{d}}U/\text{\emph{d}}x$ \emph{renewal density}.
\end{definition}

From the definition of the renewal density, $U$ must be absolutely continuous on $(0,\infty)$. Since $U(\text{d}x)=\delta_0(\text{d}x) + \sum_{n=1}^\infty F^{*n}(\text{d}x)$, this is equivalent to $\sum_{n=1}^\infty F^{*n}(\text{d}x)$ being absolutely continuous on $(0,\infty)$ (or $[0,\infty)$). Trivially, we cannot expect $U$ to be absolutely continuous on $[0,\infty)$ because of the atom at $0$.

The renewal density is easily characterized by the following proposition \cite{Asmussen:Applied}.
\begin{proposition} 
\label{prop:RenDensityCont}
    The renewal density $u$ exists if and only if $F$ has a density $f$. Then $u=\sum_{n=1}^\infty f^{*n}$ or, equivalently, $u$ is the solution of the renewal equation $u = f + F\ast u$.

    If the renewal density exists, we can write $U(\text{\emph{d}}x) = \delta_0(\text{\emph{d}}x) + u(x)\text{\emph{d}}x$.
\end{proposition}

We can now provide a practical example of these quantities.
\begin{example}
\label{ex:PP2}
    Let us consider the situation of Example \ref{ex:PP1}, i.e. $F$ is absolutely continuous with density $f(x) = \lambda e^{-\lambda x}, \: x \in [0,\infty)$. By the definition of renewal measure and Proposition \ref{prop:RenDensityCont}, $U(\text{d}x) = \delta_0(\text{d}x) + u(x)\text{d}x$, where $u(x) = \sum_{n=1}^\infty f^{*n}(x)$.

    First of all, it is possible to show that $f^{*n}(x) = \lambda e^{-\lambda x} \frac{(\lambda x)^{n-1}}{(n-1)!}$ for all $n\geq1$. Indeed, $f^{*1} = f = \lambda e^{-\lambda x}$.
    Additionally, assume $f^{*n}(x) = \lambda e^{-\lambda x} \frac{(\lambda x)^{n-1}}{(n-1)!}$. Then,
    \[
    \begin{split}
        f^{*(n+1)}(x) = f^{*n}\ast f(x) &= \int_0^x \lambda e^{-\lambda (x-y)} \lambda e^{-\lambda y} \frac{(\lambda y)^{n-1}}{(n-1)!} \text{d}y\\
        &=\lambda e^{-\lambda x} \frac{\lambda^n}{(n-1)!} \int_0^x y^{n-1}\text{d}y\\
        &= \lambda e^{-\lambda x} \frac{(\lambda x)^n}{n!}.
    \end{split}
    \]
    Now we can compute $u(x)$:
    \[
    u(x) = \sum_{n=1}^\infty f^{*n}(x) = \lambda e^{-\lambda x} \sum_{n=1}^\infty \frac{(\lambda x)^{n-1}}{(n-1)!} = \lambda e^{-\lambda x} \sum_{n=0}^\infty \frac{(\lambda x)^{n}}{n!} = \lambda e^{-\lambda x} e^{\lambda x} = \lambda.
    \]

    Hence, the renewal measure is $U(\text{d}x) =\delta_0(\text{d}x) + \lambda \text{d}x$, and the renewal function is $U(t) = U([0,t]) = 1+\lambda t$.
\end{example}

As was previously mentioned, the renewal function and renewal measure are tightly linked to renewal equations \cite[Theorem 2.4, Chap. V]{Asmussen:Applied}:
\begin{theorem} 
\label{thm:RenFunction}
    \begin{itemize}
        \item[\emph{(i)}] The renewal function $U(t)$ is finite for all $t <\infty$.
        \item[\emph{(ii)}] If the function $z(t)$ in the renewal equation \eqref{Eq:RenEq} is bounded on finite intervals, then $Z = U \ast z$ is well defined, a solution to \eqref{Eq:RenEq}, and the unique solution to \eqref{Eq:RenEq} that is bounded on finite intervals.
        \item[\emph{(iii)}] If $\|F\|=1$, then $U(t) = \mathbb{E}[N(t)]$, for a zero-delayed renewal process with interarrival distribution $F$.

        More generally, for any renewal process with interarrival distribution $F$, the expected number of renewals in $(t, t+a]$ is 
        \begin{equation}
        \label{eq:expNumbRen}
        \mathbb{E}[N(t+a)-N(t)] = \int_0^a U(a-\xi)G_t(\text{\emph{d}}\xi) = G_t \ast U (a) = U \ast G_t(a),
        \end{equation}
        where $G_t(a) = \mathbb{P}(R(t) \leq a)$. Further, the expression \eqref{eq:expNumbRen} cannot exceed $U(a)$.
    \end{itemize}
\end{theorem}
Here, we do not provide the proof of this result, but one can look at \cite{Asmussen:Applied} for details.

One also easily sees that $U(t)$ itself is a solution of a renewal equation \cite{Grimmett:RandomProcesses}:
\begin{lemma} 
\label{lemma:RenFunctRenEq}
    For a zero-delayed renewal process with interarrival distribution $F$, such that $\|F\|=1$, we have that $U(t)=\mathbb{E}[N(t)]$ is a solution of the renewal equation
    \begin{equation}
    \label{eq:renEqU}
    U(t) = 1 + \int_0^t U(t-x)F(\text{\emph{d}}x).
    \end{equation}
\end{lemma}
\begin{proof}
    By using the Tower Law we have
    \[
    U(t) = \mathbb{E}[N(t)] = \mathbb{E}\left[\mathbb{E}[N(t)|X_1]\right] = \int_0^\infty \mathbb{E}[N(t)|X_1=x]F(\text{d}x).
    \]
    Alternatively, one can directly apply the renewal argument and conclude the same.
    Now, for $t < x$, $N(t) \equiv 1$ because only one renewal happens (at 0).
    On the other hand, for $t \geq x$,
    \[
    \mathbb{E}[N(t)|X_1=x] = 1 + \mathbb{E}[N(t-x)],
    \]
    because the process restarts from scratch (identically distributed and independent of the past) after the first renewal.
    Hence we have
    \[
    \begin{split}
       \mathbb{E}[N(t)] &= \int_0^\infty \mathbb{E}[N(t)|X_1=x]F(\text{d}x)\\
       &= \int_t^\infty 1 F(\text{d}x) + \int_0^t1 F(\text{d}x) + \int_0^t \mathbb{E}[N(t-x)] F(\text{d}x)\\
       &= 1 + \int_0^t \mathbb{E}[N(t-x)] F(\text{d}x),
    \end{split}
    \]
    which concludes the proof.
\end{proof}

\begin{remark}
    Lemma \ref{lemma:RenFunctRenEq} could have been deducted directly from Theorem \ref{thm:RenFunction} and, in addition, the Theorem gives us more information on $U(t)$ as a solution of such renewal equation.
    
    Since $U(t)$ is finite for all $t<\infty$, it is also bounded on finite intervals. Furthermore, $U = U \ast 1$, which means that $U$ is a solution of the renewal equation
    \[
    U = 1 + F \ast U.
    \]
    In addition, it is the \emph{unique} solution that is bounded on finite intervals.
\end{remark}

At this point, we must observe more in detail the two different definitions of the renewal counting process. For a reader not to feel confused, we think it is best to summarize some dissimilarities between the two approaches, and how to move from one to the other easily.

By looking at the two different definitions of the renewal process (i.e. the definitions of $\tilde{N}(t)$ and $N(t)$), we can deduce that
\[
U(t) = \mathbb{E}[N(t)] = \mathbb{E}[\tilde{N}(t) + 1] = \tilde{U}(t) + 1.
\]
By an inspection of this relation, we notice that the two definitions yield distinct (but consistent) results.
    
Indeed, we can find two immediate differences between the two approaches. 
First, if we use the definition of the process $\{\tilde{N}(t)\}_{t \geq 0}$, we obtain a different renewal equation, compared to \eqref{eq:renEqU} \cite{Grimmett:RandomProcesses}:
\[
    \mathbb{E}[\tilde{N}(t)] = F(t) + \int_0^t \mathbb{E}[\tilde{N}(t-x)] F(\text{d}x).
\]
Second, by Theorem \ref{thm:RenFunction}, we know that, under suitable conditions, the solution $Z$ of the general renewal equation \eqref{Eq:RenEq} is $Z = U\ast z$. However, this is also equal to $Z = U\ast z = (\tilde{U} + 1) \ast z = z + \tilde{U} \ast z$ \cite{Grimmett:RandomProcesses}.

We conclude this section with an example.
\begin{example}
\label{ex:PP3}
    Following the setting of Examples \ref{ex:PP1} and \ref{ex:PP2}, we compute the distribution of the backward and forward recurrence time processes, using renewal theory arguments.

    We already showed in Example \ref{ex:TimeProcessesRenEq} that $Z_E(t) = \mathbb{P}(E(t)\leq x)$ and $Z_R(t) = \mathbb{P}(R(t) \leq x)$ (for $x \in [0,\infty)$) satisfy two renewal equations with $z_E^x(t) = \mathbb{I}(t\leq x)\bar{F}(t)$ and $z_R^x(t) = F(t+x)-F(t)$, respectively. Since both $z_E^x$ and $z_R^x$ are bounded (hence bounded on finite intervals), by Theorem \ref{thm:RenFunction} we know that $Z_E(t) = U \ast z_E^x(t)$ and $Z_R(t) = U\ast z_R^x(t)$ are the unique solutions to the related renewal equations that are bounded on finite intervals. Let us proceed with the calculation for the case $F\sim Exp(\lambda)$.

    We start with $Z_E(t)$:
    \[
    \begin{split}
        Z_E(t) = U\ast z_E^x(t) &= \int_0^t z_E^x(t-y)U(\text{d}y)\\
        &= \int_0^t \mathbb{I}(t-y \leq x)\bar{F}(t-y) U(\text{d}y)\\
        &= \int_0^t \mathbb{I}(t-y \leq x)\bar{F}(t-y) \delta_0(\text{d}y) + \int_0^t \mathbb{I}(t-y \leq x)\bar{F}(t-y) \lambda\text{d}y\\
        &= \mathbb{I}(t\leq x)\bar{F}(t) + \int_0^t \mathbb{I}(s \leq x)\bar{F}(s)\lambda \text{d}s\\
        &= \mathbb{I}(t\leq x)\bar{F}(t) + \int_0^{t\wedge x} \bar{F}(s)\lambda \text{d}s\\
        &= \mathbb{I}(t\leq x)\bar{F}(t) + \int_0^{t\wedge x} \lambda e^{-\lambda s} \text{d}s\\
        &= \mathbb{I}(t\leq x)(e^{-\lambda t}) + 1- e^{-\lambda(t \wedge x)}\\
        &= \begin{cases}
            1 &\quad \text{if }t\leq x\\
            1-e^{-\lambda x} &\quad \text{if } t>x
        \end{cases}
    \end{split}
    \]
    However, this distribution also has an atom at $t=x$. As we know from the definition of the process $\{E(t)\}_{t\geq0}$:
    \[
    \mathbb{P}(E(t) = t) = \mathbb{P}(X_1>t) = \bar{F}(t) = e^{-\lambda t},
    \]
    Hence, the atom has weight $e^{-\lambda t}$. As $t\rightarrow\infty$, the atom disappears.
    
    We now continue with $Z_R(t)$:
    \[
    \begin{split}
        Z_R(t) = U \ast z_R^x(t) &= \int_0^t z_R^x(t-y)U(\text{d}y)\\
        &= \int_0^t (F(t-y+x)-F(t-y))U(\text{d}y)\\
        &= \int_0^t (F(t-y+x)-F(t-y))\delta_0(\text{d}y)\\ &\qquad+ \int_0^t (F(t-y+x)-F(t-y))\lambda\text{d}y\\
        &= F(t+x)-F(t) + \int_0^t (F(s+x)-F(s))\lambda\text{d}s\\
        &= 1-e^{-\lambda(t+x)}- 1 + e^{-\lambda t} + \int_0^t \lambda(1-e^{-\lambda(s+x)}-1+e^{-\lambda s})\text{d}s\\
        &= e^{-\lambda t}(1-e^{-\lambda x})+ (1-e^{-\lambda x})\int_0^t \lambda e^{-\lambda s} \text{d}s\\
        &= e^{-\lambda t}(1-e^{-\lambda x})+ (1-e^{-\lambda t})(1-e^{-\lambda x})\\
        &= 1-e^{-\lambda x}.
    \end{split}
    \]
    We recognize this function to be the cdf of an $Exp(\lambda)$ distribution, which does not depend on $t$.
\end{example}

\section{Stationarity}
In this section we analyze the conditions under which a renewal process is stationary, and the form of its stationary distribution.

We define a renewal process to be stationary if \cite{Asmussen:Applied}
\[
    \{N(t+s)-N(t)\}_{s\geq 0} \stackrel{\mathscr{D}}{=} \{N(s)-N(0)\}_{s\geq 0}.
\]
In other words, for any $t>0$, if we shift the origin to $t$, we leave the distributions of the epochs unchanged, or equivalently, the renewal counting process has stationary increments.

\begin{remark}
\label{rmk:StationarityRtEquivStationarityTn}
    An equivalent condition to the renewal process being stationary, is that the related forward time recurrence process $\{R(t)\}_{t\geq 0}$ is stationary, meaning that the distribution of $R(t)$ does not depend on $t$ \cite{Asmussen:Applied}.
\end{remark}

The stationary distribution of the forward recurrence time process has a precise form \cite{Asmussen:Applied,Grimmett:RandomProcesses}:
\begin{lemma} 
\label{lemma:stationaryDist}
    The density of the stationary distribution $F_0$ for $\{R(t)\}_{t\geq 0}$ and $\{E(t)\}_{t\geq 0}$ is given by
    \begin{equation}
    f_0(x) = \frac{\bar{F}(x)}{\mu},
    \end{equation}
    where $\mu = \int_0^\infty xF(\text{\emph{d}}x)$.
\end{lemma}

\begin{proof}[Sketch of Proof]
    The form of the stationary distribution for $\{R(t)\}_{t\geq 0}$ can be found with a level-crossing argument, which is not reported here. See \cite{Asmussen:Applied} for more details.

    Now,
    \begin{equation}
    \{R(t) \leq \xi \} = \{\text{renewal in } (t,t+\xi]\} = \{E(t+\xi) < \xi\},
    \end{equation}
    and we can conclude that the stationary distribution of $\{E(t)\}_{t\geq 0}$ has to be same as the one of $\{R(t)\}_{t\geq 0}$.
\end{proof}

This result also shows that, for continuous time renewal processes, $R(t)$ and $E(t)$ share the same stationary distribution.

Let us continue with the usual Poisson Process example and compute the stationary distribution of the two processes.
\begin{example}
\label{ex:PP4}
In the same setting as Examples \ref{ex:PP1}, \ref{ex:PP2} and \ref{ex:PP3}, we can explicitly compute the form of the stationary distribution of the two processes $\{E(t)\}_{t\geq 0}$ and $\{R(t)\}_{t \geq 0}$. Indeed, its density is
\[
    f_0(x) = \frac{\bar{F}(x)}{\mu} = \lambda(1-F(x)) = \lambda e^{-\lambda x}, \quad x \in [0,\infty).
\]
Hence, the stationary distribution of the two processes is again $Exp(\lambda)$.

By checking the distributions computed in Example \ref{ex:PP3}, it is clear that the process $\{R(t)\}_{t \geq 0}$ is stationary (which implies stationarity of the renewal process itself), while it is not the case for the process $\{E(t)\}_{t \geq 0}$. However, from the form of the distribution of $E(t)$, it is obvious that, as $t\rightarrow\infty$, it converges to the stationary $Exp(\lambda)$ distribution.
\end{example}

Now, let us present a series of results that provide interesting insights on stationarity, taken from \cite[Chap. V.3]{Asmussen:Applied}:
\begin{lemma} 
\label{lemma:stationarity}
    \begin{itemize}
        \item[\emph{(i)}] If $\{E(t)\}$ is stationary, then so is $\{R(t)\}$.
        \item[\emph{(ii)}] If $E(t)$ has distribution $F_0$, then so has $R(t)$, and $\mathbb{P}(E(t)>x, R(t)>y) = \bar{F_0}(x+y)$
    \end{itemize}
\end{lemma}

\begin{lemma} 
    Let $F_1$ be the distribution that has density $x/\mu$ with respect to $F$, let $C$ be a random variable with distribution $F_1$, and let $U$ be a $\mathcal{U}(0,1)$ random variable independent of $C$. Define $E=CU$ and $R=C(1-U)$. Then $\mathbb{P}(E>x,R>y) = \bar{F_0}(x+y).$
\end{lemma}

\begin{theorem} 
    Let $C$ be a random variable with distribution $F_1$, and let $U\sim \mathcal{U}(0,1)$ independently of $C$. Then the version of the Markov process $\{(E(t), R(t), C(t))\}$ obtained from the initial values $E(0) = CU, R(0) = C(1-U), C(0) = C$ is strictly stationary. Moreover, the point process whose set of renewals is 
    \[
    \{t\geq 0 \colon E(t-) \neq E(t)\} = \{t\geq 0 \colon R(t-) \neq R(t)\}
    \]
    is a stationary renewal process with interarrival distribution $F$.
\end{theorem}

The three statements imply that the stationary distribution for $\{E(t)\}, \{ R(t) \}$ is $F_0$, while the stationary distribution for $\{ C(t) \}$ is $F_1$, where $F_1(\text{d}x) = \frac{x}{\mu}F(\text{d}x)$. Moreover, the last lemma tells us that we can reconstruct a stationary renewal process from a stationary three-dimensional process $\{(E(t),R(t),C(t))\}$.

In addition, as stated in Lemma \ref{lemma:stationarity} (and also because of Remark \ref{rmk:StationarityRtEquivStationarityTn}), the stationarity of a renewal process directly depends on the stationarity of $\{E(t)\}$.

\begin{remark}
\label{rmk:InspectionParadox}
    The previous results say that the stationary distribution of the current life process is $F_1$, with $F_1(\text{\emph{d}}x)=\frac{x}{\mu}F(\text{\emph{d}}x)$, which is clearly different from $F$. In particular we have that the expected value of $F_1$ is
    \[
    \mu_{F_1} = \int x F_1(\text{d}x) = \int x^2 F(\text{d}x)/\mu,
    \]
    which is greater than or equal to $\mu$ (unless $F$ is the distribution of a constant random variable).

    This is known as \emph{inspection paradox}, and it refers to the fact that when one inspect a renewal process, they are more likely to be in a long interarrival interval than a short one. In other words, it is a \emph{length-biased} observation, because, since longer intervals occupy more time on the real line, it is more likely for one of them to contain a randomly chosen time $t$.
\end{remark}

In the previous discussion, we did not specify whether the stationary distribution $F_0$ is unique. The following proposition shows that it is \cite{Asmussen:Applied}.
\begin{proposition} 
    Let $G$ be a distribution on $(0,\infty)$ such that one of the following is true:
    \begin{itemize}
        \item[\emph{(i)}] $G$ is stationary for $\{E(t)\}$;
        \item[\emph{(ii)}] $G$ is stationary for $\{R(t)\}$;
        \item[\emph{(iii)}] a renewal process with delay distribution $G$ is stationary.
    \end{itemize}
    Then $G=F_0$.
\end{proposition}

As a corollary of this proposition, we get a characterization of stationarity for a \emph{delayed} renewal process \cite{Asmussen:Applied}.
\begin{corollary}
\label{cor:StatDelay}
    A delayed renewal process is stationary if and only if the distribution of the initial delay $R(0)=X_0=T_0$ is $F_0$.
\end{corollary}

\section{The Renewal Theorem}
The renewal theorem is one of the most important results, not only of renewal theory, but of all probability theory in general, due to its applicability in numerous other areas. It has several equivalent version, that we will state in this section. We do not provide a proof of the equivalence, nor of the theorem itself, as it is out of the scope of this thesis. For references one can see \cite{Asmussen:Applied,Grimmett:RandomProcesses}.

Before we can actually present the theorem, we need to introduce a new type of integrability. 
\begin{definition}[Direct Riemann integrability] 
    Let $z$ be a nonnegative function on $[0,\infty)$, and let $h>0$. Define
    \[
    \bar{z}_h(x) = \sup_{y \in I_n^h} z(y), \underline{z}_h = \inf_{y \in I_n^h} z(y), x \in I_n^h = (nh, (n+1)h].
    \]
    We call $z$ \emph{directly Riemann integrable} (d.R.i.) if $\int\bar{z}_h = \int_0^\infty \bar{z}_h(x)\text{d}x$ is finite for some $h$, and $\int\bar{z}_h - \int\underline{z}_h \rightarrow 0$ as $h\rightarrow0$.

    For functions with compact support, this is equivalent to Riemann integrability.

    If $z$ attains also negative values, it is d.R.i. if $z^+, z^-$ are so.
\end{definition}

The next proposition gives sufficient and necessary conditions for a function to be directly Riemann integrable \cite{Asmussen:Applied}.
\begin{proposition} 
\label{prop:DRI}
    Suppose $z\geq 0$. Then if $z$ is d.R.i., it is also Lebesgue integrable and $\int \bar{z}_h, \int\underline{z}_h$ have the common limit $\int z$ as $h\rightarrow 0$. A necessary condition for $z$ being d.R.i. is 
    \begin{enumerate}
        \item \label{condition:dRi} $z$ is bounded and continuous almost everywhere w.r.t. Lebesgue measure.
    \end{enumerate}
    Sufficient conditions are:
    \begin{enumerate}
        \setcounter{enumi}{1}
        \item $\int\bar{z}_h < \infty $ for some $h$ and \eqref{condition:dRi} holds;
        \item $z$ has bounded support and \eqref{condition:dRi} holds;
        \item $z \leq z^*$ with $z^*$ d.R.i. and \eqref{condition:dRi} holds for $z$;
        \item $z$ is nonincreasing and Lebesgue integrable.
    \end{enumerate}
\end{proposition}

At this point, we have introduced all the necessary definitions to state the different forms of the renewal theorem. We begin by writing the necessary assumptions for the theorems to hold.
\begin{assumption}
    For all the equivalent versions of the renewal theorem to hold, the following assumptions are required:
    \begin{itemize}
        \item The interarrival distribution $F$ is proper $(\|F\| = 1)$.
        \item The interarrival distribution $F$ is nonlattice, i.e. it is \emph{not} concentrated on a set of the form $\{\delta,2\delta,\dots\}$, for some $\delta\geq 0$.
        \item Write $\mu = \int xF(\text{d}x)$ and $F_0(t) = \mu^{-1}\int_0^t\bar{F}(y)\text{d}y$.
    \end{itemize}
\end{assumption}

It is now possible to state the four different versions of the renewal theorem \cite[Chap. V.4]{Asmussen:Applied}.
\begin{theorem}[Blackwell's Renewal Theorem] 
    \label{thm:BlackwellRen}
    Let $U$ be the renewal function. Then, for all $a$
    \[
    U(t+a)-U(t) \longrightarrow\frac{a}{\mu} \quad ,t\rightarrow\infty.
    \]
    More generally, in any renewal process with interarrival distribution $F$, the expected number of renewals in $(t,t+a]$, tends to $a/\mu$ as $t\rightarrow\infty$.
\end{theorem}
\begin{theorem} 
    \label{thm:Ren2}
    Let $\{E(t)\}_{t\geq 0}$ be the backward recurrence time process for a (delayed or not) renewal process with interarrival distribution $F$. Then $\mathbb{P}(E(t)\leq\xi)\rightarrow F_0(\xi)$ for all $\xi$. In particular, if $\mu<\infty$, $E(t)\stackrel{\mathscr{D}}{\rightarrow}F_0$.
\end{theorem}
\begin{theorem} 
    \label{thm:Ren3}
    Let $\{R(t)\}_{t\geq 0}$ be the forward recurrence time process for a (delayed or not) renewal process with interarrival distribution $F$. Then $\mathbb{P}(R(t)\leq\xi)\rightarrow F_0(\xi)$ for all $\xi$. In particular, if $\mu<\infty$, $R(t)\stackrel{\mathscr{D}}{\rightarrow}F_0$.
\end{theorem}
\begin{theorem}[Key Renewal Theorem] 
    \label{thm:keyRen}
    Suppose that the function $z$ in the renewal equation $Z=z+F\ast Z$ is d.R.i. Then
    \[
    Z(t) = U \ast z\,(t) \longrightarrow \frac{1}{\mu}\int_0^\infty z(x)\text{\emph{d}}x, \quad t\rightarrow\infty.
    \]
\end{theorem}

In the case $\mu < \infty$, Theorems \ref{thm:BlackwellRen},\ref{thm:Ren2} and \ref{thm:Ren3} say that the renewal process is asymptotically stationary as $t\rightarrow\infty$. On the other hand, if $\mu=\infty$, Theorems \ref{thm:Ren2} and \ref{thm:Ren3} state that the mass in the distributions of $E(t)$ and $R(t)$ drifts off to $\infty$. In other words, if the interarrival mean is infinite, the process exhibits a null-recurrence type of behavior.

The Key Renewal Theorem \ref{thm:keyRen}, gives an asymptotic form of the solution of a renewal equation. This result is very useful for studying the convergence of renewal processes, as will be shown in the next Chapters.

The last result of this Chapter is a Renewal Theorem for the renewal density. It can be proved using the Key Renewal Theorem \cite[Exercise 4.2, Chap. V]{Asmussen:Applied}.
\begin{proposition}[Renewal Theorem for renewal density] 
    If $F$ has a d.R.i. density $f$ so that the renewal density $u$ exists, then $u(x)\rightarrow 1/\mu$.
\end{proposition}
\begin{proof}
    We know from Proposition \ref{prop:RenDensityCont} that $F$ having a density $f$ is equivalent to the existence of the renewal density $u$. Moreover, we know that $u$ is the solution to the following renewal equation: $u=f+F\ast u$. Since $f$ is d.R.i., it is also bounded (by Proposition \ref{prop:DRI}), hence the unique solution to the renewal equation above is $u=(U\ast f)$.

    To conclude, we can apply the Key Renewal Theorem (since $f$ is d.R.i.) and show that
    \[
    u(t) = U \ast f(t) \longrightarrow \frac{1}{\mu} \int_0^\infty f(x)\text{d}x = \frac{1}{\mu}.
    \]
\end{proof}

To conclude the chapter, we provide a practical application of the Theorems.
\begin{example}
\label{ex:PP5}
Let us consider the same setting as Examples \ref{ex:PP1}, \ref{ex:PP2}, \ref{ex:PP3}, and \ref{ex:PP4}. Clearly, the interarrival distribution is both proper and nonlattice, hence we expect the renewal theorems to hold.
\begin{itemize}
    \item Blackwell's Renewal Theorem \ref{thm:BlackwellRen}: As shown in Example \ref{ex:PP2}, the renewal function is equal to $U(t) = 1+\lambda t$. Hence $U(t+a)-U(t) = 1+\lambda(t+a)-1-\lambda t = \lambda a = a/\mu$. Hence, trivially, $U(t+a)-U(t) \rightarrow a/\mu$ as $t\rightarrow\infty$.
    \item Theorem \ref{thm:Ren2}: First of all, in this setting $\mu<\infty$. We showed in Examples \ref{ex:PP3} and \ref{ex:PP4} that $E(t)$ converges in distribution to $F_0$.
    \item Theorem \ref{thm:Ren3}: As seen in Examples \ref{ex:PP3} and \ref{ex:PP4}, the process $\{R(t)\}_{t\geq 0}$ is stationary, hence it trivially converges to $F_0$.
    \item Key Renewal Theorem \ref{thm:keyRen}: For this specific setting, we showed that $\mathbb{P}(E(t)\leq x)$ satisfies a renewal equation with $z_E(t) = \mathbb{I}(t\leq x)\bar{F}(t)$ and that $\mathbb{P}(R(t)\leq x)$ satisfies a renewal equation with $z_R(t) = F(t+x) - F(t)$. We also know that the two QoIs converge to $1-e^{-\lambda x}$ as $t\rightarrow\infty$. We can check this by using the Key Renewal Theorem.

    Let us start with $E(t)$. $z_E(t)$ is d.R.i. because it has a bounded support, and it is bounded and continuous. Hence
    \[
    \mathbb{P}(E(t)\leq x) \longrightarrow \frac{1}{\mu}\int_0^\infty \mathbb{I}(t\leq x) \bar{F}(t)\text{d}t = \int_0^x\lambda e^{-\lambda t}\text{d}t = 1-e^{-\lambda x},
    \]
    as expected.

    Now we prove the same for $R(t)$. First of all, $z_R(t) = F(t+x)-F(t)=e^{-\lambda t}(1-e^{-\lambda x})$ is d.R.i. because it is non-increasing and Lebesgue integrable in $t$. Hence,
    \[
    \mathbb{P}(R(t)\leq x) \longrightarrow(1-e^{-\lambda x}) \int_0^\infty \lambda e^{-\lambda t}\text{d}t = 1-e^{-\lambda x},
    \]
    as expected. Note that this last convergence is trivial, since $\mathbb{P}(R(t)\leq x) = 1 - e^{-\lambda x}$ for all $t$.
\end{itemize}
\end{example}

\chapter{Continuous-time Regenerative Processes: Background}
\label{chap:3}
The following Chapter is centered on the theoretical background on \emph{regenerative processes}. The theory develops straightforwardly from that of renewal processes, which implies that we can use the same tools to study both types of processes.

\section{Basic Definitions and Results}
Intuitively, a regenerative process is a stochastic process that can be divided into independent and identically distributed cycles. This is formalized by the following definition \cite{SigmanWolff:RegenerativeProcesses}:
\begin{definition}[Regenerative Process, Classical Definition] 
\label{def:RegProcClassic}
    A stochastic process $Y=\{Y(t)\}_{t\geq0}$ is called \emph{regenerative} if there is a random variable $X_1>0$ such that
    \begin{itemize}
        \item[$(i)$] $\{Y(t+X_1)\}_{t\geq0}$ is independent of the  history of $Y$ up until time $X_1$, and of $X_1$;
        \item[$(ii)$] $\{Y(t+X_1)\}_{t\geq0}$ is stochastically equivalent to $\{Y(t)\}_{t\geq0}$, meaning that for all $t\geq0$,
         $\{Y(t+X_1)\}_{t\geq0}$ and  $\{Y(t)\}_{t\geq0}$ have the same finite-dimensional distributions. 
    \end{itemize}
    We call $X_1$ a \emph{regeneration point} and say that the process \emph{regenerates} or \emph{starts over} at this point.
\end{definition}

The recursive cycle structure is not explicitly stated in this definition, but can be easily deduced from it \cite{SigmanWolff:RegenerativeProcesses}.
\begin{remark} 
    Assumption $(i)$ means that $\{Y(t+X_1)\}_{t\geq0}$ is independent of $X_1$ and of the past history of $\{Y(t)\}$ prior to $X_1$.
    Assumptions $(i)$ and $(ii)$ together show that $\{Y(t+X_1)\}$ is regenerative (because $\{Y(t)\}$ is) with regeneration epoch $X_2$, which is independent and identically distributed to $X_1$ (notice that $X_1 +X_2$ is a regeneration point for $Y$).

    Proceeding in this way, we can obtain a sequence of random variables $\{X_n\}$ which are independent and identically distributed.
    We can use such sequence to divide the process into \emph{cycles} or \emph{tours} $\{Y(t)\}_{0\leq t <X_1}$, $\{Y(t)\}_{X_1\leq t < X_1+X_2},\dots$, that are independent and identically distributed stochastic processes.
\end{remark}

As we already anticipated in the previous chapter, we can embed a renewal process into each regenerative process \cite{SigmanWolff:RegenerativeProcesses}.
\begin{definition} 
    The random variables $\{X_n\}$ are called \emph{cycle lengths} and the random process $T_n\colon \!\!\! = X_1+\dots+X_n$ is a renewal process and is called \emph{embedded renewal process} for $Y$.
\end{definition}

Sometimes, it may happen that a process has property $(i)$ and then is regenerative from $X_1$ on. This is equivalent to the first cycle having a different distribution than the rest. This type of regenerative processes are called, like it happens for renewal processes, \emph{delayed regenerative processes}. From a time-average point of view, the first cycle-length is not important, provided it is a proper random variable \cite{SigmanWolff:RegenerativeProcesses}.

Definition \ref{def:RegProcClassic} can effectively describe the majority of regenerative processes, but in some cases we need to accommodate for some kind of dependence between cycles. This yields the following, more general, definition \cite{Asmussen:Applied}.
\begin{definition}[Regenerative Process 2.0] 
    Assume that the stochastic process $\{Y(t)\}_{t\in\mathbb{T}}$ has state space $E$ and continuous or discrete time parameter $t\in \mathbb{T}$, i.e. $\mathbb{T}=[0,\infty)$ or $\mathbb{T}=\mathbb{N}$ respectively. We say that $\{Y(t)\}$ is \emph{regenerative (pure or delayed)} if there exists a renewal process (pure or delayed) $\{T_n\} = \{X_1+\dots+X_n\}$ such that, for each $n\geq 0$, the post-$T_n$ process
    \[
    \theta_{T_n}(Y) = (X_{n+1},X_{n+2},\dots,\{Y(T_n+t)\}_{t\in\mathbb{T}})
    \]
    is independent of $T_0,\dots,T_n$ (or equivalently $X_0,\dots,X_n$), and its distribution does not depend on $n$.
    We call $\{T_n\}$ the \emph{embedded renewal process} and we refer to the $T_n$'s as \emph{regeneration points}.
\end{definition}

A desirable quality of regenerative processes is that the regenerative condition is preserved under measurable mappings \cite{Asmussen:Applied}.
\begin{proposition} 
    If $\{Y(t)\}_{t\in \mathbb{T}}$ is regenerative and $\varphi\colon E \rightarrow F$ is any measurable mapping, then $\{\varphi(Y(t))\}_{t\in\mathbb{T}}$ is regenerative with the same embedded renewal process.
\end{proposition}

In the rest of the thesis, when needed, we consistently use a specific notation:
\begin{notation}
    To any given delayed regenerative process, corresponds a zero-delayed one with a unique probability law (e.g. $\{Y(T_0+t)\}_{t\in\mathbb{T}}$). Let us denote by $\mathbb{P}_0,\mathbb{E}_0$ the law corresponding to the zero-delayed case and write $X=X_1$ for the length of the first cycle and $\mu=\mathbb{E}_0[X]$.
\end{notation}

Regenerative processes are widely used, both in theory and applications, because of their power: the existence of a limiting distribution is guaranteed by mild conditions that are usually easy to verify. For continuous-time regenerative processes, it suffices that the cycle length distribution is nonlattice and has finite mean $\mu<\infty$, and that the sample paths satisfy some regularity condition \cite{Asmussen:Applied}.
\begin{theorem} 
\label{thm:RegProcLimitDistribution}
    Let $\{Y(t)\}_{t\in \mathbb{T}}$ be a (possibly delayed) regenerative process with metric state space, right-continuous paths and non-lattice cycle length distribution $F$ with finite mean $\mu$. Then the limiting distribution $\mathbb{P}_e$ (say) of $Y(t)$ exists and is given by
    \begin{equation}
    \label{eq:RegProcLimitingDist}
    \mathbb{E}_e[f(Y(t))] = \frac{1}{\mu} \mathbb{E}_0\left[ \int_0^X f(Y(s)) \text{\emph{d}}s \right].
    \end{equation}
\end{theorem}
\begin{proof}
    It is possible to show that
    \[
    A \longmapsto \frac{1}{\mu} \mathbb{E}_0\left[ \int_0^X \mathbb{I}(Y(s)\in A)\text{d}s \right]
    \]
    defines a probability measure on the Borel $\sigma$-algebra of the state space $E$ \cite{Asmussen:Applied}. Hence, it is sufficient to prove that
    \[
    \mathbb{E}[f(Y(t))] \longrightarrow \mathbb{E}_e[f(Y(t))]
    \]
    whenever $f$ is continuous and $0\leq f\leq 1$.

    Let $Z(t) = \mathbb{E}_0[f(Y(t))], \:z(t) = \mathbb{E}_0[f(Y(t)),X>t], \:F_0^*(x) = \mathbb{P}(X_0\leq x)$. Now apply the renewal argument, by conditioning on $X_0$:
    \[
    \begin{split}
        \mathbb{E}[f(Y(t))] &= \mathbb{E}[f(Y(t)), X_0>t] + \int_0^t \mathbb{E}[f(Y(t))|X_0 = s]F_0^*(\text{d}s)\\
        Z(t) &= z(t) + \int_0^t Z(t-s)F(\text{d}s).
    \end{split}
    \]
    At this point, an application of the Key Renewal Theorem \ref{thm:keyRen} yields that $Z(t) \rightarrow \frac{1}{\mu} \int_0^\infty z(s) \text{d}s = \frac{1}{\mu} \int_0^\infty \mathbb{E}_0[f(Y(s)), s<X]\text{d}s$.

    Thus, we need to show that $z(t)$ is d.R.i. Notice that $z$ is right-continuous, hence continuous almost everywhere. Moreover, $z(t)\leq z^*(t) = \mathbb{P}(X>t)=\bar{F}(t)$, which is d.R.i. because it is nonincreasing and Lebesgue integrable (c.f. Proposition \ref{prop:DRI}). Using again Proposition \ref{prop:DRI}, $z(t)$ is d.R.i., which concludes the proof.
\end{proof}


Theorem \ref{thm:RegProcLimitDistribution} can be strengthened to include total variation convergence, which is what will be used in the following chapters to study the rate of convergence of a regenerative process \cite{Asmussen:Applied}.
\begin{corollary} 
\label{cor:TotalVariationEt}
    If $E(t)$ converges in total variation to $F_0$ (i.e. the distribution with density $\bar{F}(x)/\mu$), then also $Y(t)$ converges in total variation to:
    \[
    \mathbb{E}_e[f(Y(t))] = \frac{1}{\mu} \mathbb{E}_0\left[ \int_0^X f(Y(s)) \text{\emph{d}}s \right].
    \]
\end{corollary}

We conclude the Chapter with a practical example, related to renewal processes \cite{Asmussen:Applied}.
\begin{example}[Renewal Process] 
    Consider a renewal process with nonlattice interarrival distribution $F$. If $\mu<\infty$, the stationary distributions of the recurrence times $E(t),R(t)$ and of the current life $C(t)$ have already been found.

    Let us show that their particular form comes from formula \eqref{eq:RegProcLimitingDist}: for $0\leq t < X$ we have $E(t) = t, R(t) = X-t, C(t) = X$. In particular,
    \[
    \begin{split}
        \mathbb{P}_e(E(t)\leq \xi) &= \frac{1}{\mu} \mathbb{E}_0\left[ \int_0^X \mathbb{I}(E(t) \leq \xi)\text{d}t \right] = \frac{1}{\mu}\mathbb{E}_0\left[ \int_0^X \mathbb{I}(t \leq \xi)\text{d}t \right]\\
        &\stackrel{(*)}{=} \frac{1}{\mu} \mathbb{E}_0\left[ \int_0^X \mathbb{I}(X-t \leq \xi)\text{d}t \right] = \mathbb{P}_e(R(t)\leq \xi),
    \end{split}
    \]
    and the common value is
    \[
    \begin{split}
        \frac{1}{\mu} \mathbb{E}_0\left[ \int_0^\infty \mathbb{I}(t\leq \xi, t<X)\text{d}t \right] &= \frac{1}{\mu} \int_0^\xi \mathbb{P}_0(t<X) \text{d}t\\
        &= \frac{1}{\mu} \int_0^\xi \bar{F}(t)\text{d}t\\
        &= F_0(\xi).
    \end{split}
    \]
    Note that equality $(*)$ holds because of a simple change of variables inside the integral (time-reversal with respect to $X$): $u = X-t, \text{d}u=-\text{d}t$.
    
    Finally,
    \[
    \begin{split}
        \mathbb{P}_e(C(t)\leq \xi) &= \frac{1}{\mu} \mathbb{E}_0\left[ \int_0^X \mathbb{I}(C(t)\leq \xi) \text{d}t \right] = \frac{1}{\mu} \mathbb{E}_0\left[ \int_0^X \mathbb{I}(X\leq\xi)\text{d}t \right]\\
        &= \frac{1}{\mu} \mathbb{E}_0[X,X\leq\xi] = \frac{1}{\mu} \int_0^\xi x F(\text{d}x).
    \end{split}
    \]

    With the same reasoning we can also prove the limiting distribution of the relative position of the current item, $E(t)/C(t)$. Indeed, for a given $\xi \in [0,1]$:
    \[
    \begin{split}
    \mathbb{P}_e\left(\frac{E(t)}{C(t)}\leq \xi\right) &= \frac{1}{\mu} \mathbb{E}_0 \left[ \int_0^X \mathbb{I}\left(\frac{t}{X} \leq \xi \right) \text{d}t \right]\\
    &= \frac{1}{\mu} \int_0^\infty \int_0^\infty \mathbb{I}(t \leq x) \mathbb{I}\left(\frac{t}{x} \leq \xi\right) \text{d}t\, F(\text{d}x)\\
    &= \frac{1}{\mu} \int_0^\infty \int_0^\infty \mathbb{I}(x \geq t) \mathbb{I}\left(x \geq \frac{t}{\xi}\right) F(\text{d}x)\text{d}t\\
    &\stackrel{(*)}{=} \frac{1}{\mu} \int_0^\infty \bar{F}\left(\frac{t}{\xi}\right) \text{d}t\\
    &= \frac{\xi}{\mu} \int_0^\infty \bar{F}(s)\text{d}s\\
    &= \xi,
    \end{split}
    \]
    where $(*)$ holds because $\xi \leq 1$.
\end{example}

\chapter[Exponential Convergence of Regenerative Processes]{Exponential Convergence of Regenerative Processes: The Coupling Proof}
\label{chap:4}
This Chapter presents in detail the proof of the exponential rate of convergence of renewal and regenerative processes, by first introducing the necessary background theory needed to comprehend it.

\section{Spread-out Distributions}
We have seen in the previous chapter that, under very mild assumptions (cf. Theorem \ref{thm:RegProcLimitDistribution}), a regenerative process converges in law to a limiting distribution, given by formula \eqref{eq:RegProcLimitingDist}. However, by slightly strengthening the assumptions, we obtain convergence in total variation, without considering the backward recurrence time process, as in Corollary \ref{cor:TotalVariationEt}. To this end, we first introduce the concept of \emph{spread-out distributions}.

As a reminder,  a \emph{component of a distribution} $F$ on $\mathbb{R}$ is a nonnegative measure $G$ such that $0\neq G\leq F$. This is a concept that is tightly linked to the Lebesgue-Radon-Nikodym decomposition of a regular Borel measure $\nu$ \cite{Shum:MeasureTheoreticProbability,Cohn:MeasureTheory,Rudin:RealComplexAnalysis}, for which we can decompose $\nu$ as $\nu = \nu_{ac} + \nu_{sc} + \nu_d$, where $\nu_{ac}$ is an absolutely continuous component, $\nu_{sc}$ is a continuous singular component, and $\nu_d$ is a discrete component. Since all probability distributions on $\mathbb{R}$ are regular Borel measures, this type of decomposition always exists in our probabilistic setting. However, the components do not need to be non-trivial, e.g. a discrete probability measure, or a degenerate distribution $F=\delta_0$.

\begin{definition} 
\label{def:SpreadOut}
    We say that a distribution on $\mathbb{R}$ is \emph{spread-out} if there exists an $n \in \mathbb{N}$ such that $F^{*n}$ has a component $G$ that is \emph{non-trivially} absolutely continuous (i.e. has density $g$ with respect to Lebesgue measure).
\end{definition}

\begin{remark}
    Any absolutely continuous distribution is trivially spread-out, but there are also some examples of singular distributions that are spread-out. One example can be found in Appendix \ref{chap:NonTrivialSpreadOut}.
\end{remark}

If we require our distributions to be spread-out instead of non-lattice, we can guarantee stronger convergence results. And in applications, the cases in which $F$ is nonlattice and spread-out are essentially the same. Hence, requiring this stronger condition is not prohibitive \cite{Asmussen:Applied}.

Spread-out distributions have a useful characteristic \cite[Chap. VII]{Asmussen:Applied}.
\begin{lemma} 
\label{lemma:uniformComponent}
    If $F$ is spread-out, then $F^{*m}$ has a uniform component on $(a,a+b)$ for some $m$ and $a,b>0$.
\end{lemma}
\begin{proof}
    Since $F$ is spread-out, there exists $n$ such that $F^{*n}$ has a non-trivial absolutely continuous component with density $g$. We can assume that $g$ is bounded with compact support.

    Let us choose continuous bounded functions $g_k\in L_1$ with $\|g-g_k\|_1=\int|g-g_k| \rightarrow 0$. Then $g_k\ast g(x) = \int g_k(x-y)g(y)\text{d}y$ is continuous by dominated convergence. Moreover, $\|g^{*2} - g_k\ast g\| \leq \|g\|_\infty \|g-g_k\|_1 \rightarrow 0$. Thus $g^{*2}$ is continuous as the uniform limit of continuous functions. Hence, there exists $a,b,\delta>0$ such that $g^{*2}(x)\geq\delta$ for $x\in(a,a+b)$ (we will talk about this decomposition more in detail, see Example \ref{ex:Exp1}).

    We conclude the proof by taking $m=2n$.
\end{proof}

One of the basic tools that are used when dealing with spread-out \emph{interarrival} distributions is \emph{Stone's decomposition} \cite{Stone:Decomposition,Asmussen:Applied}.
\begin{theorem}[Stone's decomposition] 
\label{thm:StoneDec}
    If the interarrival distribution $F$ of a renewal process is spread-out, then we can write the renewal measure as $U=U_1+U_2$, where $U_1,U_2$ are nonnegative measures on $[0,\infty)$, $U_1$ has a bounded continuous density $u_1(x) = \text{\emph{d}}U_1(x)/\text{\emph{d}}x$ satisfying $u_1(x) \rightarrow 1/\mu$ as $x\rightarrow\infty$, and $U_2$ is bounded, i.e. $\|U_2\|<\infty$.
\end{theorem}

We do not provide the proof of this result (it can be seen on \cite{Stone:Decomposition,Asmussen:Applied}), but we still want to provide some intuition on how the two measures $U_1,U_2$ are defined. If $G$ is the uniform component of Lemma \ref{lemma:uniformComponent}, in the case $m=1$, $U_2=\sum_{n=0}^\infty H^{*n}$, where $H$ is the residual component of $F$, $H=F-G$. Then, $U_1$ is defined as $U_1 = G\ast U_2 \ast U$, which has density $u_1 = U_2\ast(U\ast g)$. 

Let us provide a simple example, related to the Poisson Process examples of Chapter \ref{chap:2}.
\begin{example}
\label{ex:Exp1}
Assume the interarrival distribution of a renewal process is $F\sim Exp(\lambda)$, i.e. absolutely continuous with density $f(x)=\lambda e^{-\lambda x},\:x\in[0,\infty)$. Clearly, since $F$ is already non-trivially absolutely continuous, it is also spread-out, with $n=1$ in definition \ref{def:SpreadOut}. This is true for every absolutely continuous distribution, not just the exponential.

Since $F$ is spread-out, by Lemma \ref{lemma:uniformComponent}, we would expect to find $m\in\mathbb{N}\setminus\{0\}$ such that $F^{*m}$ has a uniform component on $(a,a+b)$, for $a,b>0$. Indeed, for all $a\geq 0$ and $b>0$,
\[
\alpha \colon \!\!\!= \inf_{y\in(a,a+b)}f(y) = \lambda e^{-\lambda (a+b)} \leq f(x),\quad \forall\,x\in(a,a+b).
\]
This means that the density can be decomposed as follows. Define $\varepsilon\colon \!\!\!= \alpha b\in(0,1]$ and let for all $x\in[0,\infty)$
\[
f(x) = \frac{\varepsilon}{b}\mathbb{I}_{(a,a+b)}(x) + (1-\varepsilon)h(x),
\]
where $h$ is the \emph{residual} density defined by
\[
h(x) = \frac{f(x)-\alpha\mathbb{I}_{(a,a+b)}(x)}{1-\varepsilon}.
\]
Hence, $F$ is decomposed into a uniform component on $(a,a+b)$ and a \emph{residual} component, which implies that $m=1$. This decomposition is valid for every absolutely continuous (hence spread-out) distribution with continuous density. If the density is \emph{not} continuous, but the distribution is still absolutely continuous (so that $n=1$ in the spread-outness definition), we might need to apply the convolution to obtain the uniform component (meaning $m\geq1$).

Lastly, we  compute Stone's decomposition of the renewal measure $U$ for $F\sim Exp(\lambda)$. As shown in Example \ref{ex:PP2}, $U(\text{\emph{d}}x) = \delta_0(\text{\emph{d}}x)+\lambda \text{\emph{d}}x$. It is then clear that $U_1(\text{\emph{d}}x) = \lambda\text{\emph{d}}x$ has absolutely continuous density $u_1(x) = \lambda$ which converges to $1/\mu=\lambda$ as $x\rightarrow\infty$, and $U_2 = \delta_0$ is bounded.
\end{example}

At this point, we present the main consequences of having a spread-out interarrival distribution. The first is a modification of the Key Renewal Theorem \ref{thm:keyRen}, where the stronger assumption on $F$ permits us to weaken the d.R.i. assumption on $z$ \cite{Asmussen:Applied}.

\begin{theorem}[Key Renewal Theorem 2.0] 
    Let $z$ be bounded and Lebesgue integrable with $z(x)\rightarrow0$ as $x\rightarrow\infty$. If $F$ is spread-out, then
    \[
    U\ast z(x) \longrightarrow\frac{1}{\mu}\int_0^\infty z(t) \text{\emph{d}}t.
    \]
\end{theorem}
\begin{proof}
    By the dominated convergence theorem,
    \[
    \begin{split}
        Z(x) = U\ast z(x) &= U_1\ast z(x) + U_2\ast z(x)\\
        &= \int_0^x z(y)u_1(x-y)\text{d}y + \int_0^x z(x-y)U_2(\text{d}y)\\
        &\rightarrow \int_0^\infty z(y) \frac{1}{\mu}\text{d}y + \int_0^\infty 0\cdot U_2(\text{d}y).
    \end{split}
    \]
\end{proof}

As can be seen from the proof, the previous theorem is a direct consequence of Stone's decomposition, which remarks the impact that such a tool has when dealing with spread-out distributions.

The next result guarantees total variation convergence of a regenerative process, provided the cycle length distribution is spread-out \cite{Asmussen:Applied}.
\begin{theorem}[TV convergence of regenerative processes] 
\label{thm:RegProcTVConv}
    Consider a regenerative process $\{Y(t)\}_{t\geq0}$ with cycle length distribution $F$ spread-out with finite mean $\mu<\infty$. Suppose that $Y(t,\omega)$ is measurable jointly in $(t,\omega)$. Then, no matter the initial conditions, the limiting distribution $\mathbb{P}_e$ of $Y(t)$ exists in the sense of total variation convergence and is given by
    \[
    \mathbb{E}_e[f(Y(t))] = \frac{1}{\mu} \mathbb{E}_0\left[ \int_0^X f(Y(s)) \text{\emph{d}}s \right].
    \]
\end{theorem}

Interestingly, $F$ being spread-out is also a necessary condition for total variation convergence \cite{Asmussen:Applied}.
\begin{lemma} 
\label{lemma:spreadOutEquivF0}
Let $\{R(t)\}_{t\geq0}$ be the forward recurrence time process of a renewal process with interarrival distribution $F$ with finite mean $\mu$. Define $G_t(x) = \mathbb{P}(R(t)\leq x)$, and $F_0$ to be the distribution with density $\bar{F}(x)/\mu$ w.r.t. Lebesgue measure. Then $G_t \rightarrow F_0$ in total variation for any distribution of the initial delay if and only if $F$ is spread-out.
\end{lemma}

\section{Coupling}
The term coupling is used in the literature in two different ways: in a broad and in a narrow (more classical) sense \cite{Asmussen:Applied}. In this thesis we will mainly focus on the latter, but, for the sake of completeness, we will also define the former.
\begin{definition}[Coupling, broad sense] 
    A \emph{coupling of two probability distributions} $\mathbb{P}',\mathbb{P}''$ on $(\Omega', \mathcal{F'})$ and $(\Omega'', \mathcal{F''})$, respectively, is defined as a probability distribution $\mathbb{P}$ on $(\Omega, \mathcal{F})=(\Omega'\times \Omega'', \mathcal{F'}\otimes \mathcal{F}'')$ having marginals $\mathbb{P}'$ and $\mathbb{P}''$.

    We say \emph{a coupling of} $X',X''$, where $X',X''$ are random variables, to denote a pair $(\tilde{X}',\tilde{X}'')$ of random variables defined on a common probability space such that $\tilde{X}'\stackrel{\mathscr{D}}{=} X'$ and $\tilde{X}''\stackrel{\mathscr{D}}{=} X''$. For ease of notation, one can omit the tilde, which simply means the random variables have been redefined on a common probability space without changing the marginals.
\end{definition}

\begin{definition}[Coupling, narrow sense] 
    In the narrow sense, coupling refers to two stochastic processes $\{Y'(t)\}_{t\in\mathbb{T}}, \{Y''(t)\}_{t\in\mathbb{T}}$ with the same state space $E$ and an associated random time $T\in\mathbb{T}$ such that
    \begin{equation}
    \label{eq:couplingTime}
        Y'(t)=Y''(t),\quad \text{\emph{for all }} t\geq T.
    \end{equation}
\end{definition}

Let us analyze more in detail how coupling is useful to the theory of convergence of regenerative processes. We start with a proposition related to coupling in broad sense \cite{Lindvall:Coupling,Asmussen:Applied}.
\begin{proposition} 
    Let $X',X''$ be random variables taking values in the same state space $E$ and defined on a common probability space. Then,
    \begin{equation}
        \|\mathbb{P}(X'\in\cdot)-\mathbb{P}(X''\in\cdot)\| \leq \mathbb{P}(X' \neq X''),
    \end{equation}
    where $\|\mathbb{P}(X'\in\cdot)-\mathbb{P}(X''\in\cdot)\| = \sup_{A}|\mathbb{P}(X'\in A)-\mathbb{P}(X''\in A)|$ represents the total variation distance.
\end{proposition}
\begin{proof}
    We have, for $A\subset E$,
    \[
    \begin{split}
        |\mathbb{P}(X'\in A) - \mathbb{P}(X''\in A)| &= |\mathbb{P}(X'\in A, X'=X'') + \mathbb{P}(X'\in A, X'\neq X'')\\
        &\quad - \mathbb{P}(X''\in A,X'=X'') - \mathbb{P}(X''\in A, X'\neq X'')|\\
        &= |\mathbb{P}(X'\in A,X'\neq X'') - \mathbb{P}(X''\in A, X'\neq X'')|\\
        &\leq \mathbb{P}(X'\neq X'').
    \end{split}
    \]
    We conclude the proof by taking the supremum over $A$.
\end{proof}

From this proposition one can derive a crucial inequality, called \emph{coupling inequality}. We show that there are actually two versions of this result, one slightly stronger than the other \cite{Lindvall:Coupling,Asmussen:Applied}.
\begin{corollary}[Coupling Inequality 1] 
    Let $Y'=\{Y'(t)\}_{t\in\mathbb{T}}$ and $Y''=\{Y''(t)\}_{t\in\mathbb{T}}$ be two stochastic processes defined on the same probability space. If there is a random time $T$ such that \eqref{eq:couplingTime} holds, then
    \begin{equation}
        \label{eq:couplingIneq1}
        \|\mathbb{P}(Y'(t) \in \cdot)-\mathbb{P}(Y''(t)\in \cdot)\| \leq \mathbb{P}(T>t).
    \end{equation}
\end{corollary}
\begin{proof}
    The result trivially holds by applying the previous proposition and by noting that
    \[
    \{Y'(t)\neq Y''(t)\} \subset \{T>t\}.
    \]
\end{proof}

\begin{corollary}[Coupling Inequality 2] 
    Let $Y'=\{Y'(t)\}_{t\in\mathbb{T}}$ and $Y''=\{Y''(t)\}_{t\in\mathbb{T}}$ be two stochastic processes defined on the same probability space, and let $\theta_t$ be the shift, i.e. $(\theta_tY')_s=Y'_{t+s}$. Assume there is a random time $T$ such that \eqref{eq:couplingTime} holds, then
    \begin{equation}
        \label{eq:couplingIneq2}
        \|\mathbb{P}(\theta_tY' \in \cdot) - \mathbb{P}(\theta_tY'' \in \cdot)\| \leq \mathbb{P}(T>t),
    \end{equation}
    where $\theta_tY'$ and $\theta_tY''$ denote the \emph{whole} shifted processes.
\end{corollary}

The proof of the second inequality \eqref{eq:couplingIneq2} is equivalent to the one for the first. Indeed, if $T$ is a coupling time for $\{Y'(t)\}_{t\in\mathbb{T}}$ and $\{Y''(t)\}_{t\in\mathbb{T}}$, it is also a coupling time for the shifted processes \cite{Lindvall:Coupling}.

The stronger inequality is the one regarding the whole shifted processes. We can show that, if it holds, then also \eqref{eq:couplingIneq1} does. Indeed,
\[
\begin{split}
\|\mathbb{P}(Y'(t)\in \cdot)-\mathbb{P}(Y''(t)\in \cdot)\| &= \|\mathbb{P}((\theta_tY')_0\in \cdot)-\mathbb{P}((\theta_tY'')_0\in \cdot)\|\\
&\stackrel{(*)}{\leq} \|\mathbb{P}(\theta_tY'\in \cdot)-\mathbb{P}(\theta_tY''\in \cdot)\|\\
&\leq \mathbb{P}(T>t),
\end{split}
\]
where $(*)$ holds because, by applying a measurable mapping, we cannot increase the total variation distance \cite{Lindvall:Coupling}.

The coupling inequalities have two main applications: they can be used to show convergence in distribution of $Y(t)$ as $t\rightarrow\infty$ (see Remark \ref{rmk:couplingConvDist}) and also to obtain estimates for the rate of convergence (see Remark \ref{rmk:couplingRate}).

\begin{remark}
\label{rmk:couplingConvDist}
For the sake of the example, let us consider a renewal setting. Hence, the goal is to show convergence in distribution of a renewal process $\{T_n\}$ with interarrival distribution $F$. We let $\{T'_n\}$ be stationary, i.e. started by letting the delay $X_0$ have the stationary distribution $F_0$ (c.f. Corollary \ref{cor:StatDelay}), and $\{T''_n\}$ be zero-delayed. Also consider the associated forward recurrence time processes, $\{R'(t)\}_{t\geq0}$ and $\{R''(t)\}_{t\geq0}$. If a coupling with $T<\infty$ can be constructed, we have $\mathbb{P}(T>t)\rightarrow0$ and, in an obvious notation
\begin{equation}
    \label{eq:couplingConvDist}
    \|\mathbb{P}(R(t)\in \cdot)-F_0(\cdot)\| \leq \|\mathbb{P}(R''(t)\in\cdot)-\mathbb{P}(R'(t)\in \cdot)\|\leq \mathbb{P}(T>t) \longrightarrow0,
\end{equation}
as $t\rightarrow\infty$.

We shall say that the process \emph{admits coupling} if there exist two stochastic processes $\{R'(t)\},\{R''(t)\}$ defined on the same probability space such that the two processes have the same interarrival distribution and \eqref{eq:couplingTime} holds for some finite random time $T<\infty$.
\end{remark}

\begin{remark}
\label{rmk:couplingRate}
On top of the discussion in Remark \ref{rmk:couplingConvDist}, if one shows that $T$ can be chosen with $\mathbb{E}[\varphi(T)]<\infty$ for some nonnegative function $\varphi$ increasing to $\infty$ (usually $\varphi(t)=t^p$ or $\varphi(t)=e^{\varepsilon t}$), then
\begin{equation}
    \label{eq:couplingRate}
    \|\mathbb{P}(R(t)\in \cdot)- F_0(\cdot)\| \stackrel{(*)}{\leq} \mathbb{P}(T>t) \stackrel{(**)}{\leq} \frac{1}{\varphi(t)}\mathbb{E}[\varphi(T)] = O\left(\frac{1}{\varphi(t)}\right),
\end{equation}
where $(*)$ represents the calculations done in Remark \ref{rmk:couplingConvDist}, and $(**)$ holds because
\[
\varphi(T) \geq \varphi(T)\,\mathbb{I}(T>t) \geq \varphi(t)\,\mathbb{I}(T>t) \implies\mathbb{E}[\varphi(T)] \geq \varphi(t)\, \mathbb{P}(T>t).
\]
The convergence rates obtained with this method are not necessarily the best possible \cite{Asmussen:Applied}.
\end{remark}

\section{Exponential Convergence of Regenerative Processes}
This section is devoted to the proof of the theorem regarding the geometric convergence (in total variation) of regenerative processes to their limiting distribution. We start with a lemma that is derived from Lemma \ref{lemma:uniformComponent} \cite{Asmussen:Applied}.
\begin{lemma} 
\label{lemma:UniformComponentRt}
    For a zero-delayed spread-out renewal process, there exist $A,b$ such that the distributions of $R(t)$, for $t\geq A$,have a common uniform component on $(0,b)$. That is, for some $\delta\in (0,1)$ and all $t\geq A$,
    \[
    \mathbb{P}(u<R(t)\leq b) \geq \delta\frac{v-u}{b},\quad 0<u<v<b.
    \]
\end{lemma}

Let us provide a (trivial) example of this.
\begin{example}
    Recall the setting of Examples \ref{ex:PP1}, \ref{ex:PP2}, \ref{ex:PP3}, \ref{ex:PP4}, and \ref{ex:PP5}. We consider a zero-delayed renewal process with interarrival distribution $F\sim Exp(\lambda)$, which is spread-out, and has finite mean $\mu=1/\lambda$.

    As was shown in the related examples, $R(t)$ has $Exp(\lambda)$ distribution for all $t\in[0,\infty)$. We also showed in Example \ref{ex:Exp1} that every absolutely continuous distribution with continuous density $f$ admits has a uniform component on $(a,a+b)$ for given $a,b>0$. In the specific case of the $Exp(\lambda)$ distribution, $a$ can be taken equal to 0.

    Since $R(t)$ is stationary, with $Exp(\lambda)$ distribution, the laws of $R(t)$ have a common uniform component on $(0,b)$ for all $t\geq 0$ and $b>0$.
\end{example}

\begin{remark}
\label{rmk:AlsoForDelayed}
    The previous Lemma \ref{lemma:UniformComponentRt} is stated for zero-delayed processes. However, it only depends on the properties of the interarrival distribution $F$. Moreover, it is not a result regarding the initial delay, since the statement only concerns big enough times ($t\geq A$). Thus, the Lemma holds also for delayed processes, provided their interarrival distribution is spread-out (since after the first renewal, their law only depends on $F$).
\end{remark}

In the following, we adopt the convention that a renewal process admits coupling if $\{R(t)\}_{t\geq0}$ does so. The next theorem can be found in \cite[Theorem VII.2.7]{Asmussen:Applied}.
\begin{theorem}[Coupling of renewal process] 
\label{thm:CouplingTheorem}
    A nonlattice renewal process with finite interarrival mean $\mu<\infty$ admits coupling if and only if $F$ is spread-out.
\end{theorem}
\begin{proof}
    The interesting implication to prove is $(\impliedby)$. Let us start from the other one.

    Assume that the renewal process admits coupling. By applying the coupling inequality \eqref{eq:couplingIneq2} to $\{R(t)\}_{t\geq0}$, since a coupling exists, we get total variation convergence of the forward recurrence time process to its stationary distribution. Hence, by Lemma \ref{lemma:spreadOutEquivF0}, $F$ is spread-out.

    For the opposite implication, let us assume that $F$ is spread-out. Following the reasoning in Remark \ref{rmk:couplingConvDist}, we will construct two processes: one zero-delayed, $\{T_n\}_{n\in\mathbb{N}}$, and one stationary, $\{T'_n\}_{n\in\mathbb{N}}$. We define them on the same probability space via a recursive construction. After step $k$ ($k=0,1,2,\dots$), the processes and their residual life processes $\{R(t)\}_{t\geq0}$, $\{R'(t)\}_{t\geq0}$ will have been constructed in a certain random interval $[0,t_k]$. We define them as follows:

    Let $t_0=0, R(t_0)=0$ (zero-delayed process) and $R'(t_0)\sim F_0$ (stationary process). At each step $k$, assign
    \[
    \begin{split}
        L_k &= \max\{R(t_k),R'(t_k)\},\quad t_{k+1} = t_k+L_k+A\geq A,\\
        s_k &= L_k+A-R(t_k), \quad s'_k = L_k+A-R'(t_k),
    \end{split}
    \]
    where $A$ is defined in Lemma \ref{lemma:UniformComponentRt}. Hence, all the $t_k$'s, for $k\geq1$ are greater than or equal to $A$, and also $s_k,s'_k\geq A$. By definition of forward recurrence time process, at times $t_k+R(t_k), t_k+R'(t_k)$ we observe renewals for $\{R(t)\}$ and $\{R'(t)\}$, respectively. Additionally, $t_{k+1}-s_k = t_k+R(t_k)$ and $t_{k+1}-s'_k=t_k + R'(t_k)$. Since the processes restart after each renewal, we have $R(t_{k+1}) \stackrel{\mathscr{D}}{=} R(s_k)$ and $R'(t_{k+1}) \stackrel{\mathscr{D}}{=} R'(s'_k)$.

    The goal of this construction is to define $t_k,s_k,s'_k,\:k\geq1$ greater than or equal to $A$ (so that we can apply Lemma \ref{lemma:UniformComponentRt}), and to obtain that $R(t_{k+1}) \stackrel{\mathscr{D}}{=} R(s_k)$ and $R'(t_{k+1}) \stackrel{\mathscr{D}}{=} R'(s'_k)$, so that it is possible to compute the value of $R(t_{k+1}),R'(t_{k+1})$.

    After the first renewal of $\{T'_n\}_{n\in\mathbb{N}}$, the stationary process' law only depends on $F$, which is the same interarrival distribution as the zero-delayed process. Hence, by Remark \ref{rmk:AlsoForDelayed}, we can apply Lemma \ref{lemma:UniformComponentRt} to $\{R(t)\}_{t\geq0},\{R'(t)\}_{t\geq0}$ and obtain that the laws of $R(t)$ and $R'(t)$ have the same uniform component for $t\geq A$, where $A$ is big enough so that the Lemma can be applied to the delayed process too. In other words, we can choose $U_k,V_k,H_k,H'_k$ (independent and independently of all the past $U_l,V_l,H_l,H'_l$), such that $\mathbb{P}(U_k=1) = 1-\mathbb{P}(U_k=0)=\delta$, that $V_k$ is uniform on $(0,b)$, and that
    \begin{equation}
    \label{eq:RtDecomposition}
    R(t_{k+1})= U_kV_k + (1-U_k)H_k,\quad R'(t_{k+1})=U_kV_k + (1-U_k)H'_k.
    \end{equation}

    By the above discussion, the processes defined in Equation \eqref{eq:RtDecomposition} have, by definition, the distributions of $R(s_k),R'(s'_k)$, respectively. Hence, the renewals for $\{T_n\}_{n\in\mathbb{N}}$ in $[t_{k+1}-s_k,t_{k+1}]$ are taken according to the conditional distribution of the renewal process given that its overshoot at time $s_k$ (i.e. $R(s_k)$) has the value of the constructed $R(t_{k+1})$. Similarly for $\{T'_n\}_{n\in\mathbb{N}}$.

    This construction is stopped at step $\sigma=\inf\{k:U_k=1\}$. Then, $R(\sigma+1)=R'(\sigma+1)$, and at time $T=t_{\sigma+1}+L_{\sigma+1}$ the two renewal processes will have a common renewal. Hence, by defining a new renewal process $\{T''_n\}_{n\in\mathbb{N}}$ as the renewal process with the same renewals of $\{T_n\}_{n\in\mathbb{N}}$ before $T$, and with the same renewals as $\{T'_n\}_{n\in\mathbb{N}}$ after $T$, we obtain the desired coupling.
\end{proof}

We analyze this construction from a different perspective in Section \ref{sec:CouplingProofDetailed}. 

 The coupling strategy is crucial for understanding the rest of the discussion, as everything will be based on the previous proof, specifically on the coupling time $T$. Indeed, we have the following Lemma \cite{Asmussen:Applied}.
\begin{lemma} 
\label{lemma:ExpMomentsT}
    If $\int_0^\infty e^{\eta x}F(\text{\emph{d}}x) <\infty$ for some $\eta>0$, then also $\mathbb{E}[e^{\varepsilon T}]<\infty$ for some $\varepsilon>0$.
\end{lemma}
\begin{proof}
    First, note that $\mathbb{P}(\sigma=n)=\delta(1-\delta)^n$.
    
    Since $z(t) = \mathbb{E}[e^{\eta R(t)};t<X] \leq e^{-\eta t}\mathbb{E}[e^{\eta X}]$ is d.R.i. (point 4. of Proposition \ref{prop:DRI}), the classical renewal argument yields the convergence of $Z(t)=\mathbb{E}[e^{\eta R(t)}]$ to a finite limit. In particular, $Z(t)\leq c_1<\infty$ for all $t$. In addition, we have
    \[
    \begin{split}
        \mathbb{E}[e^{\eta (A+L_{k+1})}|R(t_k),R'(t_k)] &\leq \mathbb{E}[e^{\eta (A+R(t_{k+1}))} + e^{\eta (A+R'(t_{k+1}))}|R(t_k),R'(t_k)]\\
        &= \mathbb{E}[e^{\eta (A+R(s_k))} + e^{\eta (A+R'(s'_k))}]\\
        &\leq c,
    \end{split}
    \]
    where $c=2c_1e^{\eta A}$. Similarly, $c_2=\mathbb{E}[e^{\eta (A+L_0)}]<\infty$, and by letting $S_n=\sum_0^n (A+L_k)$, it follows by induction that  $\mathbb{E}[e^{\eta S_n}]\leq c_2 c^n$. 

    For some large $p$ and some $q$ close to $1$, such that $1/p+1/q=1$, it holds that $c^{1/p}(1-\delta)^{1/q}<1$. We can then define $\varepsilon=\eta/p$ and use H{\"o}lder inequality $(*)$ to show
    \[
    \begin{split}
        \mathbb{E}[e^{\varepsilon T}] \leq \mathbb{E}[e^{\varepsilon S_{\sigma+1}}] &= \sum_{n=1}^\infty \mathbb{E}[e^{\varepsilon S_{n+1}};\sigma=n]\\
        &\stackrel{(*)}{\leq} \sum_{n=1}^\infty \left( \mathbb{E}[e^{\eta S_{n+1}}] \right)^{1/p}\left( \mathbb{P}(\sigma=n) \right)^{1/q}\\
        &\leq c_3 \sum_{n=1}^\infty c^{(n+1)/p}(1-\delta)^{n/q} <\infty,
    \end{split}
    \]
    where $c_3 = c_2^{1/p}\delta^{1/q}$.
\end{proof}

This proposition gives us an idea on how to compute the rate $\varepsilon$, but it is not clear how we can do it in practice. For some special case, it is probably possible to compute the actual value of $\varepsilon$. In any case, the only explicit information that we have on $\varepsilon$, is that $0<\varepsilon<\eta$.

With Lemma \ref{lemma:ExpMomentsT}, we have all the requirements to prove the exponential convergence of renewal processes and of regenerative processes \cite{Asmussen:Applied}.
\begin{theorem} 
\label{thm:exponentialRateRt}
    If $F$ is spread-out, and $\int_0^\infty e^{\eta x}F(\text{\emph{d}}x)<\infty$ for some $\eta >0$, then for some $\varepsilon>0$ we have:
    \begin{enumerate}
        \item $\|\mathbb{P}(R(t)\in \cdot)-F_0(\cdot)\| = O(e^{-\varepsilon t})$.
        \item In Stone's decomposition, $U_2([x,\infty))=O(e^{-\varepsilon x})$ and $u_1(x) = 1/\mu + O(e^{-\varepsilon x})$.
        \item \label{bullet:ExponentialSolutionRenEq}If $z$ is measurable with $z(x) = O(e^{-\xi x})$ for some $\xi>\varepsilon$, then
        \[
        U\ast z(x) = \frac{1}{\mu} \int_0^\infty z(t) \text{\emph{d}}t + O(e^{-\varepsilon x}).
        \]
    \end{enumerate}
\end{theorem}
\begin{proof}
    Statement \textit{1.} is a direct consequence of Equation \eqref{eq:couplingRate} and Lemma \ref{lemma:ExpMomentsT}.

    Statement \textit{2.} can be proved by re-checking the proof of Stone's decomposition, but it will not be shown here (see \cite{Asmussen:Applied} for reference).

    Statement \textit{3.} is proven as follows: first note that, if $z(x) = O(e^{-\xi x})$, then it holds $z(x) = O(e^{-\varepsilon x})$, for all $\varepsilon<\xi$. Then, we have
    \[
    \begin{split}
        U\ast z(x) &= \int_0^x z(x-y)U_2(\text{d}y) + \int_0^x z(x-y)u_1(y)\text{d}y\\
        &= \int_0^xO(e^{-\varepsilon(x-y)})U_2(\text{d}y) + \int_0^x z(y) u_1(x-y) \text{d}y\\
        &= e^{-\varepsilon x} \int_0^x O(e^{-\varepsilon y})U_2(\text{d}y) + \int_0^x z(y) \left( \frac{1}{\mu} + O(e^{-\varepsilon(x-y)}) \right) \text{d}y\\
        &= O(e^{-\varepsilon x}) + \int_0^x \frac{z(y)}{\mu}\text{d}y + \int_0^x z(y)O(e^{-\varepsilon (x-y)})\text{d}y\\
        &= O(e^{-\varepsilon x}) + \int_0^\infty \frac{z(y)}{\mu}\text{d}y - \int_x^\infty \frac{z(y)}{\mu}\text{d}y + e^{-\varepsilon x}\int_0^x O(e^{-(\xi-\varepsilon) y})\text{d}y\\
        &= O(e^{-\varepsilon x}) + \int_0^\infty \frac{z(y)}{\mu}\text{d}y - \int_x^\infty \frac{z(y)}{\mu}\text{d}y + O(e^{-\varepsilon x}).
    \end{split}
    \]
    Now, since $O(e^{-\varepsilon x})+O(e^{-\varepsilon x})=O(e^{-\varepsilon x})$, we can conclude
    \[
    \begin{split}
        U\ast z(x) &= O(e^{-\varepsilon x}) + \int_0^\infty \frac{z(y)}{\mu}\text{d}y - \int_x^\infty \frac{z(y)}{\mu}\text{d}y\\
        &\stackrel{(*)}{=} O(e^{-\varepsilon x}) + \int_0^\infty \frac{z(y)}{\mu}\text{d}y + O(e^{-\varepsilon x})\\
        &= O(e^{-\varepsilon x}) + \int_0^\infty \frac{z(y)}{\mu}\text{d}y,
    \end{split}
    \]
    where $(*)$ holds because (by using the definition of $O(e^{-\xi x})$ with constant $C$)
    \[
    \left| \int_x^\infty z(y)\text{d}y\, \right| \leq \int_x^\infty |z(y)|\,\text{d}y \leq C\int_x^\infty e^{-\xi y} \text{d}y = \frac{C}{\xi} e^{-\xi x} < \frac{C}{\xi} e^{-\varepsilon x},
    \]
    which implies $\int_x^\infty z(y) \text{d}y = O(e^{-\varepsilon x})$.
\end{proof}

This important theorem guarantees that, if the interarrival distribution of a renewal process is light-tailed (i.e. has exponential moments), then the associated residual life process converges exponentially to its stationary distribution in total variation. Additionally, we also have that every solution $U\ast z$ of a renewal equation converges exponentially to its limit $1/\mu\int_0^\infty z(x)\text{d}x$. This last result is exactly what we use to prove the exponential convergence of regenerative processes \cite{Asmussen:Applied}.
\begin{corollary}[Exponential Convergence of Regenerative Processes] 
\label{cor:ExpConvRegProc}
    Define a regenerative process $\{Y(t)\}_{t\geq0}$ such that its cycle length distribution $F$ is spread-out and with finite mean $\mu<\infty$, and assume that $\int_0^\infty e^{\eta x}F(\text{\emph{d}}x)<\infty$ (and in the delayed case $\mathbb{E}[e^{\eta X_0}]<\infty$) for some $\eta>0$. Then
    \begin{equation}
        \|\mathbb{P}(Y(t)\in\cdot) - \mathbb{P}_e(Y(t)\in\cdot)\| = O(e^{-\varepsilon t}),
    \end{equation}
    for some $\varepsilon>0$.
\end{corollary}
\begin{proof}
    Let $z_A(t) = \mathbb{P}_0(Y(t)\in A,t<X)$, then we have $z_A(t)\leq \mathbb{P}_0(X>t)=O(e^{-\eta t})$. By Statement \textit{3.} of Theorem \ref{thm:exponentialRateRt}, we get
    \[
    \mathbb{P}_0(Y(t)\in A) = U\ast z_A(t) = O(e^{-\varepsilon t}) + \mathbb{P}_e(Y(t)\in A),
    \]
    uniformly in $A$, which gives us the desired total variation convergence.

    The delayed case can be seamlessly reduced to the zero-delayed one by conditioning upon $X_0$.
\end{proof}

At this point, we have finally showed which are the conditions under which a regenerative process converges exponentially to its limiting distribution: the cycle length distribution $F$ must have exponential moments (must be light-tailed).
Such conditions are quite strong, and we cannot expect them to be satisfied in the majority of the situations. However, it can be proven by analogous coupling arguments (and via the same strategy in the proof of Corollary \ref{cor:ExpConvRegProc}), that regenerative processes can converge polynomially to their limiting distribution, under weaker assumptions (see \cite{Lindvall:Coupling,Nummelin:RateOfConvergenceOrey,Lindvall:ContRenProcCoupling} for more details on this and related results).
\begin{theorem}[Polynomial Convergence of Regenerative Processes]
    Define a regenerative process $\{Y(t)\}_{t\geq0}$ such that its cycle length distribution is spread-out and has finite mean $\mu<\infty$. Assume that $\int_0^\infty x^{p+1} F(\text{\emph{d}}x)<\infty$ (and in the delayed case $\mathbb{E}[X_0^{p}]<\infty$) for some $p>0$. Then
    \begin{equation}
        \|\mathbb{P}(Y(t)\in\cdot) - \mathbb{P}_e(Y(t)\in\cdot)\| = o(t^{-p}).
    \end{equation}
\end{theorem}

Since the recurring Poisson Process example that we used until now is too trivial to showcase the results on exponential convergence (the forward recurrence time process $\{R(t)\}_{t\geq 0}$ is stationary), we conclude the section by presenting a less trivial example, but where it is still possible to compute everything by hand.
\begin{example}
    Suppose the interarrival times of a zero-delayed renewal process follow a $Gamma(2,\lambda)$ distribution, which is spread-out and has finte mean $\mu=2/\lambda$. Recall that the pdf of such distribution is $f(x)=\lambda^2xe^{-\lambda x}, \:x\in[0,\infty)$, and the cdf is $F(x)=1-e^{-\lambda x}(1+\lambda x)$. Since the pdf of the $Gamma(2,\lambda)$ distribution is continuous, we can use the strategy proposed in Example \ref{ex:Exp1} to obtain that $F$ has a uniform component on $(a,a+b)$, for $a\geq0,b>0$.
    
    By Lemma \ref{lemma:stationaryDist}, the stationary distribution of the process $R(t)$ has density
    \[
    f_0(x)=\frac{\bar{F}(x)}{\mu} = \frac{\lambda}{2}e^{-\lambda x}(1+\lambda x),
    \]
    which is a mixture of an $Exp(\lambda)$ and a $Gamma(2,\lambda)$ distributions (with weights 0.5,0.5). Its cdf is
    \[
    F_0(x) = 1-e^{-\lambda x}\left(1+\frac{\lambda x}{2}\right).
    \]
    We analyze this example in detail.

    By proceeding as in Example \ref{ex:PP2}, we show that the renewal measure is given by
    \[
    U(\text{d}x) = \delta_0(\text{d}x) + \sum_{n=1}^\infty f^{*n}(x)\text{d}x,
    \]
    where $f^{*n}=\frac{\lambda^{2n}}{(2n-1)!}x^{2n-1}e^{-\lambda x}$ (shown via induction). Since we have that
    \[
    \begin{split}
    \sum_{n=1}^\infty \frac{\lambda^{2n}}{(2n-1)!}x^{2n-1}e^{-\lambda x} &= e^{-\lambda x}\lambda\sinh(\lambda x)\\
    &= e^{-\lambda x}\lambda\frac{e^{\lambda x}-e^{-\lambda x}}{2}\\
    &=\frac{\lambda}{2}(1-e^{-2\lambda x}),
    \end{split}
    \]
    the renewal measure is equal to
    \[
    U(\text{d}x) = \delta_0(\text{d}x) + \frac{\lambda}{2}(1-e^{-2\lambda x})\text{d}x.
    \]
    This is straightforwardly decomposed in a bounded measure $U_2=\delta_0$, and a measure $U_1$ with density $u_1(x) = \frac{\lambda}{2}(1-e^{-2\lambda x})$, such that $u_1(x) \rightarrow 1/\mu = \lambda/2$ as $x\rightarrow\infty$ (c.f. Stone's decomposition, Theorem \ref{thm:StoneDec}).

    Now, we compute the law of $R(t)$ by following the procedure outlined in Example \ref{ex:PP3}. $\mathbb{P}(R(t)\leq x)$ satisfies a renewal equation with $z_R^x(t) = F(t+x)-F(t)$, which is bounded. Hence,
    \[
    \begin{split}
    F_R^t(x)=\mathbb{P}(R(t)\leq x) = U\ast z_R^x(t) &= \int_0^t (F(t-y+x)-F(t-y))U(\text{d}y)\\
    &= \int_0^t (F(t-y+x)-F(t-y))\delta_0(\text{d}y) + \\
    &\quad + \int_0^t (F(t-y+x)-F(t-y))\frac{\lambda}{2}(1-e^{-2\lambda y})\text{d}y\\
    &\stackrel{(*)}{=} 1-e^{-\lambda x}\left(1+\frac{\lambda x}{2}\right) - \left( \frac{\lambda x}{2}e^{-\lambda x} \right)e^{-2\lambda t},
    \end{split}
    \]
    where $(*)$ summarizes long and tedious (but easy) calculations. Such law has continuous density
    \[
    f_R^t(x) = \frac{\lambda x}{2}e^{-\lambda x}((1-e^{-2\lambda t}) + \lambda x(1+e^{-2\lambda t})),
    \]
    and we can see that
    \[
    F_R^t(x) \longrightarrow F_0(x), \quad f_R^t(x)\longrightarrow f_0(x), \quad \text{as } t\rightarrow\infty.
    \]

    At this point, let us see how Lemma \ref{lemma:UniformComponentRt} applies to this particular example: we want to show that there exist $A>0,b>0,\delta>0$ such that
    \[
    \forall\,t\geq A,\forall\,x\in(0,b),\quad f_R^t(x)\geq \delta.
    \]
    We have that:
    \begin{itemize}
        \item For $t\geq A>0$: $1-e^{-2\lambda t}\geq 1-e^{-2\lambda A}$ and $1+e^{-2\lambda t}\geq 1$.
        \item For $x\in(0,b)$: $e^{-\lambda x}\geq e^{-\lambda b}$. 
    \end{itemize}
    Hence, for $t\geq A,x\in(0,b)$:
    \[
    f_R^t(x)\geq \frac{\lambda}{2}e^{-\lambda b}((1-e^{-2\lambda A})+\lambda x)\geq \frac{\lambda}{2}e^{-\lambda b}(1-e^{-2\lambda A}) =:\alpha(A,b).
    \]
    By following the strategy of Example \ref{ex:Exp1}, we define $\delta(A,b)\colon\!\!\!= \alpha(A,b)\cdot b$ and we can decompose $f_R^t(x)$ as
    \[
    f_R^t(x) = \frac{\delta(A,b)}{b}\mathbb{I}_{(0,b)}(x) + (1-\delta(A,b))h(x),
    \]
    where $h(x) = \frac{f_R^t(x)-\alpha(A,b)\mathbb{I}_{(0,b)}(x)}{1-\delta(A,b)}$.

    Hence, in this specific case, the laws of the $R(t)$'s have a common $\:\mathcal{U}(0,b)$ component for all $t\geq A$, where $b>0$ and $A$ is such that $\alpha(A,b)>0\Leftrightarrow A>0$.

    Since the process $\{R(t)\}_{t\geq 0}$ is not in stationarity, the coupling construction in the proof of Theorem \ref{thm:CouplingTheorem} is non-trivial. Since the distribution $Gamma(2,\lambda)$ is spread-out, the related renewal process admits coupling. The coupling time $T$ is defined as $T=t_{\sigma+1}+L_{\sigma+1}$ (by using the notation of the proof), where $\sigma\sim Geom(\delta(A,b))$. 

    Let us now analyze the exponential moments of $T$, following the proof of Lemma \ref{lemma:ExpMomentsT}:
    we know that
    \[
    \mathbb{E}[e^{\eta X}] = \lambda^2\int_0^\infty xe^{-(\lambda-\eta)x}\text{d}x<\infty
    \]
    for all $\eta<\lambda$.
    Moreover,
    \[
    \begin{split}
        \mathbb{E}[e^{\eta R(t)}] &= \int_0^\infty e^{\eta x}\frac{\lambda}{2}e^{-\lambda x}((1-e^{-2\lambda t})+\lambda x(1+e^{-2\lambda t}))\text{d}x\\
        &= \frac{\lambda}{2(\lambda-\eta)}(1-e^{-2\lambda t})+ \frac{\lambda^2}{2(\lambda-\eta)^2}(1+e^{-2\lambda t})\\
        &\leq \frac{\lambda}{2(\lambda-\eta)}\left(1+ \frac{2\lambda}{\lambda-\eta}\right) =: c_1.
    \end{split}
    \]
    Furthermore, define $c\colon\!\!\!=2c_1 e^{\eta A}$, and by following the proof of Lemma \ref{lemma:ExpMomentsT}, we conclude that $\mathbb{E}[e^{\varepsilon T}]<\infty$ for $\varepsilon=\eta/p$. $p$ is such that $1/p+1/q=1$ and 
    \[
    c^{1/p}(1-\delta(A,b))^{1/q}<1,
    \]
    which is equivalent to
    \[
    \frac{\lambda^{1/p}}{(\lambda-\eta)^{1/p}}\left(1+\frac{2\lambda}{\lambda-\eta}\right)^{1/p}e^{\eta A/p}\cdot\left(1-\frac{\lambda b}{2}e^{-\lambda b}(1-e^{-2\lambda A})\right)^{1/q}<1.
    \]
    An inspection of this inequality implies that, by choosing reasonable values for $A$ and $b$, $p$ can be chosen smaller than $100$. Such reasonable values for $A,b$ can be, e.g. $A=2\mu=4/\lambda$, which is hopefully big enough for moving past the initial delayed phase (see Remark \ref{rmk:AlsoForDelayed}) and $b=1/\lambda$.

    As a conclusion, notice that, following the proof, we found $\varepsilon<\eta<\lambda$. By Theorem \ref{thm:exponentialRateRt}, this means that $\|\mathbb{P}(R(t)\in\cdot)-F_0(\cdot)\| = O(e^{-\varepsilon t})$. However, as was specified in Remark \ref{rmk:couplingRate}, this might not be the best rate possible. Indeed, by applying \cite[Lemma 2.1]{Tsybakov:NonparametricEstimation} $(*)$
    \[
    \begin{split}
        \|\mathbb{P}(R(t)\in\cdot)-F_0(\cdot)\| &\stackrel{(*)}{=} \frac{1}{2}\int_0^\infty |f_R^t(x)-f_0(x)|\text{d}x\\
        &= \frac{\lambda}{4}e^{-2\lambda t}\int_0^\infty e^{-\lambda x}|\lambda x-1|\text{d}x\\
        &=\frac{\lambda^2}{2e}e^{-2\lambda t}\\
        &=O(e^{-2\lambda t}).
    \end{split}
    \]
    It is clear that the convergence is actually at the higher rate $2\lambda>\lambda$.
\end{example}

\section{A Deeper Analysis of the Coupling Proof}
\label{sec:CouplingProofDetailed}

We conclude the Chapter with an in-depth analysis of the proof of Theorem \ref{thm:CouplingTheorem}, as it is based on an interesting construction. In particular, we are mainly interested in one of the two directions: if a renewal process has a spread-out interarrival distribution, then it admits coupling. As a convention, a renewal process admits coupling if the associated forward recurrence time process does. Hence, the proof is centered around the process $\{R(t)\}_{t\geq 0}$. Let us recall the general construction (c.f. proof of Theorem \ref{thm:CouplingTheorem}).

We define two renewal processes (and the associated forward recurrence time processes) with the same spread-out interarrival distribution: $\{T_n\}_{n\in\mathbb{N}}$ ($\{R(t)\}_{t\geq 0}$) is zero-delayed and $\{T'_n\}_{n\in\mathbb{N}}$ ($\{R'(t)\}_{t\geq 0}$) is stationary. To this end, we construct an increasing sequence of times $\{t_k\}_k$, and define the values of $R(t_k)$ and $R'(t_k)$ for all $k$. The renewal processes are then assigned the conditional distribution of each renewal process given the previously defined value of the associated forward recurrence time process.

The sequence $\{t_k\}_k$ is chosen in such a way that, for all $k\geq 1$, the laws $F_{R(t_k)}$ and $F_{R'(t_k)}$ of $R(t_k)$ and $R'(t_k)$ depend both on $R(t_{k-1})$ and on $R'(t_{k-1})$, and satisfy the following condition, for given $\delta\in(0,1),\,b>0$:
\begin{align}
    F_{R(t_k)}(\cdot) &\geq \delta \, Unif_{(0,b)}(\cdot)\label{eq:Minorization1}\\
    F_{R'(t_k)}(\cdot) &\geq \delta \, Unif_{(0,b)}(\cdot)\label{eq:Minorization2},
\end{align}
where $Unif_{(0,b)}$ represents the uniform measure on $(0,b)$. These conditions are a special case of a so called \emph{minorization condition}, used in \cite{Nummelin:Splitting} to construct a \emph{split chain} (this technique is explained in detail in Appendix \ref{chap:NummelinSplitting}). It could be tempting to try and apply the same reasoning to our situation. However, the split chain technique is used on general state space recurrent Markov Chains to create a sequence of regeneration times, such that after each one of them the chain restarts from scratch independently of all the past, with a given distribution. In our case, the identification of regeneration times is not of interest, but instead we would like to determine the first time $R(t_k)=R'(t_k)$ occurs, i.e. the coupling time. This condition is clearly different from the regenerative one, as it does not require future steps to be independent of the past. Fortunately, we can still exploit conditions \eqref{eq:Minorization1}, \eqref{eq:Minorization2} to identify the coupling time.

First, let us define an inhomogeneous Markov Chain $\{(X_k,X'_k)\}_k$ on the measurable space $\left([0,\infty)^2, \mathcal{B}\left([0,\infty)^2\right)\right)$:
\[
(X_k,X'_k)\colon \!\!\! = (R(t_k),R'(t_k)),\quad k\geq 1.
\]
Call its transition kernel $P_k(x,A)$, for $x\in [0,\infty)^2,\,A\in\mathcal{B}\left([0,\infty)^2\right)$. 

Thanks to the minorization conditions on the marginal laws, it is possible to construct a similar one for the joint chain. There are many choices that can be made to preserve the conditions on the marginals, and here we follow Asmussen's construction \cite{Asmussen:Applied} (one can see, for example, \cite{Thorisson:RegProcCoupling} for a different argument, which is, however, based on a slightly wider definition of coupling). In the rest we will denote the uniform measure on $(0,b)$ by $\mathcal{U}_{(0,b)}$. We have:
\begin{equation}
    P_k(x,\cdot) \geq \mathcal{U}_{(0,b)}^{diag}(\cdot)\label{eq:CouplingMinorization},
\end{equation}
where
\begin{equation}
    \mathcal{U}_{(0,b)}^{diag}(A) = \mathcal{U}_{(0,b)}\left(\left\{ x \in [0,\infty): (x,x)\in A\right\}\right),\quad A \in \mathcal{B}\left([0,\infty)^2\right).
\end{equation}

Condition \eqref{eq:CouplingMinorization} preserves the conditions on the marginal, and enables us to write the transition kernel as 
\[
P_k(x,A) = \delta\,\mathcal{U}_{(0,b)}^{diag}(A) + (1-\delta)H_k(x,A),
\]
where $H_k(x,A) = (P_k(x,A) - \delta\,\mathcal{U}_{(0,b)}^{diag}(A))/(1-\delta)$. This means that, at each step, there is a probability $\delta$ of having $X_k=X'_k$ and distributed according a $\mathcal{U}_{(0,b)}$ distribution.

Formally, take a sequence $\{\xi_k\}_k$ of i.i.d. $Ber(\delta)$ random variables, and enlarge the Markov Chain by defining $\{(X_k,X'_k,\xi_k)\}_k$ on the enlarged state space $\left( [0,\infty)^2\times\{0,1\}, \mathcal{B}\left( [0,\infty)^2\times\{0,1\} \right) \right)$. The enlarged chain has the following transition kernel:
\[
\hat{P}_k(x\times i, A\times j) = \begin{cases}
    \delta \,\mathcal{U}_{(0,b)}^{diag}(A),\quad &j=1\\
    (1-\delta) H_k(x,A), \quad &j=0
\end{cases}
\]
for $x\in[0,\infty)^2,\: A\in \mathcal{B}\left([0,\infty)^2\right),\:i,j\in\{0,1\}$.

It is clear that the coupling time $T$ is the first $t_k$ for which $\xi_k=1$, i.e. $T=t_\sigma$, where $\sigma=\inf\{k:\xi_k=1\}$. Since the $\xi_k$'s are i.i.d. Bernoulli, $\sigma$ will follow a geometric distribution of parameter $\delta$.

\chapter{Regenerative Rejection Sampling}
\label{chap:5}
This chapter is devoted to the presentation of a novel method, recently introduced in \cite{Botev:MachineLearning}, to sample from a given probability distribution, called \emph{Regenerative Rejection Sampling}. This method expands on the concept of \emph{Rejection Sampling}, and for this reason, we start our study by reviewing this method.

\section{Rejection Sampling}
The \emph{Rejection Sampling method}, introduced by John Von Neumann in 1951 \cite{VonNeumann:RS}, is one of the most common Monte Carlo algorithms for sampling from a given probability distribution. The method is \emph{universal}, meaning that one could, in principle, sample from any target density $f(x)$, with the aid of a proposal density $g(x)$ \cite{Martino:Sampling}. However, the method needs some assumptions to work:
\begin{itemize}
    \item The target density is known up to a multiplicative constant, i.e. only $f_\propto(x) = M f(x)$ is known, where $M=\int f_\propto(x)\text{d}x$.
    \item The likelihood ratio $f_\propto(x)/g(x)$ must be upper-bounded by a \emph{known} constant $C \geq M$ for all $x$.
\end{itemize}

The first assumption does not look restrictive, since it is usually the case for most practical applications. On the other hand, the second one shows an important limitation of the method. In cases where this constant $C$ is cumbersome to calculate analytically, Rejection Sampling cannot be used as a sampling procedure. If one is only able to compute a loose upper-bound of the likelihood ratio, the sampling procedure will be highly inefficient, since the proposed samples would be rejected with high probability (we explain this more in detail in Remark \ref{rmk:ProbAccRS}).

The algorithm is outlined as follows:
\algrenewcommand\algorithmicrequire{\textbf{Input:}}
\algrenewcommand\algorithmicensure{\textbf{Output:}}
\begin{algorithm}
\caption{Rejection Sampling}
\begin{algorithmic}[1]
\Require{Proposal pdf $g$, constant $C$ such that $f_\propto(x)/g(x)\leq C$ for all $x$.}
\Ensure{Random variable $Y$ distributed according to target pdf $f$}
\Repeat
    \State Simulate $X \sim g$ and $U\sim\mathcal{U}(0,1)$ independently.
    \State $W \gets f_\propto(X)/g(X)$
\Until{$W\geq CU$}
\State \textbf{return} $Y=X$
\end{algorithmic}
\end{algorithm}

We prove that the algorithm generates a random variable with the correct desired distribution $f$:
\begin{proposition}
    The random variable $Y$ generated by the Rejection Sampling algorithm has the correct pdf $f(x)=M^{-1} f_\propto (x)$.
\end{proposition}
\begin{proof}
    First, the distribution of $Y$ is equal to the distribution of $X$, conditional on $W\geq CU$. Hence
    \[
    \mathbb{P}(Y\leq x) = \mathbb{P}\left( X\leq x \,\Bigg|\, CU\leq\frac{f_\propto(X)}{g(X)}\right) = \frac{\mathbb{P}\left( X\leq x,U\leq\frac{f_\propto(X)}{Cg(X)} \right)}{\mathbb{P}\left( U\leq\frac{f_\propto(X)}{Cg(X)} \right)}.
    \]
    Hence, we have:
    \[
    \begin{split}
        \mathbb{P}\left( X\leq x,U\leq\frac{f_\propto(X)}{Cg(X)} \right) &= \int_{-\infty}^x\int_0^{f_\propto(y)/Cg(y)} \text{d}u \ g(y)\text{d}y\\
        &= \int_{-\infty}^x\frac{f_\propto(y)}{Cg(y)} g(y) \text{d}y\\
        &=\frac{1}{C}\int_{-\infty}^x f_\propto(y)\text{d}y,
    \end{split}
    \]
    and
    \[
    \begin{split}
        \mathbb{P}\left( U\leq\frac{f_\propto(X)}{Cg(X)} \right) &= \int_\mathbb{R}\int_0^{f_\propto(y)/Cg(y)} \text{d}u \ g(y)\text{d}y\\
        &= \int_\mathbb{R}\frac{f_\propto(y)}{Cg(y)} g(y) \text{d}y\\
        &=\frac{1}{C}\int_\mathbb{R} f_\propto(y)\text{d}y.
    \end{split}
    \]
    To conclude,
    \[
    \mathbb{P}(Y\leq x) = \frac{\int_{-\infty}^x f_\propto(y)\text{d}y}{\int_\mathbb{R} f_\propto(y)\text{d}y} = \int_{-\infty}^x f(y)\text{d}y.
    \]
\end{proof}

The algorithm has an interesting geometric/graphical interpretation \cite{Martino:Sampling}: first, uniform points $(X,U)$ are sampled in the region $\mathcal{A}_{Cg} = \{(x,u)\in \mathbb{R}^2: 0\leq u\leq Cg(x)\}$, and then only the samples $(Y,U)$ that lie in the region $\mathcal{A}_{f_\propto} = \{(y,u)\in\mathbb{R}^2: 0\leq u \leq f_\propto(y)\}$ are retained. These points have the correct distribution $f(x) = f_\propto(x)/\int f_\propto(x)$. This is shown in the following lemma, based on \cite[\emph{Fundamental theorem of simulation}]{RobertCasella:MC} and \cite{Devroye:RVG}.

\begin{lemma}
\label{lemma:geomInterp}
    Consider a non-negative integrable function $\tilde{h}\colon \mathbb{R}\rightarrow\mathbb{R}_+$, such that $\int_\mathbb{R} \tilde{h}\,\text{\emph{d}}y\neq 0$. Consider also the region $\mathcal{A}_{\tilde{h}}=\{(y,u)\in\mathbb{R}^2:0\leq u\leq \tilde{h}(y)\}$, and the normalized probability density function $h(y) = \tilde{h}(y)/\int_\mathbb{R} \tilde{h}(y)\text{\emph{d}}y$. Then, a pair of random variables $(Y,U)$ is uniformly distributed in $\mathcal{A}_{\tilde{h}}$ if and only if $Y\sim h$ and $U|Y\sim\mathcal{U}_{[0,\tilde{h}(Y)]}$.
\end{lemma}
The proof of this result can be seen in Appendix \ref{sec:ProofGeometricInterpretationRS}.

Usually, the algorithm will reject some samples. By following the geometric interpretation, the rejected points are the ones that lie in the region $\mathcal{A}_{Cg-f_\propto}=\{(x,u)\in \mathbb{R}^2:f_\propto(x)\leq u \leq Cg(x)\}$. We can easily characterize the probability that a sample will be accepted/rejected.
\begin{remark}
\label{rmk:ProbAccRS}
    The probability of acceptance of the algorithm is
    \[
    \mathbb{P}\left( U\leq\frac{f_\propto(X)}{Cg(X)} \right) = \frac{M}{C}.
    \]
    Since each point trial $(X,U)$ is generated independently, the number of trials that we have to generate before getting a success $(X,U)$ has $Geom(M/C)$ distribution, which has mean equal to $C/M$. Thus, $C$ should be as close as possible to $M$ for the algorithm to be efficient, which implies that the choice of the auxiliary distribution $g$ can significantly change the performance of the method.
\end{remark}

To show the functioning of the algorithm and its graphical representation, we provide an example.
\begin{example}
\label{ex:RSGammaExp}
    Suppose we want to generate a random variable from the $Gamma(\alpha,\lambda)$ distribution, with pdf
    \[
    f(x) = \frac{\lambda^\alpha x^{\alpha-1}e^{-\lambda x}}{\Gamma(\alpha)},\quad x\geq 0.
    \]
    For this example, consider the case $\alpha=2,\lambda=1$. Such pdf lies under the graph of $Cg(x)$, where $C=1.6$ and $g(x)=0.4e^{-0.4x},\:x\geq 0$. Thus, we can use as proposal distribution an $Exp(0.4)$.
    
    Figure \ref{fig:RejectionSampling} shows $250$ independent samples from the $Gamma(2,1)$ distribution (in blue), which we can clearly see are lying under the graph of 
    \[
    f_\propto(x)=x^{\alpha-1}e^{-\lambda x}.
    \]In orange we see all the rejected samples, that lie in the region between the graphs of $f_\propto$ and $Cg$.
    For this particular example, the theoretical probability of acceptance is $M/C=(\Gamma(\alpha)/\lambda^\alpha)/C = 1/1.6=0.625$, while the empirical acceptance rate was $0.611$.
    \begin{figure}[h]
        \centering
        \includegraphics[width=0.55\linewidth]{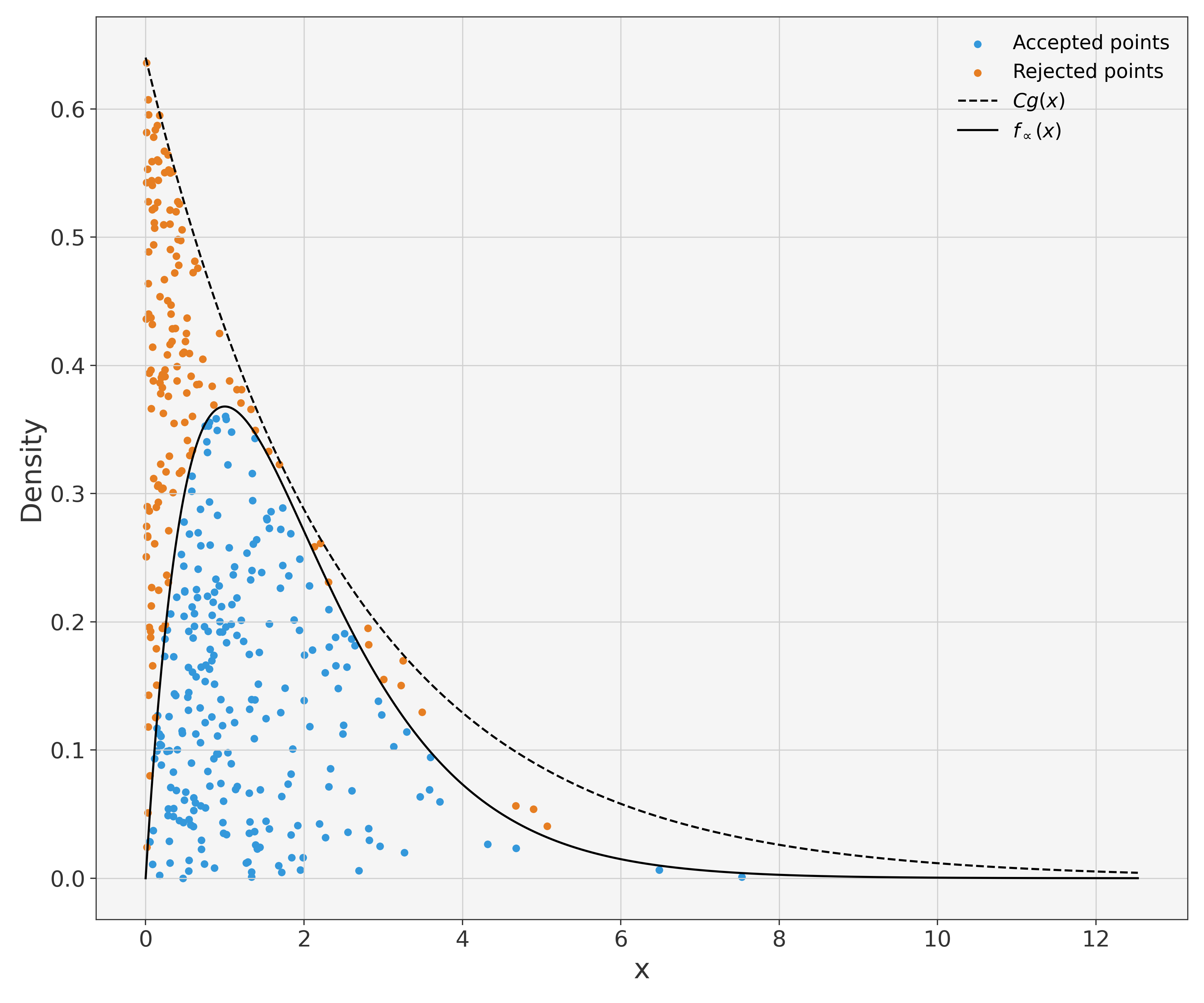}
        \caption{Rejection Sampling}
        \label{fig:RejectionSampling}
    \end{figure}
\end{example}

Even if the previous discussion was formulated for one-dimensional distributions, the Rejection Sampling method can be used also in the multi-dimensional case. However, it is known that the method suffers from the curse of dimensionality \cite{Kroese:HandbookMC}.

\section{Regenerative Rejection Sampling}
In the cases where the constant $C$ is too cumbersome to compute, or does not even exist (i.e. the ratio $f_\propto/g$ is unbounded), we cannot use the Rejection Sampling method to obtain a sample from the target distribution. Instead, one could resort to using the MCMC Independence Sampler, with the same proposal distribution, because this method doesn't require knowledge, or existence, of the bounding constant $C$ to work. However, as we shall see later on, when $C$ does not exist, the MCMC Independence Sampler cannot converge geometrically fast in total variation.

In this section we propose an alternative and simpler method, called \emph{Regenerative Rejection Sampling} (RRS). It expands on the concept of Rejection Sampling and can be used as an alternative to Markov Chain Monte Carlo methods. It is based on the same setting of the Rejection Sampling and Independence Sampler methods: the goal is to generate a sample from a target distribution with density $f(x)$, known only up to a multiplicative constant (we know $f_\propto(x)=M f(x)$), and we only rely on samples from a proposal distribution with density $g(x)$. In this instance, however, we do not need to know the value of the upper bound  $\sup_x f_\propto(x)/g(x)$, and do not even require that it is finite.

Unfortunately, The RRS method does not return exact samples from the target distribution, but only approximate ones. This is due to its iterative structure based on the construction of a regenerative process. In this sense, it is similar to existing MCMC methods. 

We provide the pseudo-code of the Regenerative Rejection Sampling algorithm in Algorithm \ref{algo:RRS}, and a memory-efficient version in Algorithm \ref{algo:MERRS}. Even if the pseudo-codes are specified for targeting a one-dimensional distribution, everything works identically for the multidimensional case.
\begin{algorithm}
\caption{Regenerative Rejection Sampling}
\label{algo:RRS}
\begin{algorithmic}[1]
\Require{Proposal pdf $g$, time $t$}
\Ensure{Random variable $X$ approximately distributed according to target pdf $f$}
\State $N \gets 0$
\Repeat
    \State $N \gets N+1$
    \State Simulate $X_N \sim g$
    \State $W_N \gets f_\propto(X_N)/g(X_N)$
\Until{$W_1 + \dots + W_N > t$}
\State \textbf{return} $X \gets X_N$
\end{algorithmic}
\end{algorithm}

\begin{algorithm}
\caption{Memory-Efficient Regenerative Rejection Sampling}
\label{algo:MERRS}
\begin{algorithmic}[1]
\Require{Proposal pdf $g$, time $t$}
\Ensure{Random variable $X$ approximately distributed according to target pdf $f$}
\State $S \gets 0$
\Repeat
    \State Simulate $X \sim g$
    \State $W \gets f_\propto(X)/g(X)$
    \State $S \gets S+W$
\Until{$S > t$}
\State \textbf{return} $X$
\end{algorithmic}
\end{algorithm}

The method might be most beneficial in situations where $C$ is difficult to compute or even infinite. On the downside, since it only returns random variables approximately distributed as the target distribution, we may need to run the algorithm for a long time to obtain a high quality sample, potentially resulting in a large number of simulations from the proposal $g$.

In Theorem \ref{thm:RRSConvergence}, we prove the convergence of the algorithm to the desired target distribution, but first let us point out the construction of the underlying regenerative process.
\begin{remark}
\label{rmk:RRSRegProcConstr}
    The random variables $X_1,X_2,\dots$ are an i.i.d. sequence drawn from the proposal distribution with density $g$. Now consider the random variables $W_i=f_\propto (X_i)/g(X_i), \: i=1,2,\dots$. Since the $X_i$'s are i.i.d., so are the $W_i$'s.

    Assuming that $\mathbb P[W>0]=1$, Algorithms \ref{algo:RRS} and \ref{algo:MERRS} construct a zero-delayed regenerative process $\{Y(t)\}_{t\geq0}$ with regeneration times $T_0=0,T_n = W_1+\dots+W_n, \: n\geq 1$, which takes the constant value $X_i$ on the cycle $[T_{i-1},T_i)$. Hence, the $W_i$'s represent the cycle lengths of the regenerative process.
\end{remark}

The algorithm converges to the correct target distribution, as is shown in Theorem \ref{thm:RRSConvergence}. Its proof is based on the theory of regenerative processes, see Chapter \ref{chap:3}.
\begin{theorem}[Convergence of Regenerative Rejection Sampling]
\label{thm:RRSConvergence}
    If the distribution of the $W_i$'s has finite mean $\mu<\infty$ and is spread out, then the output of Algorithm \ref{algo:RRS} converges in total variation sense to its limiting distribution with pdf $f(x)$ as $t\rightarrow\infty$.
\end{theorem}
\begin{proof}
    We have shown the construction of the underlying regenerative process in Remark \ref{rmk:RRSRegProcConstr}. Hence, by Theorem \ref{thm:RegProcTVConv}, the limiting distribution of the regenerative process is given by, for any function $h$
    \[
    \begin{split}
    \mathbb{E}_e[h(Y(t))] = \frac{1}{\mu} \mathbb{E}_0\left[ \int_0^{W_1} h(Y(s))\text{d}s \right] &= \frac{\mathbb{E}_g\left[h(X_1)\int_0^{W_1} 1 \,\text{d}s\right]}{\mathbb{E}_g[W_1]}\\
    &= \frac{\mathbb{E}_g[h(X_1)W_1]}{\mathbb{E}_g[W_1]}\\
    &= \frac{\int h(x) \frac{f_\propto(x)}{g(x)}g(x)\,\text{d}x}{\int \frac{f_\propto(x)}{g(x)}g(x) \,\text{d}x}\\
    &= \int h(x)f(x) \, \text{d}x\\
    &= \mathbb{E}_f[h(Y)],
    \end{split}
    \]
    where $Y\sim f$. Hence, $Y(t)$ converges in total variation sense to the distribution with density $f(x)$. 
\end{proof}

Let us now show an application of this method, also seen in \cite{Botev:MachineLearning}, to the setting of Example \ref{ex:RSGammaExp}.
\begin{example}
\label{ex:RRSGammaExp1}
    As in the previous example, we aim to simulate from the target pdf $f_\propto(x)=f(x) = x e^{-x},\: x\in[0,\infty)$, i.e. a Gamma(2,1) distribution. Suppose that the proposal pdf is $g(x) = e^{-x},\:x\in[0,\infty)$. It is clear that the likelihood ratio $f(x)/g(x)$ is \emph{not} bounded, hence we cannot use the classical Rejection Sampling method. But we can use the regenerative rejection sampling procedure to get approximate samples.
    \begin{figure}[h]
        \centering
        \includegraphics[width=0.55\linewidth]{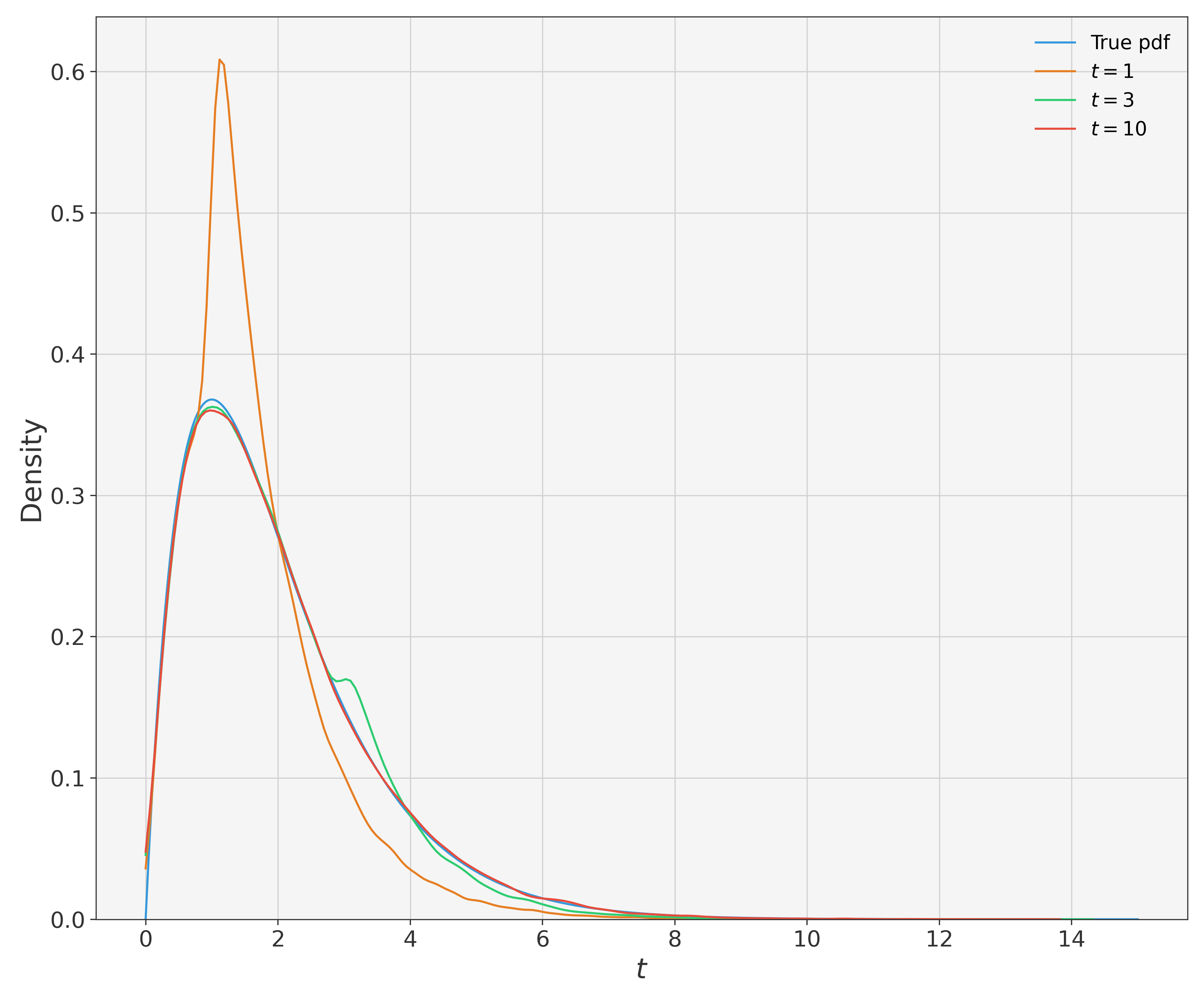}
        \caption{Regenerative Rejection Sampling}
        \label{fig:RRS}
    \end{figure}
    Figure \ref{fig:RRS} shows the performance of the Regenerative Rejection Sampling method. For each value of $t\in\{1,3,10\}$, we naively ran the RRS method $10^6$ times independently, and then plotted the KDE of the samples (using functions from the python package \texttt{seaborn}). It is straightforward to see that, as $t$ increases, the difference between the KDE and the true pdf becomes negligible (for $t=10$ the two are almost indistinguishable). Hence, we clearly appreciate the convergence to the correct target distribution.
\end{example}

\subsection{Rate of convergence of the Regenerative Rejection Sampling method}
The goal of this Section is to provide a thorough explanation of the convergence behavior of the Regenerative Rejection Sampling method, as is done for many Markov Chain Monte Carlo algorithms (see, for example, \cite{MergensenTweedie:MHConvergence,Roberts:Langevin}). In our case, we have at our disposal the full power of the theory of regenerative processes, and the related results on their convergence \cite{Asmussen:Applied,Lindvall:Coupling}. 

Since the RRS method only involves the construction of the underlying regenerative process, the results of Chapter \ref{chap:4} are directly applicable to this situation. Hence, we can formulate the following Theorem concerning the Regenerative Rejection Sampling method.
\begin{theorem}[Rate of Convergence of Regenerative Rejection Sampling method]
\label{thm:RateConvergenceRRS}
    If the cycle length distribution of the underlying regenerative process (i.e. the distribution of $W$) is spread-out, has finite mean and satisfies, for some $\eta>0$
    \begin{equation}
    \label{eq:RRSExpMoments}
    \mathbb{E}_g[e^{\eta W}] = \int_0^\infty e^{\eta f_\propto(x)/g(x)} g(x) \,\text{\emph{d}}x <\infty,
    \end{equation}
    then the RRS method has exponential convergence
    \[
    \|\mathbb{P}(Y(t)\in \cdot)-\mathbb{P}(Y(\infty)\in\cdot)\|=O(e^{-\varepsilon t}),
    \]
    for some $\varepsilon>0$.

    If, instead of condition \eqref{eq:RRSExpMoments} the cycle length distribution only verifies, for some $p>0$
    \begin{equation}
        \label{eq:RRSPolyMoments}
        \mathbb{E}_g[W^{p+1}] = \int_0^\infty \left( \frac{f_\propto(x)}{g(x)} \right)^{p+1} g(x)\,\text{\emph{d}}x<\infty,
    \end{equation}
    then the RRS method has polynomial convergence
    \[
    \|\mathbb{P}(Y(t)\in\cdot)-\mathbb{P}(Y(\infty)\in\cdot)\| = o(t^{-p}).
    \]
\end{theorem}

Theorem \ref{thm:RateConvergenceRRS} summarizes the assumptions needed for exponential and polynomial convergence of the Regenerative Rejection Sampling method. Since the cumulative distribution function of the $W$'s is given by
\[
F_W(y)=\mathbb{P}(W\leq y) = \mathbb{P}\left( \frac{f_\propto(X)}{g(X)}\leq y \right) = \int_0^\infty \mathbb{I}\left( \frac{f_\propto(x)}{g(x)} \leq y \right) g(x) \,\text{d}x,
\]
we do not have access to a general closed form, apart from restricting ourselves to cases in which $w(x)=f_\propto(x)/g(x)$ is monotonic (or bijective and differentiable if $x\in\mathbb{R}^n,\:n\geq2$). In such situations we can apply the renowned \emph{change of variables formula} and obtain a closed form for the density of the distribution of the $W$'s.\\
But this does not represent the general case. Thus, the theoretical convergence analysis has to be mainly done on a case-by-case basis, to be able to understand the rate at which the method converges.

At this point, we present two theoretical examples that thoroughly show the convergence behaviour of the algorithm under different assumptions.
\begin{example}
\label{ex:RRSGammaExp2}
    Assume we are in the same setting of Example \ref{ex:RRSGammaExp1}. In this case,
    \[
    \frac{f_\propto(X_i)}{g(X_i)} = \frac{X_i e^{-X_i}}{e^{-X_i}} = X_i,\quad i=1,2,\dots.
    \]
    Hence, the cycle length distribution is equal to the proposal distribution $Exp(1)$. Since this distribution is light-tailed, we expect to see an exponential rate of convergence of the method. Let's check this.

    By employing the notation of the theory of renewal processes, we define the regenerative process as 
    \[
    Y(t) = X_{N(t)}, \quad t\geq 0.
    \]
    Given the particularly easy setting of $X_i\sim Exp(1)$, we can actually compute the pdf of the distribution of $X_{N(t)}$. By a classical renewal argument, $Z^y(t) = \mathbb{P}(X_{N(t)}\leq y)$ satisfies a renewal equation with $z^y(t) = \mathbb{P}(X_{N(t)}\leq y,X_1>t) = \mathbb{P}(X_1\leq y,X_1>t) = (F(y)-F(t))\mathbb{I}(y>t)$. Then (since $U(\text{\emph{d}}x)=\delta_0(\text{\emph{d}}x)+\text{\emph{d}}x$, by Example \ref{ex:PP2}), we have
    \[
    \begin{split}
        Z^y(t) = U\ast z^y(t) &= \int_0^t z^y(t-x)U(\text{d}x)\\
        &= \int_0^t(F(y)-F(t-x))\mathbb{I}(y>t-x)\delta_0(\text{d}x)\\
        &\quad + \int_0^t(F(y)-F(t-x))\mathbb{I}(y>t-x)\text{d}x\\
        &= (F(y)-F(t))\mathbb{I}(y>t) + \int_0^t (F(y)-F(t-x))\mathbb{I}(x>t-y)\text{d}x\\
        &= \mathbb{I}(y>t)(e^{-t}-e^{-y}) + \mathbb{I}(y>t)\int_0^t (e^{-(t-x)}-e^{-y})\text{d}x\\
        &\quad + \mathbb{I}(y\leq t) \int_{t-y}^t (e^{-(t-x)}-e^{-y})\text{d}x\\
        &= \begin{cases}
            1-(1+t)e^{-y}, \quad &y>t\\
            1-(1+y)e^{-y}, \quad&y\leq t
        \end{cases}
    \end{split}
    \]
    By piecewise differentiation of this continuous cdf, we obtain the pdf $f_{X_{N(t)}}$, that has a discontinuity in $y=t$ (note that the discontinuity is clearly seen in Figure \ref{fig:RRS} for low values of $t$).
    \[
    f_{X_{N(t)}} (y)= \begin{cases}
        ye^{-y}, \quad &y\leq t\\
        (1+t)e^{-y},\quad &y>t
    \end{cases}
    \]
    For this specific case we see that, when $y\leq t$, the pdf of the regenerative process $Y(t)$ corresponds to the target $Gamma(2,1)$ pdf.

    To conclude, apply \cite[Lemma 2.1]{Tsybakov:NonparametricEstimation} $(*)$ to get
    
    \begin{equation}
    \label{eq:RRSGammaExpTV}
    \begin{split}
    \|\mathbb{P}(X_{N(t)}\in \cdot)-\mathbb{P}(Y(\infty)\in \cdot)\| &\stackrel{(*)}{=} \frac{1}{2} \int_0^\infty |f_{X_{N(t)}}(x)-f(x)|\text{\emph{d}}x\\
    &= \frac{1}{2}\int_t^{\infty} e^{-x}(x-(1+t))\text{\emph{d}}x\\
    &= e^{-t}(1+t)\\
    &= O(e^{-t}).
    \end{split}
    \end{equation}
    Thus, we have shown that the method converges exponentially, and, by using the closed form for the total variation distance \eqref{eq:RRSGammaExpTV}, we have a practical way of determining how large $t$ has to be to attain a certain error bound.
\end{example}




To conclude the Section, we present a generalization of Examples \ref{ex:RRSGammaExp1} and \ref{ex:RRSGammaExp2}, and analyze the rate of convergence of the Regenerative Rejection Sampling method.
\begin{example}
    Let us generalize the $Gamma-Exp$ framework of the previous examples. Suppose that we want to sample from a $Gamma(\alpha,1)$ target distribution, with density $f(x) = x^{\alpha-1} e^{-x}/\Gamma(\alpha),\: x\in[0,\infty)$ (assume we only know $f_\propto(x) = x^{\alpha-1} e^{-x}$), and we choose an $Exp(\beta)$ proposal distribution, with density $g(x)=\beta e^{-\beta x},\:x\in[0,\infty)$.

    Hence, for $X\sim g$, we have $W=X^{\alpha -1}e^{-(1-\beta)x}$. To obtain exponential convergence, we need, for some $\eta>0$
    \[
    \mathbb{E}_g\left[e^{\eta X^{\alpha-1}e^{-(1-\beta)X}}\right] <\infty.
    \]
    For all $\beta < 1$, $x^{\alpha-1}e^{-(1-\beta)x}\leq C< \infty$ for all $x$ and all $\alpha$. Hence, we trivially obtain exponential convergence, since $\mathbb{E}_g[e^{\eta W}]<\infty$ for all $\eta>0$.

    When $\beta = 1$, we obtain
    \[
    \mathbb{E}_g[e^{\eta W}] = \mathbb{E}_g[e^{\eta X^{\alpha -1}}] = \int_0^\infty e^{-x+\eta x^{\alpha-1}} \text{\emph{d}}x,
    \]
    which is finite (for some $\eta>0$) only when $\alpha \leq 2$.

    If $\beta >1$, we can only aim to get polynomial convergence. We have
    \[
    \begin{split}
    \mathbb{E}_g[W^p] &= \mathbb{E}_g\left[X^{p(\alpha-1)} e^{-p(1-\beta)X}\right]\\
    &= \int_0^\infty x^{p(\alpha-1)}e^{-[p(1-\beta)+\beta]x} \text{\emph{d}}x,
    \end{split}
    \]
    which is finite if $p(1-\beta)+\beta>0 \Leftrightarrow\beta<1+\frac{1}{p-1}$. For instance, if we choose $\beta=1+\frac{1}{9}$, then $\mathbb{E}_g[W^p]$ is finite for all $p\leq9$. Hence, by Theorem \ref{thm:RRSConvergence}, the Regenerative Rejection Sampling method converges polynomially with rate $o(t^{-8})$.

    Unfortunately, if we choose $\beta \geq 2$, the method is not even guaranteed to converge at polynomial rate. 
\end{example}

\section{Comparison with the Independence Sampler and Limitations of the Method}
We conclude the Chapter by comparing the Regenerative Rejection Sampling method with the popular Independent Metropolis Hastings Algorithm, which we will call \emph{Independence Sampler} (IS). Additionally, we explain which are the intrinsic limitations of the RRS method. Let us start with a description of the IS algorithm.

The Independent Metropolis Hastings algorithm is a special case of the general Metropolis Hastings algorithm introduced in \cite{Hastings:MCMC}. It wasn't until with \cite{Tierney:IndependenceSampler} that it was considered as a distinct method with its specific theoretical properties.

Its peculiarity is that the proposal distribution does not change at each step, i.e. it is \emph{independent} of the current state of the chain. Hence, if the target density is $f_\propto(x)$ and the proposal density is $q(x,y) = g(y)$, the acceptance probability takes the form (here $x$ is the current state of the chain and $y$ is the proposed state)
\[
\alpha(x,y) = \min\left\{1,\frac{w(y)}{w(x)}\right\},
\]
where $w(x) = f_\propto(x)/g(x)$ is the likelihood ratio, or the importance weight that would be used in importance sampling if we were to sample exactly from $\pi$ using the proposal distribution $f$ \cite{MergensenTweedie:MHConvergence}.

Due to its particular structure, we notice similarities with the Regenerative Rejection Sampling method. In both algorithms we always sample from the same proposal distribution, and construct an underlying stochastic process, regenerative in our case, and a Markov Chain for the Independence Sampler. For this reason, we compare, from a theoretical point of view, the convergence properties of the two algorithms.

The main result on the rate of convergence of the Independence Sampler is due to \cite{MergensenTweedie:MHConvergence}:
\begin{theorem}[Rate of Convergence of Independence Sampler]
    If we have that, for $C>0$ 
    \begin{equation}
    \label{eq:IndepSamplerBounded}
    w^{-1}(x) = \frac{g(x)}{f_\propto(x)} \geq \frac{1}{C}, \quad \text{for all $x$}
    \end{equation}
    then the Independence Sampler converges geometrically, i.e.
    \begin{equation}
        \|P^n(x,\cdot)-\pi(\cdot)\|\leq\left(1-\frac{1}{C}\right)^n,
    \end{equation}
    where $P^n$ is the $n$-step transition kernel of the underlying Markov Chain.

    Conversely, if condition \eqref{eq:IndepSamplerBounded} is not verified, the method does not converge geometrically.
\end{theorem}

It is clear that condition \eqref{eq:IndepSamplerBounded} is equivalent to 
\[
w(x) = \frac{f_\propto(x)}{g(x)} \leq C,\quad \text{almost everywhere}.
\]
which we recognize as the boundedness condition of the likelihood ratio in the Rejection Sampling context.

Now, assume that \eqref{eq:IndepSamplerBounded} is satisfied, and consider the Regenerative Rejection Sampling setting. Then:
\[
\mathbb{E}[e^{\eta W}] = \mathbb{E}\left[e^{\eta (f_\propto(X)/g(X))}\right] \leq \mathbb{E}\left[e^{\eta C}\right] < \infty,
\]
for all $\eta>0$. Hence, due to Theorem \ref{thm:RRSConvergence}, the underlying regenerative process converges to its stationary distribution $f$ at the exponential rate of $O(e^{-\varepsilon t})$, for some $\varepsilon > 0$.
In general, we do not have a practical way to calculate $\varepsilon$, so it is infeasible to measure if $1-1/C$ is larger or smaller than $e^{-\varepsilon}$. Hence, we cannot know if the Independence Sampler converges at a higher rate compared to the Regenerative Rejection Sampling method. Nonetheless, they are both in exponential regime, if the likelihood ratio is bounded.\\ 
Thus, our method converges exponentially whenever the Independence Sampler does (at a possibly lower rate, however). In addition, from Theorem \ref{thm:RRSConvergence} we know that we only need the cycle length distribution (i.e. the distribution of $f_\propto(X)/g(X),\: X\sim g$) to have exponential moments for the Regenerative Rejection Sampling algorithm to converge exponentially. This implies that the RRS method can have geometric convergence even when the likelihood ratio $f_\propto/g$ is unbounded, as shown in Example \ref{ex:RRSGammaExp2}. 
The comparison of the convergence properties of the two methods is summarized in Table \ref{tab:ConvergenceComparison}.
\begin{table}[htbp]
\centering
\caption{Convergence of RRS vs Indpendence Sampler}
\label{tab:ConvergenceComparison}
\begin{tabular}{p{8cm} p{3cm} p{3cm}}
\toprule
 & \textit{RRS} & \textit{Independence Sampler} \\
\midrule
\textit{Bounded likelihood ratio} & Geometric & Geometric \\
\addlinespace
\textit{Unbounded likelihood ratio; finite exponential moments of cycle length distribution} & Geometric & Sub-geometric\\
\addlinespace
\textit{Unbounded likelihood ratio; finite $p$-moments of cycle length distribution} & Polynomial &  Sub-geometric\\
\bottomrule
\end{tabular}
\end{table}

In the literature, sub-geometric rates include logarithmic, polynomial and super-polynomial (called sub-exponential) rates. Even though, to our knowledge, there is no work on sub-geometric convergence specifically of the Independence Sampler, one could consult \cite{Nummelin:RateOfConvergenceOrey,Tuominen:Subgeometric,Douc:Subgeometric} for results on sub-geometric rates in general, and \cite{Jarner:Polynomial} for more specific results on polynomial rates.

Even if the RRS method has exponential convergence for a larger class of instances compared to the Independence Sampler, it still inherits its limitations, together with those of the Rejection Sampling algorithm.
\begin{remark}
    First, to use the RRS method, one must be able to evaluate the (possibly unnormalized) densities of the target and proposal distributions, but this is usually the case in most applications.
    
    Moreover, the RRS method heavily depends on the choice of proposal distribution. Selecting an unsuitable proposal distribution may even preclude proper convergence of the method, since the cycle length distribution (i.e. the likelihood ratio distribution) must have at least polynomial moments to ensure a mere polynomial convergence. However, the RRS method does not need the likelihood ratio to be bounded, as it is the case for the Rejection Sampling algorithm. Hence, in cases where it is unclear whether Rejection Sampling can be efficiently used, it would be possible to opt for the RRS method, without worrying about the boundedness of the likelihood ratio.

    Additionally, the RRS method's performance is highly influenced by the choice of the time threshold $t$. Depending on the application (e.g. obtaining \emph{good} samples from the target distribution) one could be tempted to run the process until a very large time $T\gg 1$. If the convergence is exponential, the method would indeed return a high-quality sample, but that would entail a very large computational cost, given by the generation of numerous proposal samples. On the other hand, $t$ cannot be chosen too low, because the process may not have converged to its stationary distribution yet.\\
    Furthermore, what \emph{large} actually means in the context of the time threshold $t$, is also heavily dependent on the particular form of the cycle length distribution, and cannot be prescribed \emph{a priori}.

    Lastly,  as happens with both the IS and the Rejection Sampling algorithms, the performance of the RRS method decreases as the number of dimensions increase. 
\end{remark}

\chapter{Time-Average Estimator Results}
\label{chap:6}
Up to this point, we have studied the convergence properties of the Regenerative Rejection Sampling method towards its stationary distribution. This type of analysis provides a characterization of how close the law of the underlying regenerative process (see Remark \ref{rmk:RRSRegProcConstr}) is, at each time step $t$, to the limiting distribution.

Such information is mostly relevant when the goal is to get a "good" (i.e. almost exact) sample from the stationary distribution, and hence, we decide to run the RRS method for a long time, to ensure that the process' law is basically undistinguishable from the target distribution.

However, if we wish to compute a quantity of the type $q = \mathbb{E}_\infty[h(Y(\infty))]$, where $\mathbb{E}_\infty[\cdot]$ represents an expectation under the stationary distribution of the regenerative process, and $h$ is a given Borel function such that $\mathbb{E}_\infty\left[|h(Y(\infty))|\right]<\infty$, we could use all the samples generated from one run of the RRS method and produce an estimate $\hat{q}$ of the desired quantity $q$. Since the underlying process is regenerative (in the sense of Definition \ref{def:RegProcClassic}), we can exploit the i.i.d. cycle structure to construct the estimator. This method falls under the name of \emph{Regenerative method of simulation}, of which one can find a thorough introduction in \cite{CraneLemoine:RegenerativeMethodSimulation}.

It works as follows. Let us use the notation from Remark \ref{rmk:RRSRegProcConstr} and define
\[
V_n \colon= \int_{[T_{n-1},T_n)}h(Y(s))\text{d}s.
\]
Since the cycles are i.i.d. by definition of regenerative process, we deduce that the random vectors $\{(V_n, W_n)\}_{n\in\mathbb{N}}$ are i.i.d. (we do not require independence between the $V_n$'s and the $W_n$'s) and that \cite[Proposition A.3]{CraneIglehart:SimulatingStableStochSystIII}
\[
q = \frac{\mathbb{E}[V_n]}{\mathbb{E}[W_n]}.
\]
The most natural estimator of this quantity that we can construct is then defined as
\[
\tilde{q}(N) = \frac{\sum_{n=1}^N V_n}{\sum_{n=1}^N W_n},
\]
which can be shown to converge almost surely to $q$ as $N\rightarrow\infty$ \cite{CraneIglehart:SimulatingStableStochSystIII, MeketonHeidelberger:BiasReduction}. These type of estimators, called \emph{ratio estimators}, have been studied extensively in the literature, and we will use known results to proceed with the analysis of our estimator. For additional information and a more general discussion (i.e. not related to just the specific form of the estimator that originates from the RRS method), one could look at \cite{MeketonHeidelberger:BiasReduction}, \cite{CraneIglehart:SimulatingStableStochSystIII}, \cite{AwadGlynn:SteadyStateEstimators}, \cite{HendersonGlynn:RegenerativeSimulationDiscreteEvent}, \cite{BrownSolomon:SecondOrderApproximationVariance}, \cite{GlynnHeidelberger:BiasBudgetConstrainedSimulations}, \cite{HeidelbergerLewis:RegressionAdjust} and references therein.

\begin{remark}
    For our specific case, the Regenerative Rejection Sampling method is run until the partial sum of the cycle lengths exceeds a given time threshold $t$. By analyzing how the method is specified (see Algorithm \ref{algo:RRS}), it is clear that in each independent run of the process, we generate a different random number of samples from the proposal distribution, equal to the specific realization of the random variable $N(t) \colon = \inf\{n: W_1+\cdots + W_n >t\}$. Hence, the associated time-average estimator is $\tilde{q}(N(t))$.

    Additionally, since the RRS method simulates a regenerative process that is constant along each cycle, the form of the estimator simplifies. Notice that
    \[
    V_n = \int_{[T_{n-1},T_n)} h(Y(s))\text{d}s = \int_{[T_{n-1},T_n)} h(X_n)\text{d}s = h(X_n)(T_{n}-T_{n-1}) = h(X_n)W_n,
    \]
    where $X_n$ is the $n$-th sample generated from the proposal distribution.

    Hence, in the rest of the Chapter, we analyze the following time-average ratio estimator: 
    \[
    \hat{q}(t) \colon = \tilde{q}(N(t)) = \frac{\sum_{n=1}^{N(t)}h(X_n)W_n}{\sum_{n=1}^{N(t)} W_n}.
    \]
\end{remark}

\section{Bias}
\subsection{Known Results}
As previously remarked, ratio estimators have been studied in numerous occasions, also in the form of $\hat{q}(t)$. It has been shown in \cite{CraneIglehart:SimulatingStableStochSystIII} that $\tilde{q}(N(t)-1)$ converges almost surely to $q$ as $t\rightarrow\infty$, and the result can be extended for $\hat{q}(t)$.

One of the main references for the study of the bias properties of ratio estimators based on a random number of samples is \cite{MeketonHeidelberger:BiasReduction}. The authors showed that the bias of the estimator $\tilde{q}(N(t)-1)$ is of order $1/t$, while that of $\tilde{q}(N(t))$ is of order $1/t^2$. This difference is due to the inspection paradox of renewal theory (see Remark \ref{rmk:InspectionParadox}), which states that the mean of the stationary distribution of the current life process $C(t)$, i.e. $\mathbb{E}[W_{N(t)}]$, is (usually much) higher than the mean of the cycle length distribution \cite{AwadGlynn:SteadyStateEstimators}. As a consequence, the last cycle contains more information than the rest, and including it in the ratio estimator brings a significant bias reduction \cite{MeketonHeidelberger:BiasReduction}. 

The results above are summarized in the following Theorem, which we state for the specific case of our estimator $\hat{q}(t)$ \cite{MeketonHeidelberger:BiasReduction}. We refer the reader to the reference for a proof.
\begin{theorem}
\label{thm:BiasReduction}
    If the cycle length distribution has an absolutely continuous component, $\mathbb{P}(W>0)=1$, $\mathbb{E}[W^3]<\infty$, and $\mathbb{P}(|Wh(X)|\leq KW)=1$ (i.e. $h$ is an almost surely bounded function) for some constant $K$, then
    \begin{enumerate}
        \item $\mathbb{E}\left[\hat{q}(t)\right] = q + O(1/t^2)$;
        \item $\mathbb{E}\left[\tilde{q}(N(t)-1)\right] = q + c/t + O(1/t^2)$,
    \end{enumerate}
    where $c = \mathbb{E}[Wh(X)]\left(\mathbb{V}ar(W)/\mathbb{E}[W]^2 - \mathbb{C}ov(Wh(X),W)/\mathbb{E}[Wh(X)]\mathbb{E}[W]\right)$.
\end{theorem}

In \cite{AwadGlynn:SteadyStateEstimators}, the authors provide a higher-order expansion of the bias of $\hat{q}(t)$ under different assumptions. This confirms the result of \cite{MeketonHeidelberger:BiasReduction}, and gives more information on the other terms of the bias expansion. In particular, it is shown that the coefficient of $1/t$ is zero, while all the others are generally non-zero \cite[Theorem 4.1]{AwadGlynn:SteadyStateEstimators}.

\begin{remark}
    In the context of Markov Chain Monte Carlo methods, it is easy to show that (see Appendix \ref{sec:BiasMCMC}) if the method is geometrically ergodic, then the bias of estimators of the type $(1/N)\sum_{n=1}^Nh(X_n)$, where the $X_n$'s are the steps of the Markov chain, is of order 1/N.
    Hence by running the RRS method until time $t$, and a geometrically ergodic MCMC for $N=M$ steps, where $M$ is chosen in order to make, on average, the same simulation effort (here we could take for example $M=\lfloor t/\mathbb{E}[W]\rfloor$, see the Renewal LLN \ref{thm:LLNCountingProcess}), we obtain a lower bias by using the samples generated from the RRS method.
\end{remark}

\subsection{A Non-Asymptotic Bound}
The previous section presented a very powerful result concerning bias reduction in the case of the ratio estimator $\hat{q}(t)$. However, Theorem \ref{thm:BiasReduction}, only gives an order result, which might not be extremely useful for practical simulation purposes. The solution to this issue would be to compute a non-asymptotic error bound for the bias of the estimator. 

The following Theorem shows that such a bound exists, and gives its specific form.
\begin{theorem}
\label{thm:BiasBound}
    Under the same assumptions of Theorem \ref{thm:BiasReduction}, denote $\mu = \mathbb{E}[W], \mu_2 = \mathbb{E}[W^2], \mu_3 = \mathbb{E}[W^3]$. Then
    \[
    \left| \mathbb{E}\left[\hat{q}(t)\right] - q\right| \leq \frac{\sqrt{\frac{16}{3}K^2\mu_3 \mu_2(\frac{\mu_2}{t} + \mu)\frac{1}{\mu^3}}}{t^{3/2}}.
    \]
\end{theorem}
\begin{proof}
    Recall the definition of the forward recurrence time process $R(t) = T_{N(t)}-t$. Since $N(t)$ is a stopping time, we can apply Wald's first identity \cite{Wald:Indentity1} to get $r(t) \colon=\mathbb{E}[R(t)] = \mu \mathbb{E}[N(t)] - t$.

    Now, define $Z_n \colon= W_n h(X_n) - qW_n$. We have
    \[
    \hat{q}(t) = \frac{\sum_{n=1}^{N(t)}W_n h(X_n)}{\sum_{n=1}^{N(t)}W_n} = \frac{\sum_{n=1}^{N(t)}W_n h(X_n)}{T_{N(t)}} = q + \frac{\sum_{n=1}^{N(t)}Z_n}{T_{N(t)}}.
    \]

    Then,
    \[
    \begin{split}
    \mathbb{E}\left[\hat{q}(t)\right] - q &= \mathbb{E}\left[\frac{\sum_{n=1}^{N(t)}Z_n}{T_{N(t)}}\right]\\
    &= \mathbb{E}\left[\frac{\frac{1}{t}\sum_{n=1}^{N(t)}Z_n}{1+R(t)/t}\right]\\
    &\stackrel{(i)}{=} \mathbb{E}\left[\left(\frac{1}{1+R(t)/t}-1\right) \bar{Z}_t\right],
    \end{split}
    \]
    where $\bar{Z}_t = \frac{1}{t}\sum_{n=1}^{N(t)}Z_n$, and $(i)$ comes from applying Wald's first identity \cite{Wald:Indentity1} to get $\mathbb{E}\left[\sum_{n=1}^{N(t)}Z_n\right] = \mathbb{E}[N(t)]\mathbb{E}[Z_n] = 0$.

    Since $\frac{1}{1+R(t)/t}\leq 1$, we obtain
    \[
    \begin{split}
    \left|\mathbb{E}\left[\hat{q}(t)\right] - q \right| &= \frac{1}{t}\left| \mathbb{E}\left[ \frac{R(t)}{1+R(t)/t}\bar{Z}_t \right] \right|\\
    &\leq \frac{1}{t}\left| \mathbb{E}\left[ R(t)\bar{Z}_t \right] \right|\\
    &\stackrel{(ii)}{\leq} \frac{\sqrt{\mathbb{E}\left[ R(t)^2 \right] \mathbb{E}\left[ \bar{Z}_t^2 \right]}}{t}\\
    &\stackrel{(iii)}{=} \frac{\sqrt{\mathbb{E}\left[ R(t)^2 \right]}}{t}\sqrt{\frac{\mathbb{E}[N(t)]\mathbb{E}\left[ Z^2 \right]}{t^2}}\\
    &\stackrel{(\dagger)}{=} \frac{\sqrt{\mathbb{E}\left[ R(t)^2 \right]}}{t^{3/2}}\sqrt{\frac{\mathbb{E}[N(t)]}{t}\mathbb{E}\left[ Z^2 \right]},
    \end{split}
    \]
    where $(ii)$ the Cauchy-Schwarz inequality, and $(iii)$ is an application of Wald's second moment identity \cite{Wald:Indentity1, BlackwellGirshick:WaldIndentity2}.

    Since $\mathbb{P}(|Wh(X)|\leq KW)=1$, we can deduce that $q \in [-K,K]$, and that
    \[
    \begin{split}
    Z^2 &= (Wh(X)-qW)^2\\
    &\leq (|Wh(X)|+|q|W)^2\\
    &\leq (KW + KW)^2\\
    &= 4K^2W^2.
    \end{split}
    \]

    By applying Lorden's inequalitites \cite{Lorden:Inequalities} to bound the moments of $R(t)$, and by using that $\mathbb{E}[N(t)] = (r(t) + t)/\mu$ we can conclude that
    \[
    \begin{split}
        \left|\mathbb{E}\left[\hat{q}(t)\right] - q \right| &\stackrel{(\dagger)}{=} \frac{\sqrt{\mathbb{E}\left[ R(t)^2 \right]}}{t^{3/2}}\sqrt{\frac{\mathbb{E}[N(t)]}{t}\mathbb{E}\left[ Z^2 \right]}\\
        &\leq \frac{\sqrt{\frac{4}{3}\frac{\mu_3}{\mu}}}{t^{3/2}}\sqrt{4K^2\frac{\mu_2}{\mu}\left( \frac{\mu_2}{t\mu} + 1 \right)}\\
        &= \frac{\sqrt{\frac{16}{3} K^2 \mu_3 \mu_2 \left( \frac{\mu_2}{t} + \mu \right)\frac{1}{\mu^3}}}{t^{3/2}}.
    \end{split}
    \]
\end{proof}

This proof takes inspiration from the proof of Theorem 2 in \cite{Botev:GS}. It returns a non-asymptotic upper bound on the bias of the ratio estimator $\hat{q}(t)$, which decays as $t^{-3/2}$. Theorem \ref{thm:BiasReduction}, states that $\left|\mathbb{E}\left[\hat{q}(t)\right]-q\right|$ is $O(t^{-2})$, hence the upper bound gives a lower order than the actual one. However, the bound is non-asymptotic, which means that it holds for all $t > 0$, and all its components (mainly moments of $W$) are easily estimable by simulation, which enables us to meaningfully control the bias of $\hat{q}(t)$ during each run of the process.

\subsection{Numerical Experiment}
At this point, we perform a numerical experiment to verify the bias results for the time-average estimator $\hat{q}(t)$. Let us recall the simple setting of Example \ref{ex:RRSGammaExp1}. We want to simulate random variables from a $Gamma(2,1)$ distribution, using samples from an $Exp(1)$ distribution. To do it, we run the Regenerative Rejection Sampling algorithm until a given time threshold $t$. This procedure will return a process whose states are random variables approximately distributed as $Gamma(2,1)$, say, $\{X_n\}_{n=1}^{N(t)}$, as well as the same number of cycle lengths $\{W_n\}_{n=1}^{N(t)}$.

With these random variables, one can construct the time-average estimator
\[
\hat{q}(t) = \frac{\sum_{n=1}^{N(t)}W_nh(X_n)}{\sum_{n=1}^{N(t)}W_n},
\]
for a given bounded Borel function $h$. In our experiments we concentrated on three different functions:
\begin{itemize}
    \item The hyperbolic tangent $h(x)=\tanh(x)$, which satisfies $|\tanh(x)|\leq 1$ for all $x$;
    \item The logistic function $h(x) = 1/(1+e^{-x})$, which satisfies $1/(1+e^{-x})\leq 1$ for all $x$;
    \item And $h(x) = \mathbb{I}_{[1,\infty)}(x)$, which corresponds to estimating the tail probability $\mathbb{P}(X>1), X\sim Gamma(2,1)$, and that satisfies $\mathbb{I}_{[1,\infty)}(x)\leq 1$ for all $x$.
\end{itemize}

Since it is very easy to compute good estimates of $\mathbb{E}_{Gamma(2,1)}[h(X)]$ for each of these functions $h$ (e.g. by using well known python libraries as \texttt{scipy} and \texttt{numpy}), we have a practical way to check the bias results on the ratio estimator $\hat{q}(t)$.

To this end, we ran $M=10^6$ independent RRS and computed the bias of $\hat{q}(t)$ for all $t\in\{0.1,0.5,1,5,10,15,20,30,50,75,100,150,200,300,500,750,1000\}$, and we repeated this procedure for the three different choices of $h(x)$. To confirm the theoretical results, we plotted the biases together with the computable upper bound (recall that $K=1$ for all the three $h$ functions) and the $O(t^{-2})$ order. Additionally, for the sake of comparison, we also plotted the bias of the estimator $\tilde{q}(N(t)-1)$, i.e. constructed only with the generated samples up to $N(t)-1$, together with its theoretical order of decay $O(t^{-1})$.

\begin{figure}[htpb]
    \centering
    \includegraphics[width=0.55\linewidth]{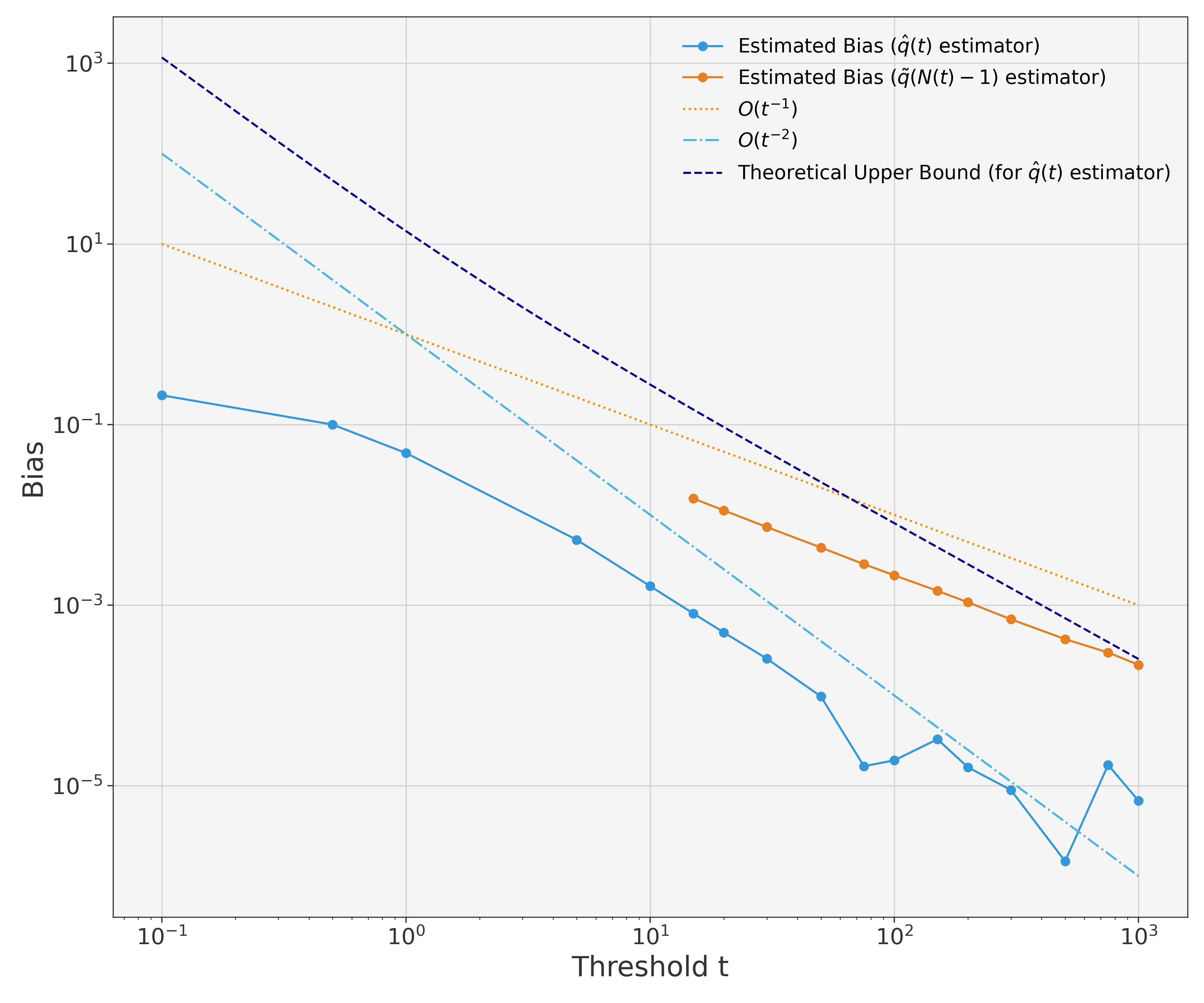}
    \caption{Bias of $\hat{q}(t)$ - $h(x) = \tanh(x)$}
    \label{fig:BiasTanh}
\end{figure}
\begin{figure}[htpb]
    \centering
    \includegraphics[width=0.55\linewidth]{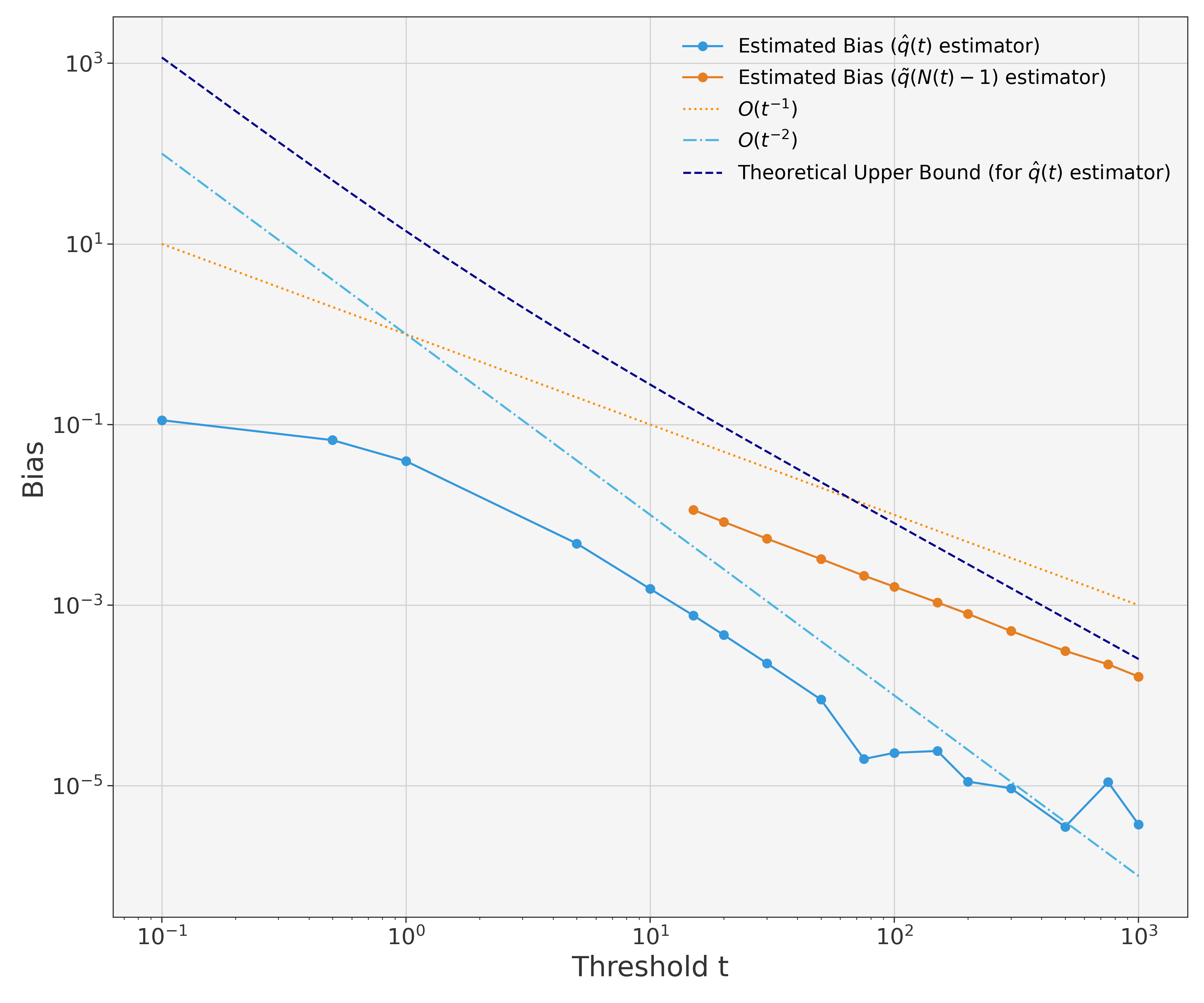}
    \caption{Bias of $\hat{q}(t)$ - $h(x) = 1/(1+e^{-x})$}
    \label{fig:BiasLogistic}
\end{figure}
\begin{figure}[htpb]
    \centering
    \includegraphics[width=0.55\linewidth]{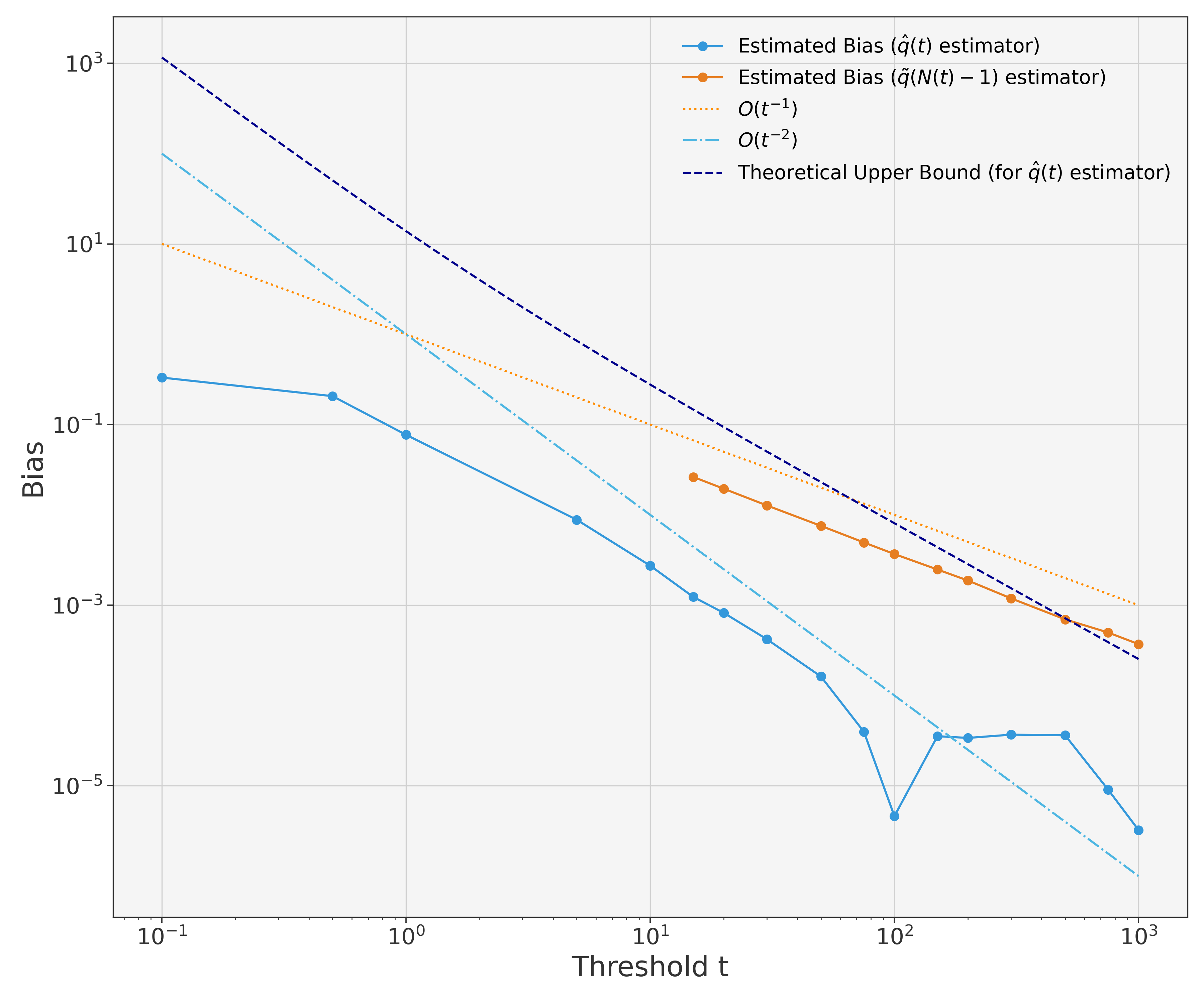}
    \caption{Bias of $\hat{q}(t)$ - $h(x) = \mathbb{I}_{[1,\infty)}(x)$}
    \label{fig:BiasTail}
\end{figure}

In Figures \ref{fig:BiasTanh}, \ref{fig:BiasLogistic}, and \ref{fig:BiasTail}, we clearly show that the theoretical results are respected. The estimates of the biases of $\hat{q}(t)$ lay below the theoretical upper bound, and from around $t=10$, they seem to decay as $1/t^2$. The noise in the estimated bias of $\hat{q}(t)$ for high values of $t$ is due to the fact that we are plotting a Monte Carlo estimate of the bias, and for those $t$'s, the actual bias is lower than the Monte Carlo error. Nonetheless, we are still able to meaningfully understand the decay.

Moreover, we can also remark that the $O(1/t)$ order result for the decay of the bias of $\tilde{q}(N(t)-1)$ is experimentally confirmed (for this choices of $h$). In the plots we see less data points for this estimator because, in some runs for low values of $t$, the RRS method only returned one sample.

\section{Variance}
\subsection{Known Results}
Numerous references (e.g. \cite{CraneIglehart:SimulatingStableStochSystIII},\cite[Chapter IV]{AsmussenGlynn:StochSim}) report CLT-type results for the ratio estimator based on a fixed number $N$ of cycles, i.e. $\tilde{q}(N)$. Specifically, the following Proposition holds \cite[Proposition IV.4.1]{AsmussenGlynn:StochSim} (here we used the notation presented at the beginning of this Chapter):

\begin{proposition}
    The estimator $\tilde{q}(N)$ satisfies the CLT
    \[
    N^{1/2}(\tilde{q}(N)-q) \stackrel{\mathscr{D}}{\longrightarrow} \mathcal{N}(0,\eta^2)
    \]
    as $N\rightarrow\infty$, where
    \[
    \eta^2 = \frac{\mathbb{E}[Z_n^2]}{\mathbb{E}[W_n]^2},\quad Z_n = V_n - qW_n.
    \]
\end{proposition}

In \cite{CraneIglehart:SimulatingStableStochSystIII}, the authors explain a practical method for constructing confidence intervals (CIs), i.e. a way to estimate the \emph{Time Average Variance Constant} (TAVC) $\eta^2$, which is not known \emph{a priori}. Let $\bm{U}_n = (V_n, W_n)$ be the i.i.d vectors outputted by the simulation. Additionally, denote by
\[
\bar{\bm{U}}(N) = \begin{pmatrix}
\bar{V}(N)\\
\bar{W}(N)
\end{pmatrix} = \frac{1}{N}\sum_{n=1}^N\bm{U}_n
\]
the sample mean of $N$ vectors ($\bar{V}(N)$ and $\bar{W}(N)$ are the sample means of the $V$'s and the $W$'s, respectively), and by
\[
\bm{S}(N) = \begin{pmatrix}
s_{11}(N) & s_{12}(N)\\
s_{21}(N) & s_{22}(N)
\end{pmatrix} = \frac{1}{N-1}\sum_{n=1}^{N}(\bm{U}_n-\bar{\bm{U}}(N))(\bm{U}_n-\bar{\bm{U}}(N))^T
\]
their sample covariance matrix. Then, it is possible to show that $s^2(N) = s_{11}(N) - 2qs_{12}(N) + q^2s_{22}(N)$ converges to $\mathbb{E}[Z_n^2]$ with probability one as $N\rightarrow\infty$ \cite{CraneIglehart:SimulatingStableStochSystIII}. Hence, a possible estimator for the TAVC $\eta^2$ is given by (\cite{AsmussenGlynn:StochSim})
\[
\tilde{\eta}^2(N) \colon = \frac{s^2(N)}{(\bar{W}(N))^2}.
\]
It is then possible to use this estimator to construct confidence intervals, using the CLT provided above.

The ratio estimator that is computable after a run of the RRS method (i.e. $\hat{q}(t)$) is, however, based on a fixed \emph{simulation time}, not on a fixed number of regenerative cycles. Fortunately, the results and methods presented above hold for the estimators based on a random number of cycles, i.e. both $\hat{q}(t)$ and $\tilde{q}(N(t)-1)$. The only difference lies in the value of the time average variance constant, which is equal to 
\[
\sigma^2 = \frac{\mathbb{E}[Z_n^2]}{\mathbb{E}[W_n]}.
\]
Nonetheless, this difference is consistent with the previous results, as it only represents a change of time scale. By Theorem \ref{thm:LLNCountingProcess}, $N(t)/t\rightarrow1/\mathbb{E}[W_n]$ as $t\rightarrow\infty$. As a consequence, we can write, heuristically, 
\[
\mathbb{V}ar(\tilde{q}(N(t)-1)) \approx \frac{1}{N(t)}\frac{\mathbb{E}[Z_n^2]}{\mathbb{E}[W_n]^2} \approx \frac{\mathbb{E}[W_n]}{t}\frac{\mathbb{E}[Z_n^2]}{\mathbb{E}[W_n]^2} = \frac{1}{t}\frac{\mathbb{E}[Z_n^2]}{\mathbb{E}[W_n]}.
\]

The literature also provides order results for the variance of $\hat{q}(t)$ (and also $\tilde{q}(N(t)-1)$), see for example \cite{MeketonHeidelberger:BiasReduction}. The following Proposition summarizes what is known \cite{MeketonHeidelberger:BiasReduction}:
\begin{proposition}
    Under the same assumptions of Theorem \ref{thm:BiasReduction} and if $\mathbb{E}[W^5]<\infty$, we have
    \begin{enumerate}
        \item $\mathbb{V}ar(\hat{q}(t)) = \mathbb{E}[Z^2]/(t\mathbb{E}[W]) + O(1/t^2)$;
        \item $\mathbb{V}ar(\tilde{q}(N(t)-1)) = \mathbb{E}[Z^2]/(t\mathbb{E}[W]) + O(1/t^2)$.
    \end{enumerate}
\end{proposition}
Hence, from a mean square error (MSE) point of view, the two estimators are equivalent, because for both $\hat{q}(t)$ and $\tilde{q}(N(t)-1)$, the variance is the leading term in the MSE expansion. This means that the MSE is insensitive to the bias. Nonetheless, it is still favorable to use the low-bias estimator $\hat{q}(t)$:
\begin{remark}
    As pointed out in \cite{MeketonHeidelberger:BiasReduction}, for small values of $t$, $\hat{q}(t)$ and the estimator for its variance will be highly correlated, which implies a reduced confidence interval coverage. This can be solved by employing multiple independent replications of the simulation procedure, and computing the variance across replications. However, a low-bias estimate and a proper selection of the replications length are still needed for obtaining satisfactory CI coverage.
\end{remark}
What this means, in practice, is that for any fixed \emph{total simulation budget} $T$, we can choose $R$ replications of length $t=T/R$, such that:
\begin{itemize}
    \item $R$ is high enough, so that the variance across replications gives proper CI coverage;
    \item $t$ is not too low, so that the bias of the single-replication ratio estimators is small.
\end{itemize}

\chapter{Some Applications}
\label{chap:7}
This chapter is devoted to the study of two practical applications of the Regenerative Rejection Sampling algorithm, together with empirical comparisons with well-known Markov Chain Monte Carlo methods. In the first, we sample from a given \emph{synthetic} bi-dimensional target distribution, and compare the performance of the Regenerative Rejection Sampling method with that of the Random Walk Metropolis Algorithm. The second example is centered on a real dataset containing data of patients with the \emph{Lupus} disease, on which we perform a probit regression. We compare the performance of the Gibbs sampler with that of our Regenerative Rejection Sampling method.

\section{A Synthetic Example}
\label{sec:SyntheticExample}
Consider the problem of sampling from the distribution defined on the bounded domain $[-2\pi,2\pi]^2\in\mathbb{R}^2$, with density
\begin{equation}
    f(x_1,x_2) = \frac{e^{-\frac{1}{4}\sqrt{x_1^2+x_2^2}}\,\left(\sin\left( 2\sqrt{x_1^2+x_2^2} \right)+1\right)}{\mathcal{Z}}, \quad x_1\in[-2\pi,2\pi], x_2\in[-2\pi,2\pi],
\end{equation}
where $\mathcal{Z}$ is an \emph{unknown} normalizing constant. We call the unnormalized density $f_\propto(x_1,x_2)$. Figure \ref{fig:SyntheticBoundedSurface} shows the plot of $f_\propto(x_1,x_2)$.
\begin{figure}[h]
    \centering
    \includegraphics[width=0.55\linewidth]{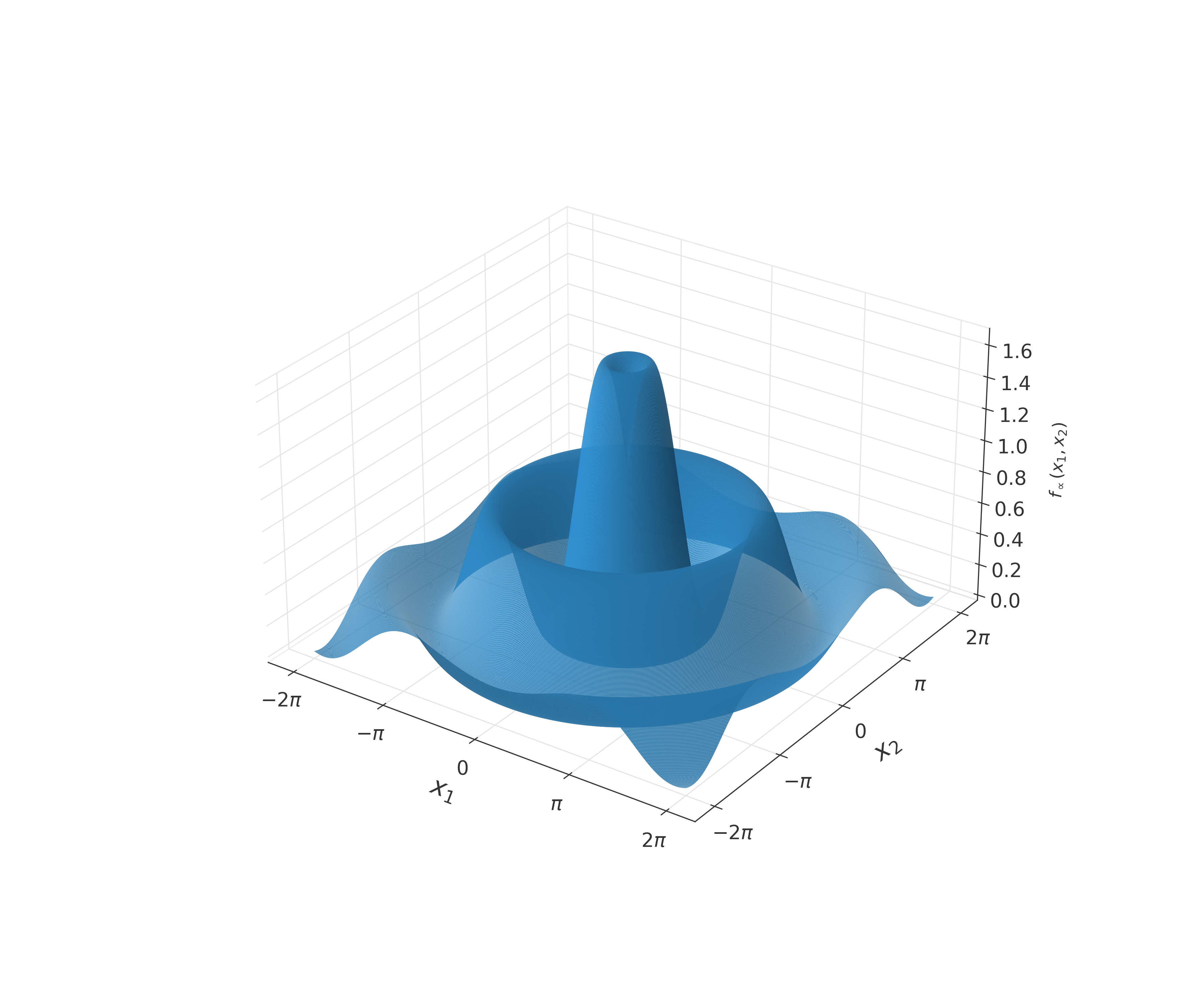}
    \caption{Surface plot of $f_\propto(x_1,x_2)$}
    \label{fig:SyntheticBoundedSurface}
\end{figure}

We start the analysis by showing the performance of the Random Walk Metropolis Algorithm. First, let us recall how the method works.

The Random Walk Metropolis (RWM) Algorithm \cite{Metropolis:RWM} was the first Metropolis-type algorithm to be introduced. As always happens with such algorithms, it constructs a Markov Chain that has as stationary distribution the desired target distribution (in this case, with density $f$). At each step of the algorithm, we generate a sample from a \emph{symmetric} proposal distribution $g$, meaning that $g(\bm{x}|\bm{y}) = g(\bm{y}|\bm{x})$ (in classical MCMC notation). This yields a simplified formula for the acceptance probability $\alpha_{RWM}(\bm{x},\bm{y})$:
\[
\alpha_{RWM}(\bm{x},\bm{y}) = \min\left\{1, \frac{f_\propto(y_1,y_2)}{f_\propto(x_1,x_2)}\right\},
\]
where $\bm{y}=(y_1,y_2)$ is the proposed state, and $\bm{x}=(x_1,x_2)$ is the current state of the chain. For this particular example, we chose a  proposal distribution of the type
\[
\tilde{g}(x_1,x_2) = \frac{1}{64}e^{-\frac{|x_1|}{4}-\frac{|x_2|}{4}},\quad x\in\mathbb{R}^2.
\]
This corresponds to sampling two independent $Laplace(0,4)$ random variables, one for each component. I.i.d. sampling from the $Laplace(0,4)$ distribution can be achieved via the Inverse-Transform method, i.e. sampling i.i.d. random variables $U\sim\mathcal{U}\left(-\frac{1}{2},\frac{1}{2}\right)$, and then computing $4\text{ sign}(U)\log(1-2|U|)$. We can easily center the proposed samples in the current step of chain, say $X^n$, by computing $X^n_i + 4\text{ sign}(U)\log(1-2|U|)$, for $i=1,2$, where $i$ represents the component of the random vector. In this way, the symmetry of the proposal distribution is ensured.

\subsection{Random Walk Metropolis - Bounded domain}
Before showing the sampling results, let us make a remark concerning the supports of the target and of the proposal distribution.

Since the unnormalized target density $f_\propto$ is defined only on the square $[-2\pi,2\pi]^2$, one may wonder why we sample from a proposal distribution defined on the whole real plane $\mathbb{R}^2$. However, doing this does not represent a problem. If we propose a step that lies outside of the support of the target distribution, the probability of acceptance $\alpha(\bm{x},\bm{y})$ will be equal to $0$. Hence, the step will be rejected. This means that using a proposal distribution defined on $\mathbb{R}^2$ does not invalidate the produced samples, but only decreases the efficiency of the procedure, compared to a proposal defined on the same support of the target, and that handles well the behavior around the border of the support.

We now move to the sampling results. To asses the performance of the Random Walk sampler, we plotted $10^4$ \emph{consecutive} samples taken from the chain after a burn-in period of $1000$. This was done to reduce the initialization bias arising from starting the chain away from its stationary distribution.
\begin{figure}[h]
    \centering
    \includegraphics[width=0.55\linewidth]{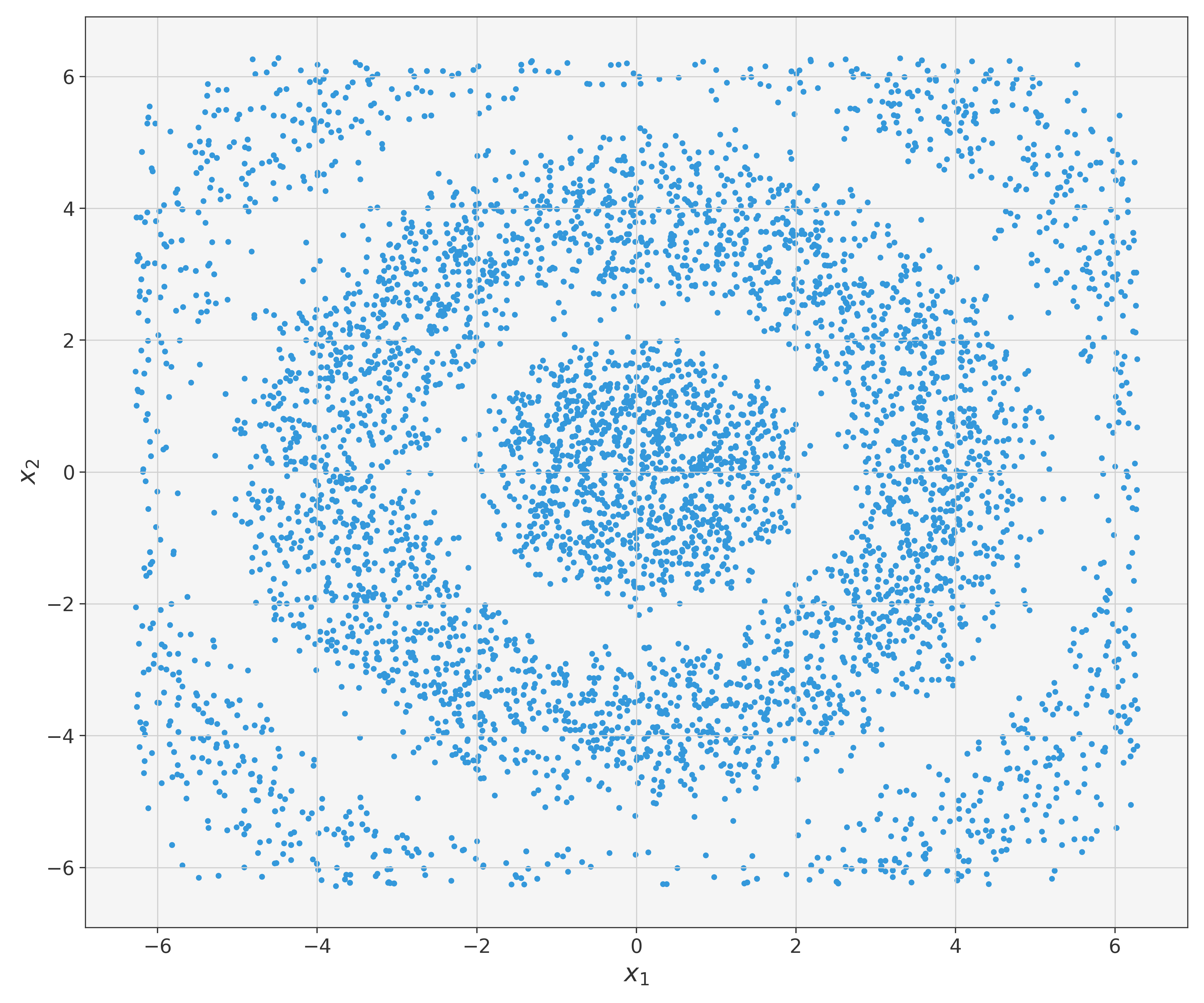}
    \caption{$10^4$ consecutive samples from RWM Algorithm}
    \label{fig:BoundedRWMSamples}
\end{figure}
As we see in Figure \ref{fig:BoundedRWMSamples}, the RW sampler is able to explore the whole domain quite well, and the returned samples correctly concentrate in the areas where the target density presents more mass. To further asses the quality of the samples, we plot the autocorrelation function (ACF) for the each of two components of the samples, and analyze their decay.
\begin{figure}[htbp]
    \centering
    \begin{minipage}{0.48\textwidth}
        \centering
        \includegraphics[width=\linewidth]{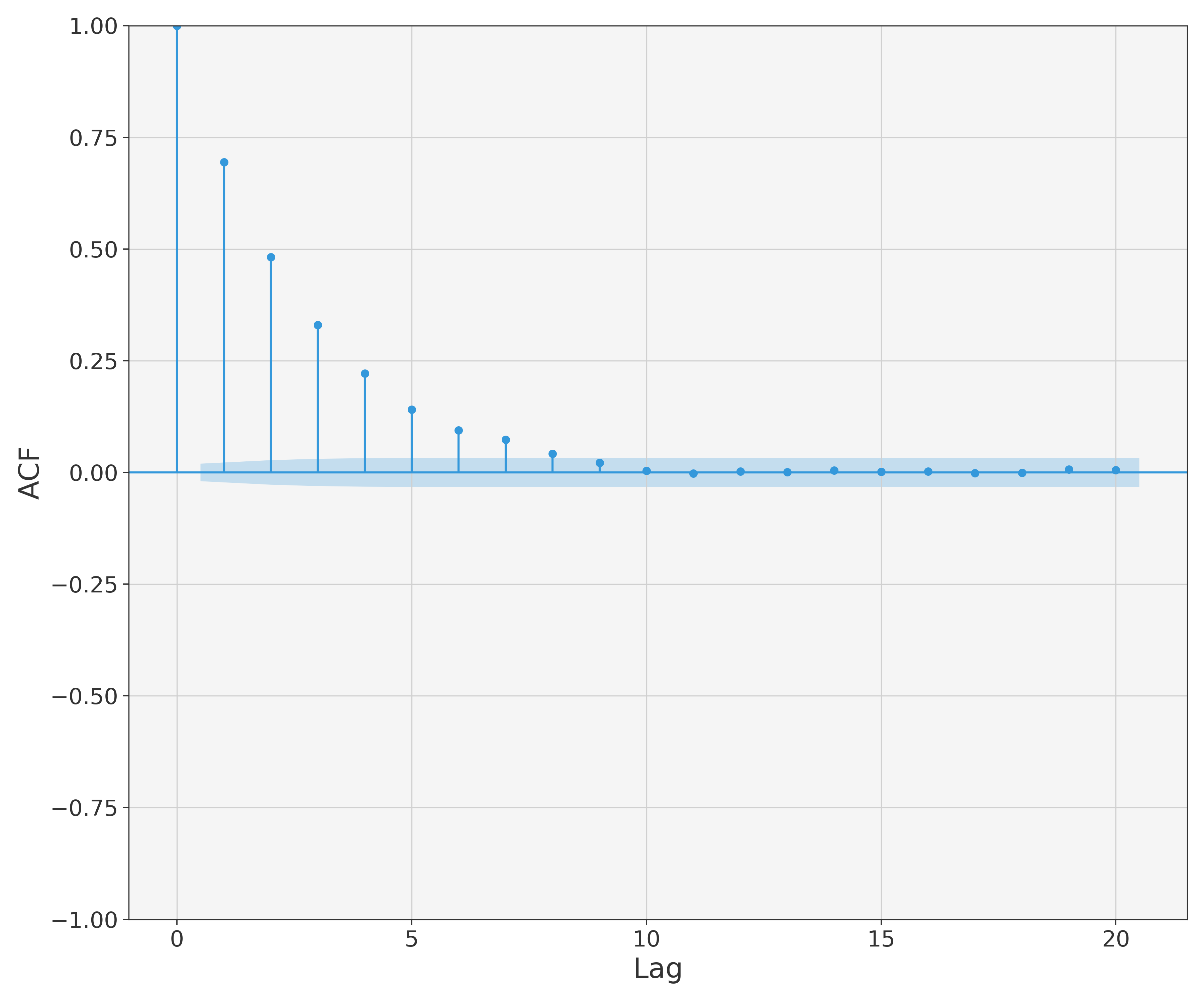}
        \caption{RWM - ACF plot of $1^{st}$ component}
        \label{fig:SyntheticBoundedRWMACF1}
    \end{minipage}
    \hfill
    \begin{minipage}{0.48\textwidth}
        \centering
        \includegraphics[width=\linewidth]{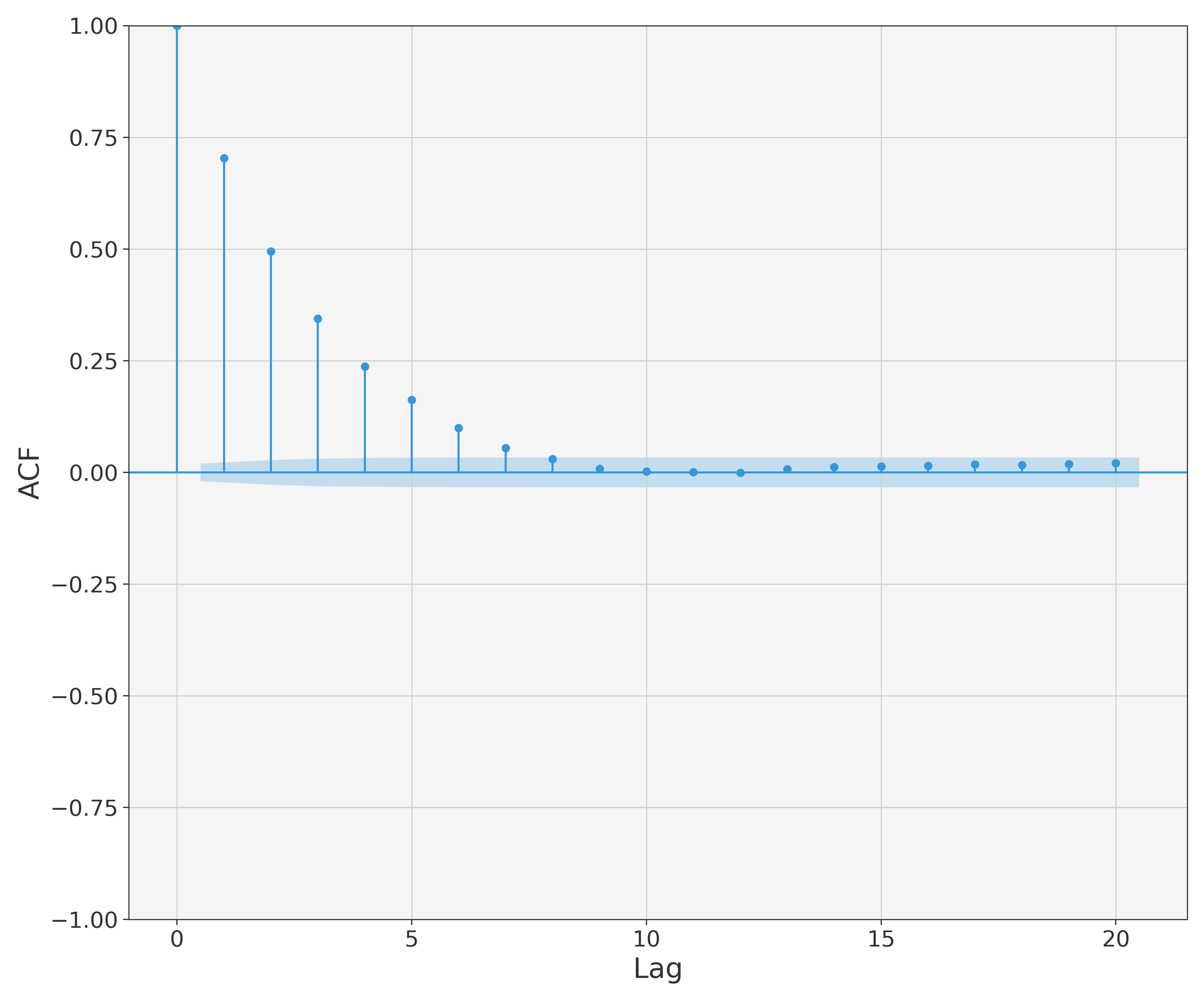}
        \caption{RWM - ACF plot of $2^{nd}$ component}
        \label{fig:SyntheticBoundedRWMACF2}
    \end{minipage}
\end{figure}
From Figures \ref{fig:SyntheticBoundedRWMACF1} and \ref{fig:SyntheticBoundedRWMACF2}, it is possible to see that the RWM decorrelates quite fast, indicating that the returned samples are close to being independent. However, this is a reasonably simple setting, where the RWM is expected to behave well.

\subsection{Regenerative Rejection Sampling - Bounded domain}
Now we turn to the analysis of the performance of the Regenerative Rejection Sampling method applied to the previously described setting.

First of all, to meaningfully compare the two methods, we chose the same $Laplace(0,4)$-type proposal distribution that we used for the Random Walk Metropolis sampler.

However, the proposal density needs to have zero mass outside the support of the target density, i.e. $f(x_1,x_2)=0\implies \tilde{g}(x_1,x_2)=0$. This is required because otherwise we could have a cycle length of 0. Thus, in our case we need to restrict the proposal to the square $[-2\pi,2\pi]^2$. Sampling from the truncated proposal can be done via naive Rejection Sampling, i.e. sample from the non-truncated proposal until we obtain a point inside the square. This yields i.i.d. samples from the correct truncated proposal.

In this particular setting, we can easily calculate the truncated proposal density $g$. The normalization constant is $\int_{[-2\pi,2\pi]^2}\tilde{g}(x_1,x_2)\,\text{d}x_1\,\text{d}x_2 = (1-e^{-\pi/2})^2$, hence the closed form for $g$ is
\begin{equation}
\label{eq:truncatedLaplace}
g(x_1,x_2) = \frac{1}{(1-e^{-\pi/2})^2}\frac{1}{64}e^{-\frac{|x_1|}{4}-\frac{|x_2|}{4}}\mathbb{I}_{[-2\pi,2\pi]^2}(x_1,x_2).
\end{equation}
Note that using the naive Rejection Sampling strategy to sample from the truncated proposal may not be computationally efficient, and one could use better methods to achieve the same result. However, in this specific case, the probability of acceptance is $(1-e^{-\pi/2})^2\approx0.63$, which is sufficiently high. To improve efficiency in this procedure one could tailor a Rejection Sampling proposal distribution concentrated on the square $[-2\pi,2\pi]^2$, but we chose not to do it here.

Moreover, we must make a remark on the method. In Example \ref{ex:RRSGammaExp1}, to generate the $N$ samples, we naively run the algorithm up to time $t$ independently for $N$ times. However, this method is not the most time-efficient. Hence, we slightly modified the scheme, by running only one long regenerative process, and taking as the $n^{th}$ sample ($n=1,2,\dots$) the sample for which the embedded renewal process $T_{N(nt)}$ exceeds $nt$, for a fixed $t>0$ (in a fashion similar to thinning in MCMC). This change introduces correlation among the samples, but it also ensures that the regenerative process is closer to stationarity as more samples are taken. The algorithm is outlined in Algorithm \ref{algo:SequentialRRS}.
\begin{algorithm}
\caption{Sub-Sampled Regenerative Rejection Sampling}
\label{algo:SequentialRRS}
\begin{algorithmic}[1]
\Require{Proposal pdf $g$, time $t$, number of samples $N$}
\State $S \gets 0$
\For{$i=1,\dots,N$}
\Repeat
    \State Simulate $X \sim g$
    \State $W \gets f_\propto(X)/g(X)$
    \State $S \gets S+W$
\Until{$S > i\cdot t$}
\State $X_i \gets X$
\EndFor
\State \textbf{return} $X_1,\dots,X_N$
\end{algorithmic}
\end{algorithm}

The fixed threshold $t$ could have been chosen in various ways: by inspecting the cycle length 's pdf, or by relying on the LLN for the counting process $N(t)$ (cf. Theorem \ref{thm:LLNCountingProcess}), for example. For the sake of completeness, we generated $10^6$ i.i.d. samples from the cycle length distribution (as it only requires sampling from the proposal distribution $g$) and plotted the KDE.
\begin{figure}[h!]
    \centering
    \includegraphics[width=0.55\linewidth]{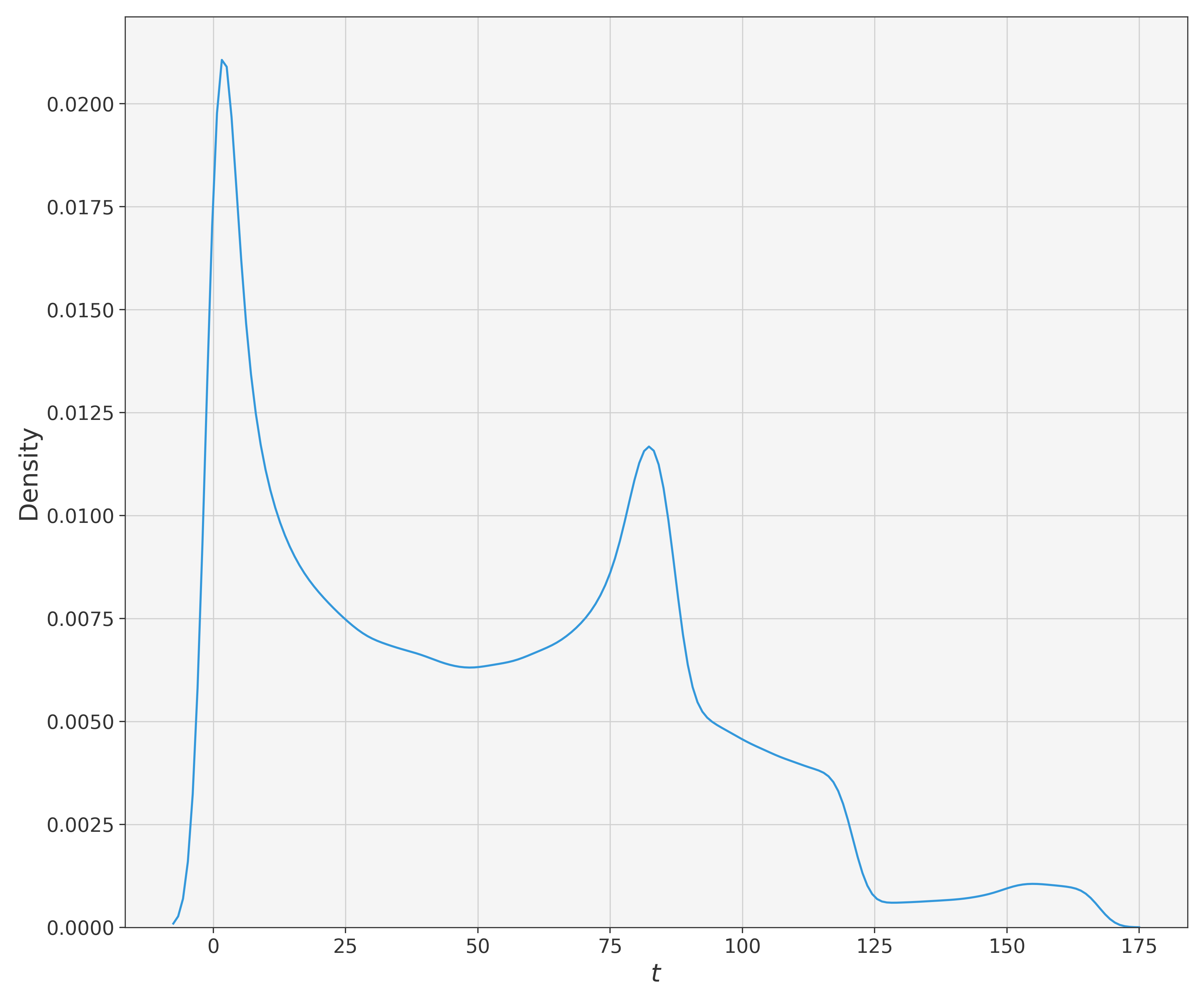}
    \caption{KDE of cycle length distribution}
    \label{fig:SyntheticBoundedRRSCycleLength}
\end{figure}

As is seen in Figure \ref{fig:SyntheticBoundedRRSCycleLength}, the cycle length distribution has most of the mass in $(0,125]$, hence taking a value of $t$ too close to 0, would result in an inefficient sampling procedure. To avoid an excessive increase of computational cost, we could have chosen $t=50$.

However, since a \emph{random} number of samples from the proposal is generated at each run of the RRS method, by fixing an arbitrary value of $t$ we would not necessarily obtain a similar number of samples to those generated by the Random Walk Metropolis (i.e. $N_{RWM}=10^4$ plus a burn-in of 1000). Hence, the comparison between the two methods could be misleading, if not unfair.

To avoid this problem, we can resort to using the LLN for the counting process $N(t)$: as $t\rightarrow\infty$, $N(t)/t\rightarrow1/\mathbb{E}[W]$, where $\mathbb{E}[W]$ is the expected value of the cycle length distribution. In practice, since the goal is to generate (on average and for large values of the time at which the process is stopped) $N_{RWM}+burnin = 11000$ samples, we should set $T = (N_{RWM}+burnin)\mathbb{E}[W]$. $T$ is the time at which the last sample is selected, and since we aim to sub-sample the regenerative process $N_{RRS}=10^4$ times, we select $t=T/N_{RRS} = \frac{N_{RWM}+burnin}{N_{RRS}}\mathbb{E}[W]$. For this specific example, the computed value is $t=56.91$, which is close to what could have been chosen by inspecting the KDE of the cycle length density. In practice, the expected value of the cycle length distribution is computed as the sample mean of the $10^6$ i.i.d samples generated for the KDE.

At this point, we analyze the sampling results from the Regenerative Rejection Sampling method.
\begin{figure}[h!]
    \centering
    \includegraphics[width=0.55\linewidth]{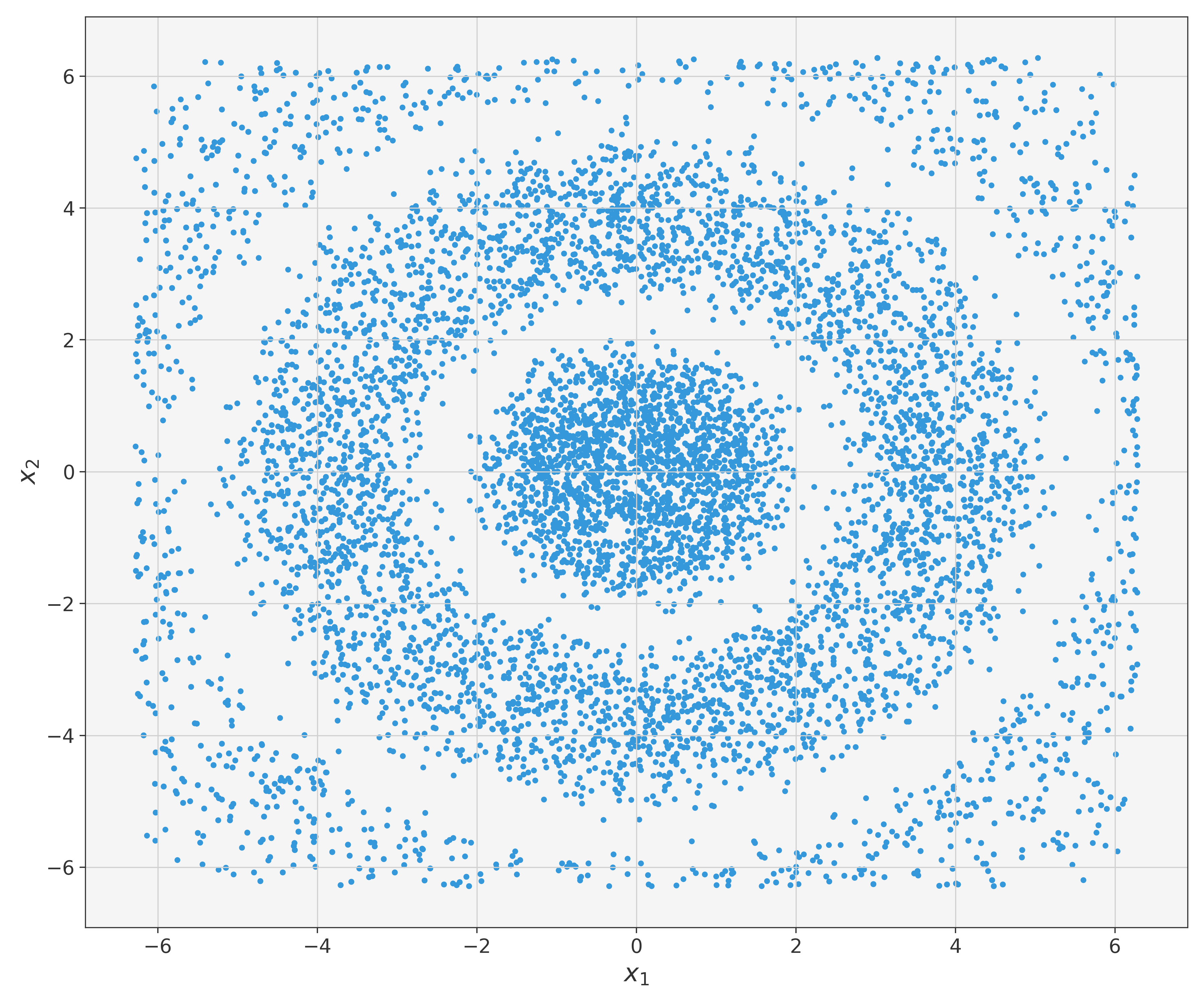}
    \caption{$10^4$ samples from Regenerative Rejection Sampling}
    \label{fig:SyntheticBoundedRRSSamples}
\end{figure}
Figure \ref{fig:SyntheticBoundedRRSSamples} shows that the Regenerative Rejection sampler is able to explore the whole space following the correct distribution of the points, and the performance seems comparable to that of the Random Walk Metropolis algorithm. However, we can also plot the autocorrelation functions for each component to see their decay (Figures \ref{fig:SyntheticBoundedRRSACF1} and \ref{fig:SyntheticBoundedRRSACF2}).
\begin{figure}[htbp]
    \centering
    \begin{minipage}{0.48\textwidth}
        \centering
        \includegraphics[width=\linewidth]{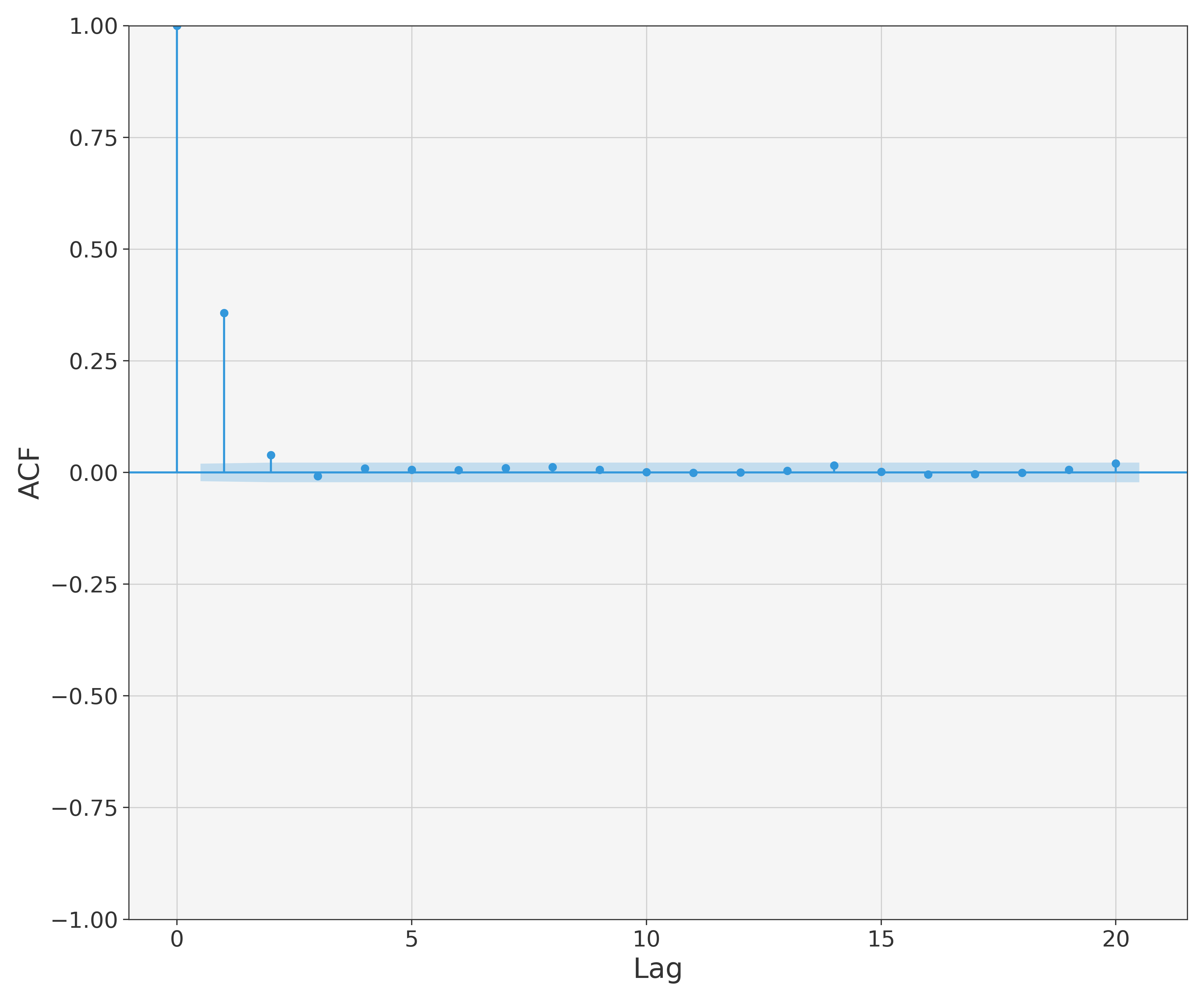}
        \caption{RRS - ACF plot of $1^{st}$ component}
        \label{fig:SyntheticBoundedRRSACF1}
    \end{minipage}
    \hfill
    \begin{minipage}{0.48\textwidth}
        \centering
        \includegraphics[width=\linewidth]{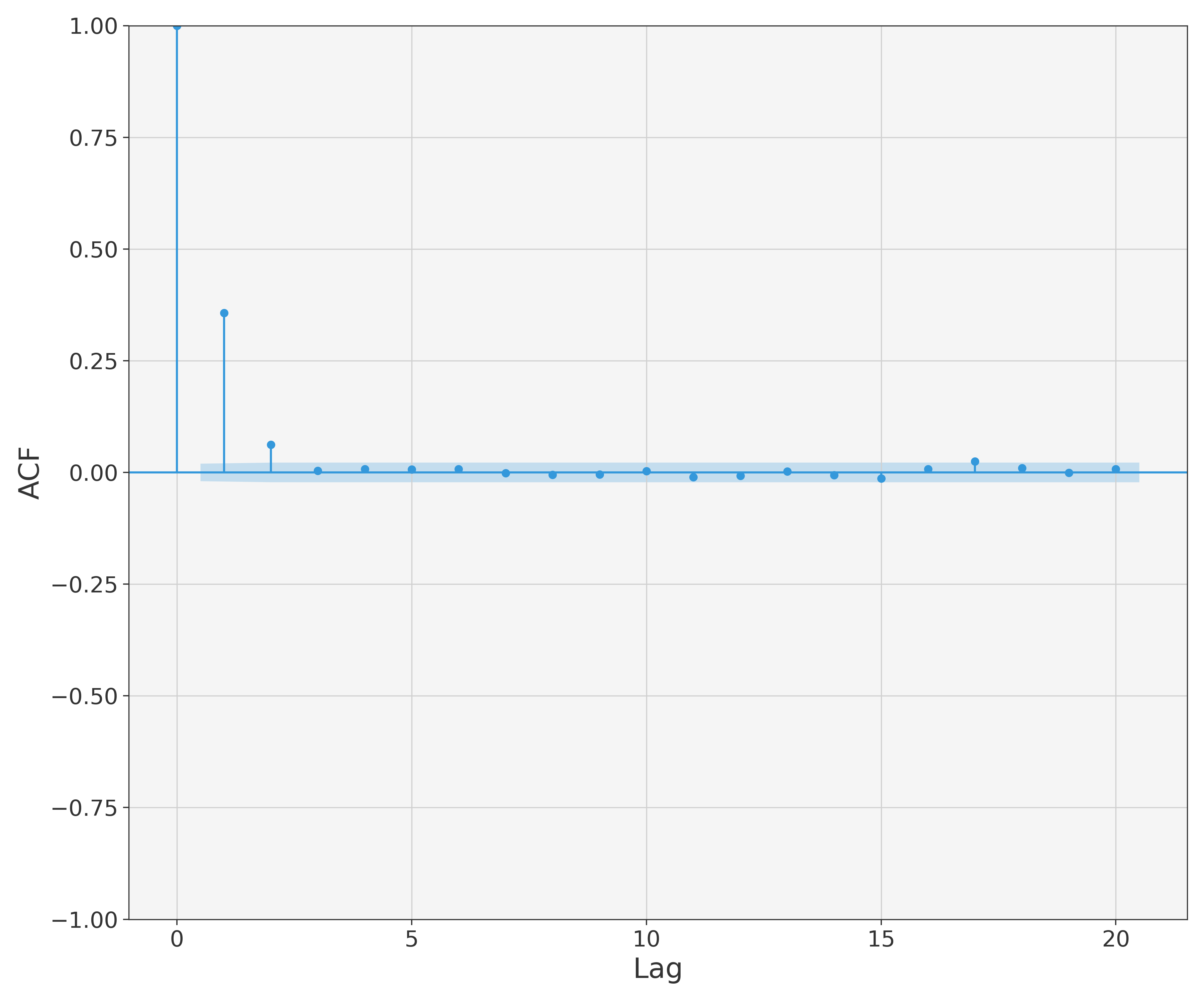}
        \caption{RRS - ACF plot of $2^{nd}$ component}
        \label{fig:SyntheticBoundedRRSACF2}
    \end{minipage}
\end{figure}
In this case we see that, even for a relatively small time-threshold $t=56.91$, the sampler returns random variables that are almost independent, suggesting that the method converges slightly faster compared to the Random Walk Metropolis algorithm.

Note that we expected the Regenerative Rejection Sampling method to perform well in this setting. As seen in Theorem \ref{thm:RRSConvergence}, the process converges geometrically in total variation distance if the cycle length distribution has exponential moments. For this example, since the likelihood ratio $f_\propto(x_1,x_2)/g(x_1,x_2)$ is a continuous function defined on a compact support, it is bounded by a constant, say, $c$, and hence
\[
\mathbb{E}[e^{\eta W}] = \mathbb{E}[e^{\eta f_\propto(X_1,X_2)/g(X_1,X_2)}]\leq\mathbb{E}[e^{\eta c}]<\infty,
\]
which guarantees geometric convergence.

\subsection{Random Walk Metropolis - Unbounded domain}

What happens if instead we want to sample from the same target density, but defined on the whole real plane $\mathbb{R}^2$?
In the following paragraphs we analyze this situation using the Random Walk Metropolis algorithm, as well as the Regenerative Rejection Sampling method, with the same Laplace-type proposal, but defined on the unbounded domain $\mathbb{R}^2$. In Figure \ref{fig:SyntheticUnboundedSurface}, we can see a surface plot of the unnormalized target density on the unbounded domain.
\begin{figure}[h]
    \centering
    \includegraphics[width=0.55\linewidth]{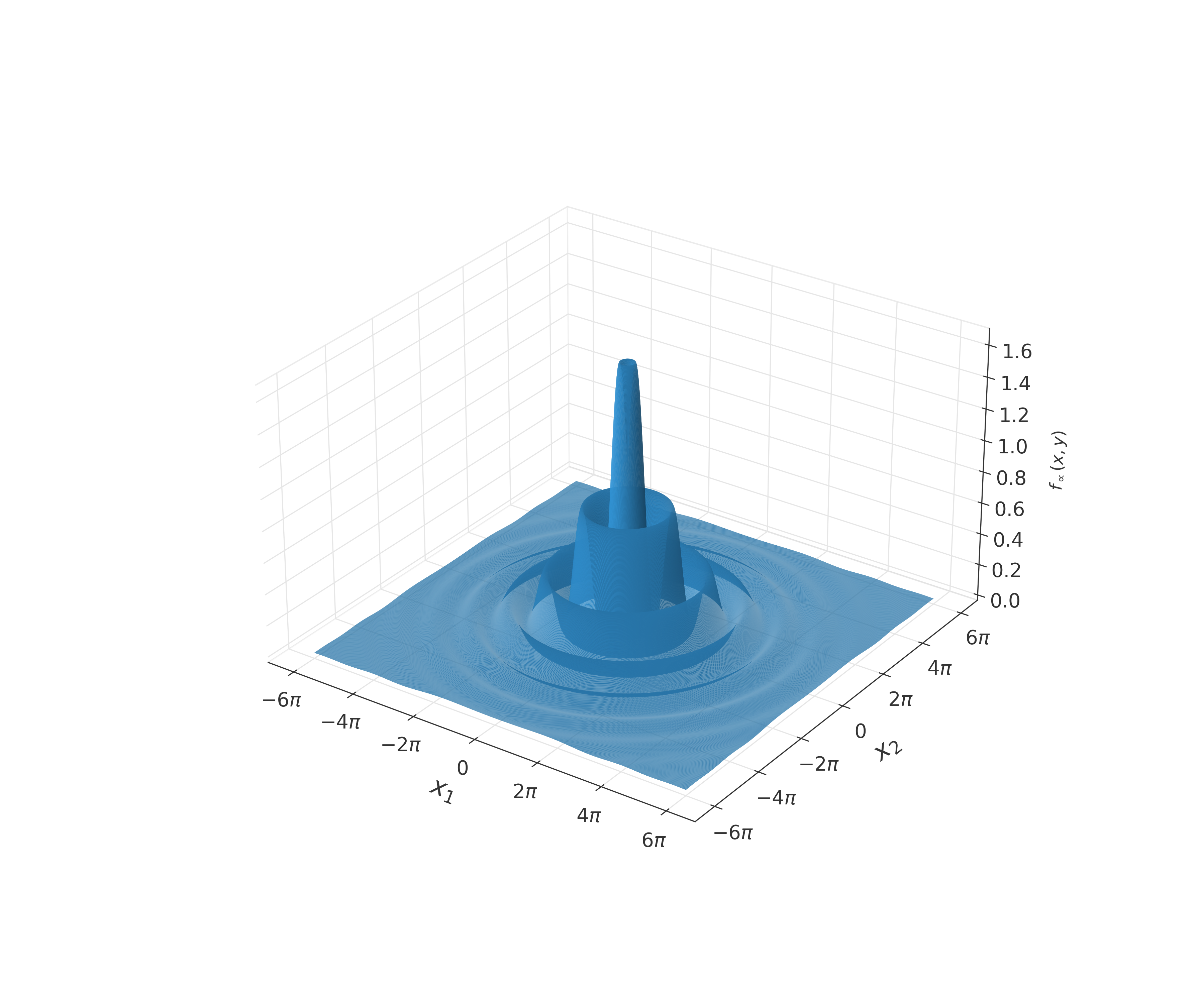}
    \caption{Surface plot of unnormalized $f_\propto(x_1,x_2)$ defined on $\mathbb{R}^2$}
    \label{fig:SyntheticUnboundedSurface}
\end{figure}

As was done for the previous analysis, we start by plotting $10^4$ \emph{consecutive} samples generated with the Random walk Metropolis algorithm after a burn-in period of 1000. From Figure \ref{fig:UnboundedRWMSamples}, it is clear that the RWM sampler is still able to explore the space sufficiently well, as the Markov Chain also moves in the remote parts of the support.
\begin{figure}[h]
    \centering
    \includegraphics[width=0.55\linewidth]{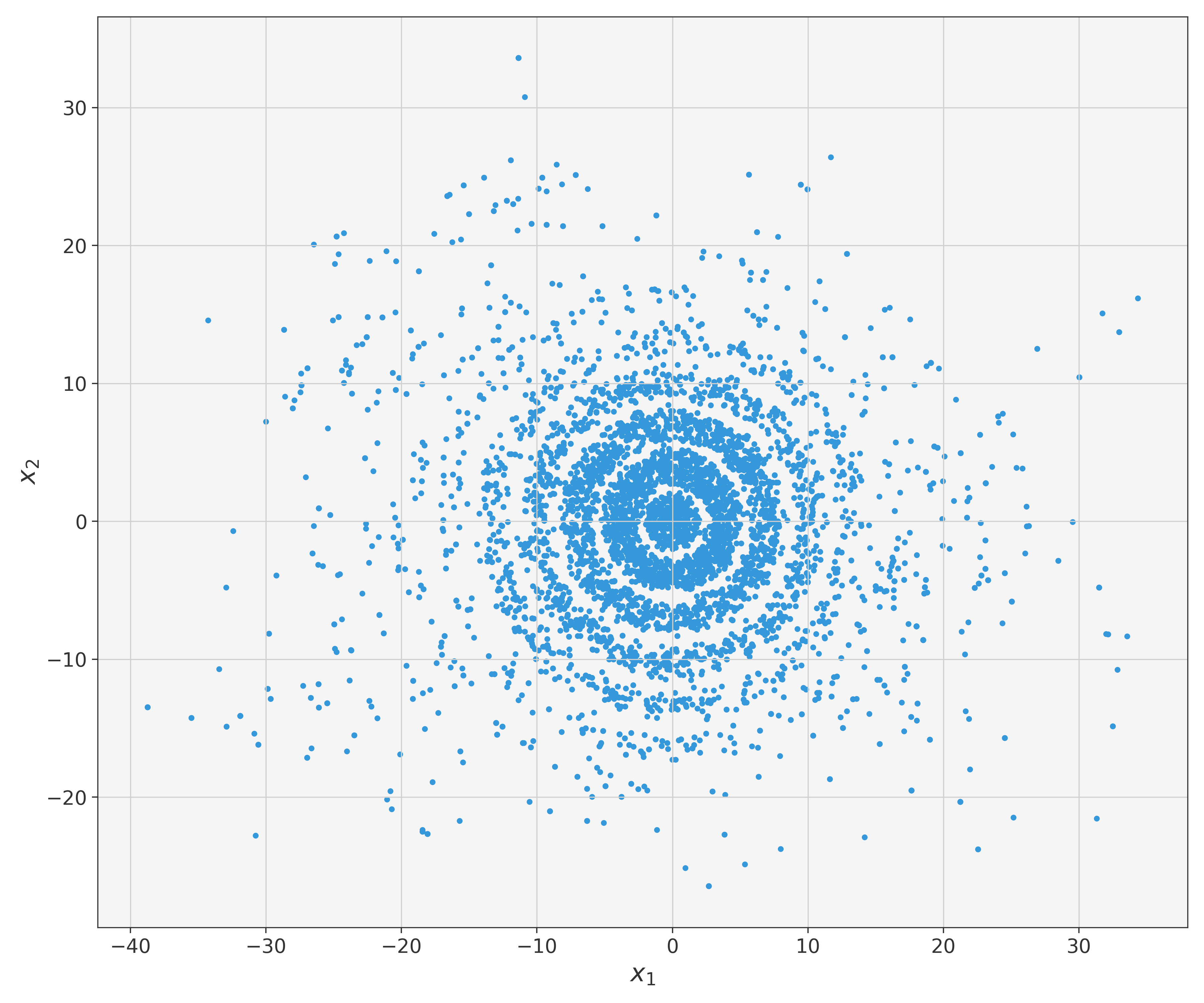}
    \caption{$10^4$ consecutive samples from RWM Algorithm}
    \label{fig:UnboundedRWMSamples}
\end{figure}

However, by plotting the autocorrelation function for each component of the samples (see Figures \ref{fig:SyntheticUnboundedRWMACF1} and \ref{fig:SyntheticUnboundedRWMACF2}), we notice that the samples decorrelate slowly. For this reason, the Random Walk Metropolis algorithm is not the best choice for sampling from the target distribution defined on the unbounded domain $\mathbb{R}^2$. 
\begin{figure}[htbp]
    \centering
    \begin{minipage}{0.48\textwidth}
        \centering
        \includegraphics[width=\linewidth]{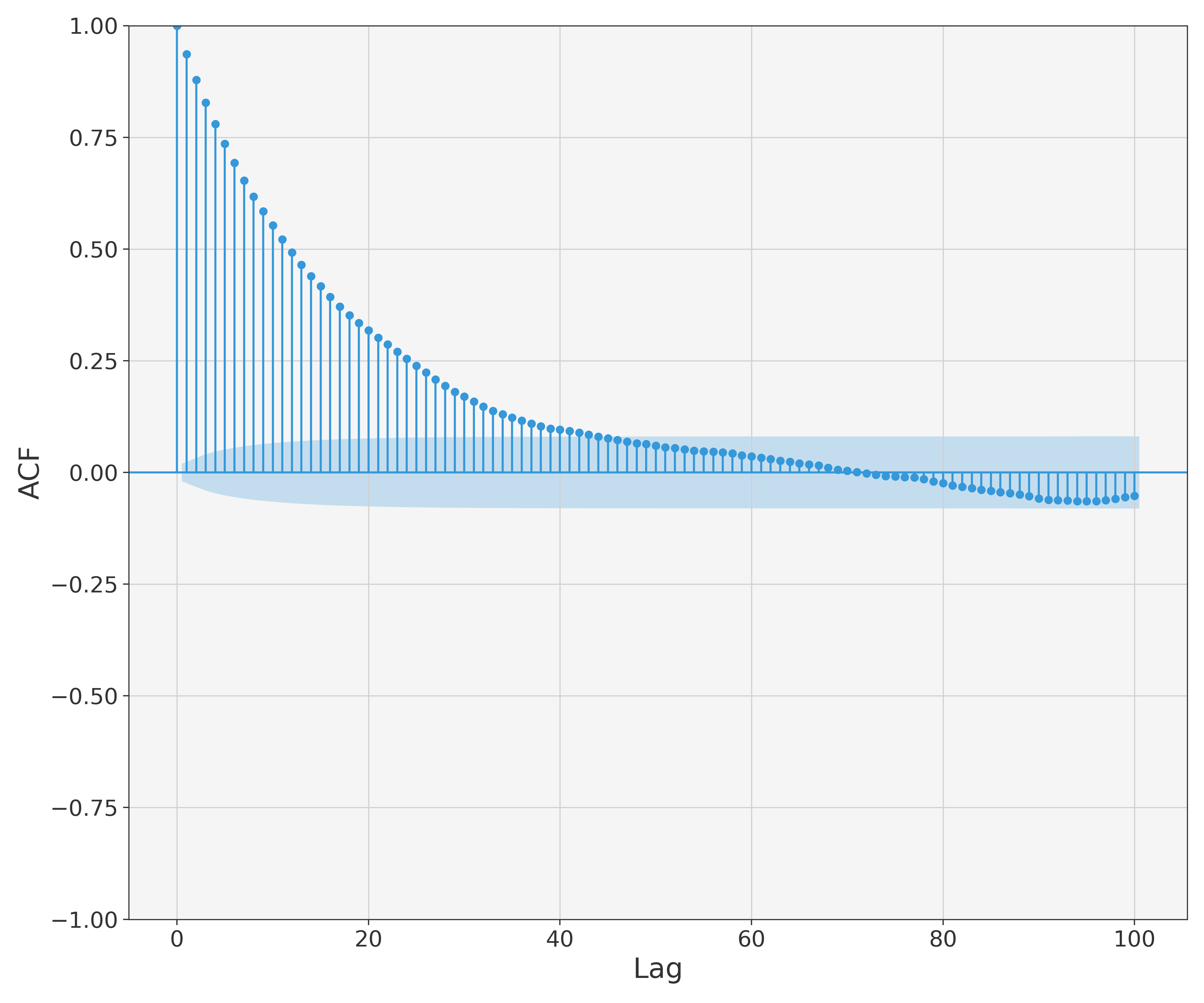}
        \caption{RWM - ACF plot of $1^{st}$ component}
        \label{fig:SyntheticUnboundedRWMACF1}
    \end{minipage}
    \hfill
    \begin{minipage}{0.48\textwidth}
        \centering
        \includegraphics[width=\linewidth]{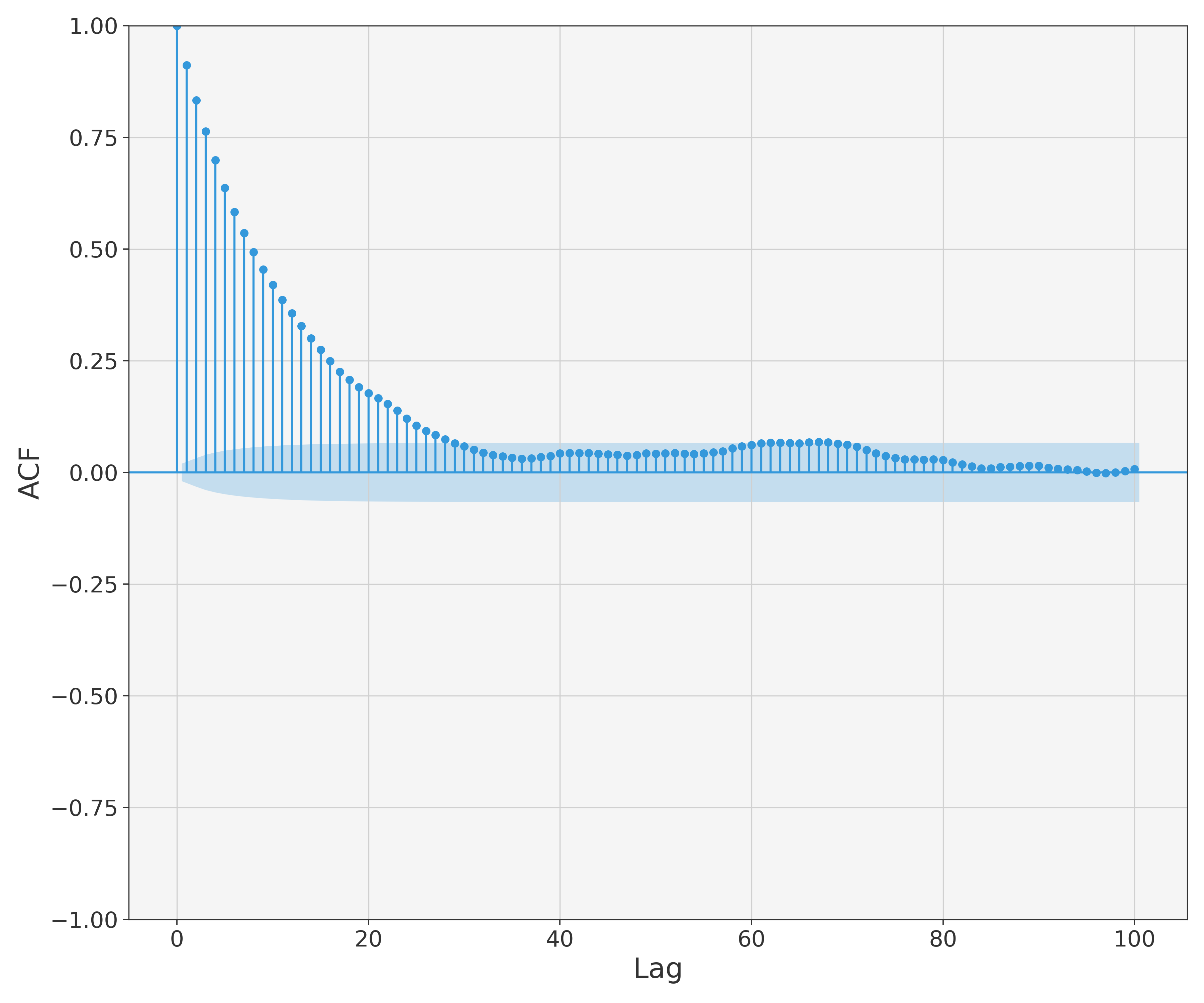}
        \caption{RWM - ACF plot of $2^{nd}$ component}
        \label{fig:SyntheticUnboundedRWMACF2}
    \end{minipage}
\end{figure}

As a positive remark, the theory suggested that we could not expect the RWM algorithm to perform well in this situation, because of the unboundedness of the domain \cite[Theorem 3.1]{MergensenTweedie:MHConvergence}. Nonetheless, the plots show that the method has a reasonably good performance.

\subsection{Regenerative Rejection Sampling - Unbounded domain}
Since $\text{supp}\,(f_\propto)=\mathbb{R}^2$, we should not sample from the truncated proposal distribution \eqref{eq:truncatedLaplace}. For this reason, we need to re-compute the time-threshold $t$. For the sake of completeness, we also inspected the KDE of the cycle length distribution, which is shown in Figure \ref{fig:SyntheticUnboundedRRSCycleLength}.
\begin{figure}[h!]
    \centering
    \includegraphics[width=0.55\linewidth]{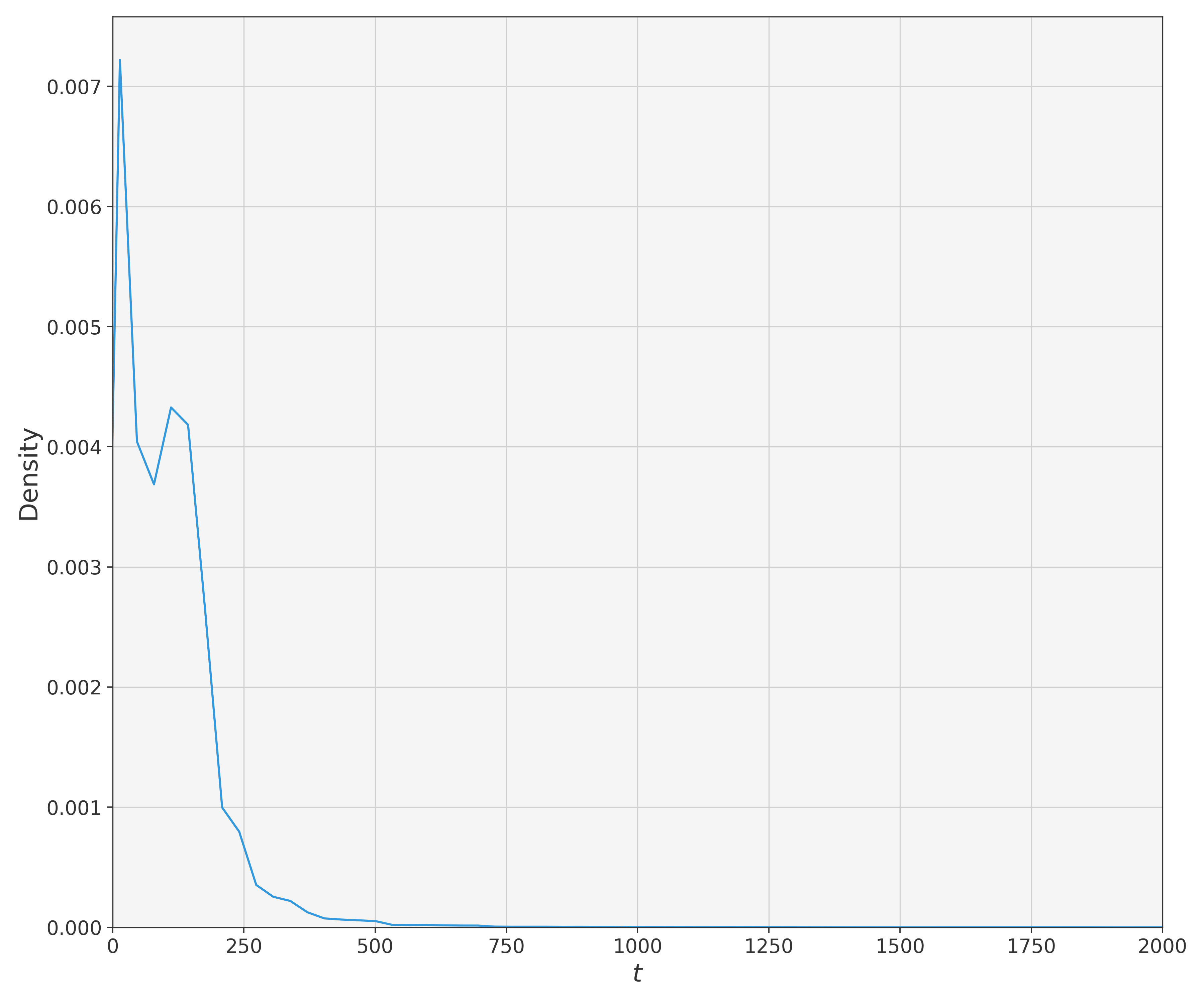}
    \caption{KDE of cycle length distribution}
    \label{fig:SyntheticUnboundedRRSCycleLength}
\end{figure}
As most of the mass is in $(0,250]$, a possible value for the time threshold could be $t=125$, to balance between efficiency and accuracy.

However, as was previously explained, we should set $t=\frac{N_{RWM}+burnin}{N_{RRS}}\mathbb{E}[W]$. In this instance the value is $t=111.1$, which is, again, close to the one that could have been chosen via a visual inspection of the cycle length density.

Let us turn to the simulation results. Similarly to the Random Walk Metropolis, the sampler explores well the support of the target distribution (see Figure \ref{fig:SyntheticUnboundedRRSSamples}), but the procedure is less effective than how it was for the RWM. 
\begin{figure}[htpb]
    \centering
    \includegraphics[width=0.55\linewidth]{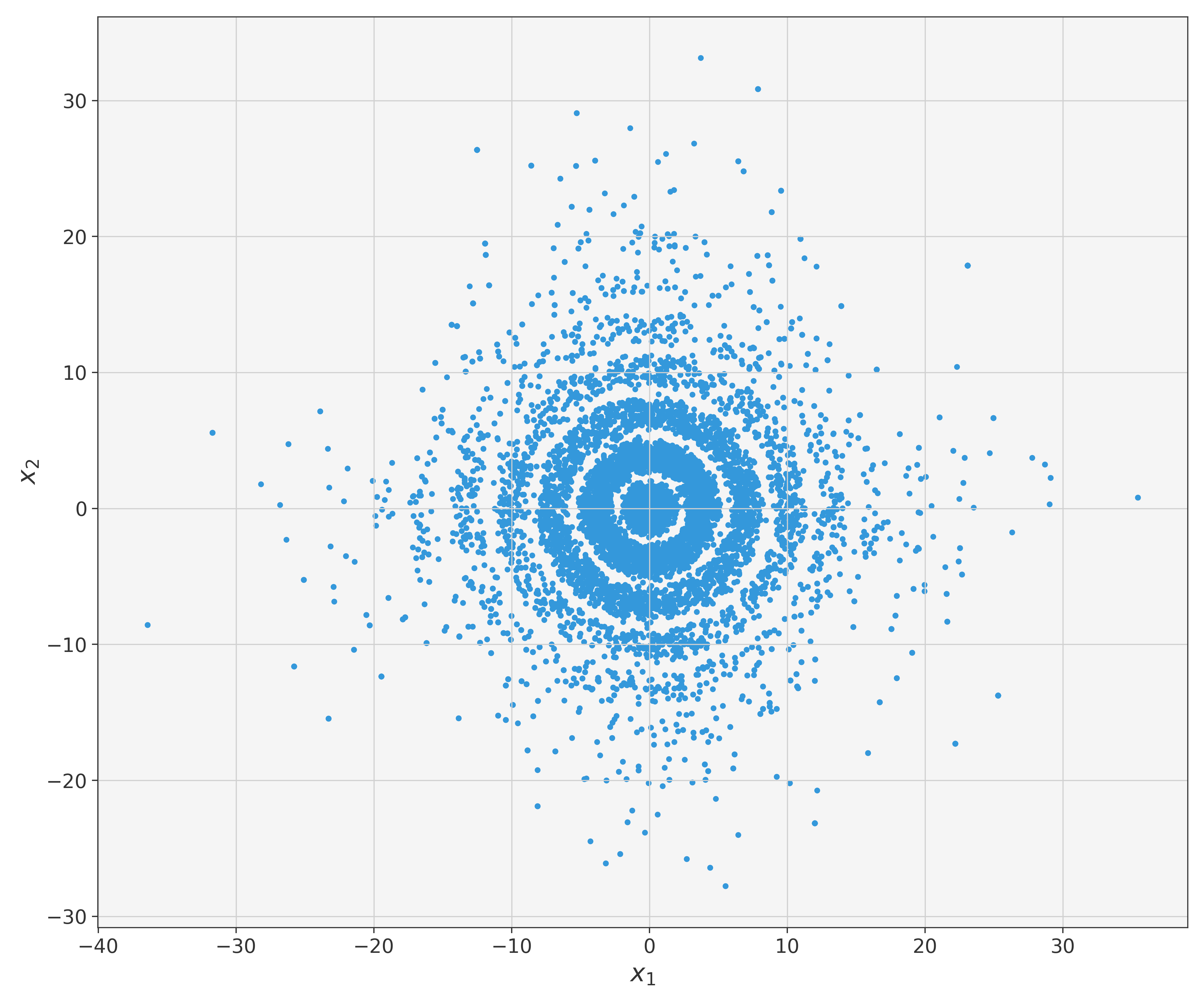}
    \caption{$10^4$ samples from Regenerative Rejection Sampling}
    \label{fig:SyntheticUnboundedRRSSamples}
\end{figure}
Here we notice that the sampler captures the major oscillations of the target density, and the process explores a part of the remote areas of the support of the target, i.e. the ones close to the two axes $\{x=0\}$ and $\{y=0\}$. This is a clear example of one of the limitations of the RRS method: it is highly dependent on the choice of proposal distribution (the one we use here has more mass along the axes). In this case we see that, even if the sampling procedure is still satisfactory, the method does not explore effectively the tails of the target distributions.

Additionally, we also inspect the ACF plots (Figures \ref{fig:SyntheticUnboundedRRSACF1} and \ref{fig:SyntheticUnboundedRRSACF2}), as in the previous examples.
\begin{figure}[htbp]
    \centering
    \begin{minipage}{0.48\textwidth}
        \centering
        \includegraphics[width=\linewidth]{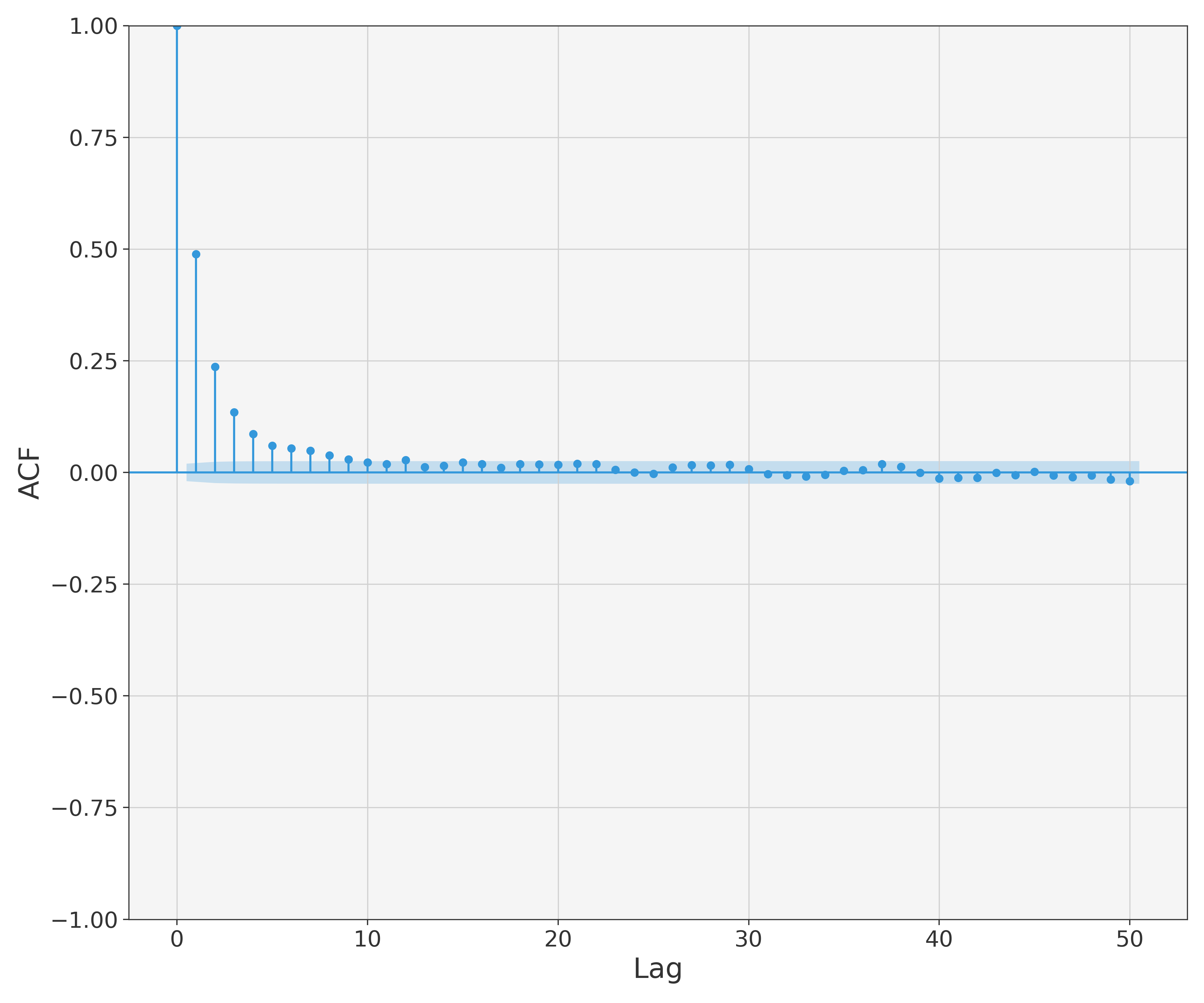}
        \caption{RRS - ACF plot of $1^{st}$ component}
        \label{fig:SyntheticUnboundedRRSACF1}
    \end{minipage}
    \hfill
    \begin{minipage}{0.48\textwidth}
        \centering
        \includegraphics[width=\linewidth]{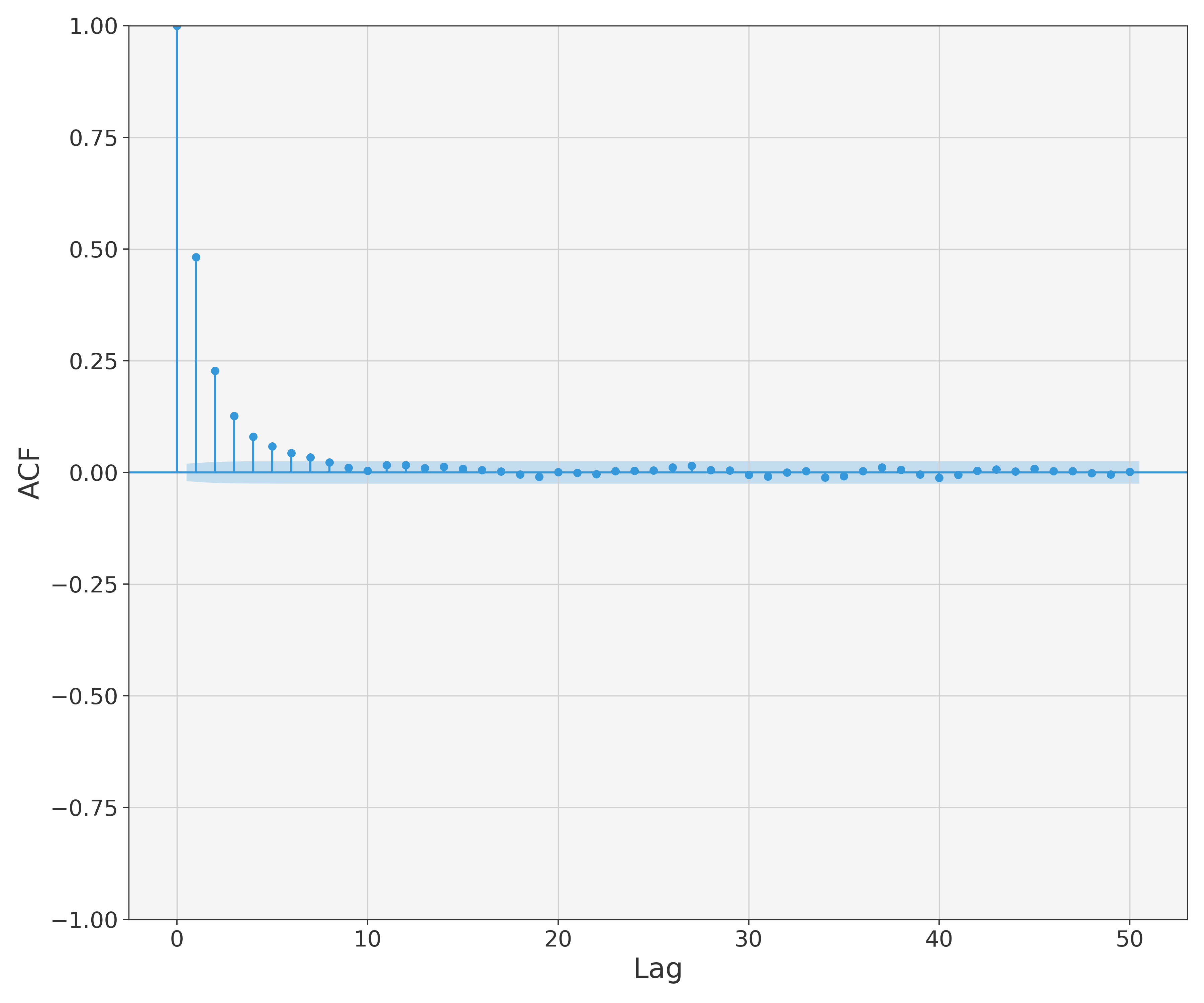}
        \caption{RRS - ACF plot of $2^{nd}$ component}
        \label{fig:SyntheticUnboundedRRSACF2}
    \end{minipage}
\end{figure}
There is some correlation in the first lags, but it is not worrying, as the autocorrelation decays quite fast (in less than ten lags). This, together with the scatter plot of the samples, suggests a good performance of the RRS method even in this more complex setting, even if the extreme tails of the target distribution are not explored effectively.

\subsection{A Note on Efficiency}
In the previous sections, we showed that the Regenerative Rejection Sampling method had, generally, a better performance than the Random Walk Metropolis algorithm on the specific synthetic applications that were considered. Nonetheless, one could still wonder whether the RRS method is more computationally expensive than the Random Walk Metropolis sampler.

A heuristic way to compare the cost of the two methods is to compute the average number of samples generated each second by each algorithm, and the algorithm that who returns more samples per second has the lower computational cost. Clearly, for the RRS method, there is a double source of randomness, coming from both the CPU time (which varies in each run), and the number of generated samples. On the other hand, the RWM algorithm generates a fixed number samples, meaning that the only source of randomness comes from the CPU time.

To compute this quantity, we ran the two methods for $100$ independent times. For each run, we calculated the number of generated samples (fixed to $11000$ for the RWM and random for the RRS) and the elapsed CPU time. With these two values, we were able to compute the number of samples per second for each run. The desired quantity could then be obtained by calculating their sample mean. The results can be seen in Table \ref{tab:SamplesSecondSynthetic}.

\begin{table}[htbp]
\centering
\caption{Average number of samples per second -- RRS vs RWM}
\label{tab:SamplesSecondSynthetic}
\begin{tabular}{p{4cm} p{2cm} p{2cm}}
\toprule
 & \textit{RRS} & \textit{RWM} \\
\midrule
\textit{Avg. Samples/Second, Bounded Domain} & 20882 & 18720 \\
\addlinespace
\textit{Avg. Samples/Second, Unbounded Domain} & 43745 & 28264 \\
\bottomrule
\end{tabular}
\end{table}

From this analysis, it is possible to see that, especially in the unbounded domain case, the RRS method is less computationally expensive than the RWM algorithm. Hence, since it also shows better performances, in our opinion, it should be the preferred method, at least in this synthetic setting.

\section{Bayesian Probit Regression}
Consider a sequence of $n$ binary random variables $y_1,\dots,y_n$. Suppose that the $y_i$'s are observed and follow a Bernoulli distribution with parameter $p_i = \Phi(\bm{x}_i^T\bm{\beta})$, where $\bm{x_i} = (x_i^1,\dots,x_i^k)^T$ is a vector of known covariates, $\bm{\beta}$ is an unknown $k\times 1$ vector of parameters, and $\Phi(\cdot)$ is the cumulative distribution function of a standard normal random variable. The likelihood of the data is
\begin{equation}
    \mathcal{L}(\bm{y}|\bm{\beta}) = \prod_{i=1}^n (\Phi(\bm{x}_i^T\bm{\beta}))^{y_i} (1-\Phi(\bm{x}_i^T\bm{\beta}))^{1-y_i} = \prod_{i=1}^n \Phi((2y_i-1)\bm{x}_i^T\bm{\beta}).
\end{equation}
The classical approach to fitting this \emph{Probit regression model} is to use maximum likelihood estimation, hence inferences will be based on the linked asymptotic theory. This procedure is shown to be significantly biased for small sample sizes \cite{Griffiths:SmallSampleProperties}.

However, it is possible to employ Bayesian methods to obtain samples from the posterior distribution of the parameters
\[
\pi(\bm{\beta}|\bm{y})\propto\mathcal{L}(\bm{y}|\bm{\beta})\pi(\bm{\beta}),
\]
where $\pi(\bm{\beta})$ is the prior distribution of the parameters, which represents the information that is known on the parameters before observing the data.

Sampling from the posterior distribution can be done by
\begin{itemize}
    \item Exploiting a latent-variable structure of the model, and using the Gibbs Sampler MCMC method;
    \item Using the Regenerative Rejection Sampling method.
\end{itemize}
In the following pages we describe the two approaches and compare their performance on a specific dataset.

\subsection{Data Augmentation and Gibbs Sampling}
The data augmentation approach has been first presented by \cite{TannerWong:DataAugmentation}, and has been applied to the context of Bayesian Analysis of binary data by \cite{AlbertChib:BayesBinary}. It consists of introducing $n$ latent variables, $z_1,\dots,z_n$, with $z_i\sim\mathcal{N}(\bm{x}_i^T\bm{\beta},1)$ distribution, of which we only observe the sign, i.e. $y_i = 1$ if $z_i>0$ and $y_i=0$ if $z_i\leq 0$. It can be easily shown that the $y_i$'s follow a Bernoulli distribution with parameter $\Phi(\bm{x}_i^T\bm{\beta})$ \cite{AlbertChib:BayesBinary}.

In practice, notice that
\[
\Phi((2y_i-1)\bm{x}_i^T\bm{\beta}) = \mathbb{E}\left[\mathbb{I}\left(z\leq(2y_i-1)\bm{x}_i^T\bm{\beta}\right)\right],
\]
where $z\sim\mathcal{N}(0,1)$. Hence, we can think of $\mathcal{L}(\bm{y}|\bm{\beta})$ as the marginalization of the following joint pdf (here $\varphi(\cdot)$ is the pdf of a $\mathcal{N}(0,1)$ random variable):

\[
\pi(\bm{y},\bm{z}|\bm{\beta}) = \prod_{i=1}^n \varphi(z_i)\mathbb{I}\left(z\leq(2y_i-1)\bm{x}_i^T\bm{\beta}\right),
\]
This is easily shown, since
\[
\begin{split}
\mathcal{L}(\bm{y}|\bm{\beta}) &= \int_{\mathbb{R}^n} \prod_{i=1}^n \varphi(z_i)\mathbb{I}\left(z\leq(2y_i-1)\bm{x}_i^T\bm{\beta}\right) \text{d}\bm{z}\\
&=\prod_{i=1}^n \int_{\mathbb{R}} \varphi(z_i)\mathbb{I}\left(z\leq(2y_i-1)\bm{x}_i^T\bm{\beta}\right) \text{d}z_i\\
&= \prod_{i=1}^n \Phi((2y_i-1)\bm{x}_i^T\bm{\beta})).
\end{split}
\]
Similarly, the posterior $\pi(\bm{\beta}|\bm{y})=\mathcal{L}(\bm{y}|\bm{\beta})\pi(\bm{\beta})$ can also be thought as the marginal of the joint pdf
\begin{equation}
\pi(\bm{\beta},\bm{z}|\bm{y}) \propto \prod_{i=1}^n \varphi(z_i-(2y_i-1)\bm{x}_i^T\bm{\beta})\mathbb{I}(z_i\geq 0)\,\cdot\, \pi(\bm{\beta}),
\end{equation}
since
\[
\begin{split}
    \pi(\bm{\beta}|\bm{y}) &= \int_{\mathbb{R}^n}\prod_{i=1}^n \varphi(z_i-(2y_i-1)\bm{x}_i^T\bm{\beta})\mathbb{I}(z_i\geq 0)\,\cdot\, \pi(\bm{\beta}) \text{d}\bm{z}\\
    &= \pi(\bm{\beta})\prod_{i=1}^n\int_{\mathbb{R}} \varphi(t_i)\,\mathbb{I}(t_i\geq - (2y_i-1)\bm{x}_i^T\bm{\beta})\text{d}t_i\\
    &= \pi(\bm{\beta})\prod_{i=1}^n (1-\Phi(-(2y_i-1)\bm{x}_i^T\bm{\beta})\\
    &= \pi(\bm{\beta})\prod_{i=1}^n \Phi((2y_i-1)\bm{x}_i^T\bm{\beta}).
\end{split}
\]

The idea of \cite{AlbertChib:BayesBinary} is then to sample from the augmented posterior $\pi(\bm{\beta},\bm{z}|\bm{y})$, instead of the marginal $\pi(\bm{\beta}|\bm{y})$. The natural way to do it is to sample from the two conditional distributions in an iterative fashion, i.e. via the Gibbs Sampler MCMC method. In other words, at the $i^{th}$ iteration, we sample $\bm{z}_i\sim\pi(\bm{z}|\bm{y},\bm{\beta}_{i-1})$, and $\bm{\beta}_i\sim\pi(\bm{\beta}|\bm{y},\bm{z}_i)$, always accepting the proposed steps of the chain.

However, to be able to do it, we first need to compute the two conditional distributions. The pdf of the distribution of $\bm{z}$ given $\bm{\beta}$ is straightforwardly defined as 
\[
\pi(\bm{z}|\bm{y},\bm{\beta}) \propto \prod_{i=1}^n \varphi(z_i-(2y_i-1)\bm{x}_i^T\bm{\beta})\mathbb{I}(z_i\geq 0),
\]
which can be written in vector form as
\[
\pi(\bm{z}|\bm{X},\bm{\beta}) \propto \varphi(\bm{z}-\bm{X}\bm{\beta})\mathbb{I}(\bm{z}\geq\bm{0}),
\]
where
\[
\bm{X}\colon\!\!\!=((2y_1-1)\bm{x}_1,\dots,(2y_n-1)\bm{x}_n)^T.
\]
From the above formulas, we easily deduce that $\bm{z}|\bm{X},\bm{\beta}$ follows a $n$-dimensional truncated normal distribution of mean $\bm{X}\bm{\beta}$ and covariance $I_n$.

On the other hand, the conditional distribution of $\bm{\beta}$ given $\bm{X},\bm{z}$ is found by setting a specific prior distribution on the parameters $\bm{\beta}$, and by using standard linear model results \cite{AlbertChib:BayesBinary}. Let us set $\pi(\bm{\beta})= \mathcal{N}_k(\bm{0}, \sigma^2I_k)$ to be a multivariate centered normal prior. Hence,
\[
\bm{\beta}|\bm{X},\bm{z} \sim \mathcal{N}(\Sigma\bm{X}^T\bm{z}, \Sigma),
\]
where $\Sigma^{-1}\colon\!\!\!= \bm{X}^T\bm{X} + \sigma^{-2}I_k$.

The Gibbs Sampling scheme is outlined in Algorithm \ref{algo:Gibbs}.
\begin{algorithm}
\caption{Bayesian Probit Regression Gibbs Sampler}
\label{algo:Gibbs}
\begin{algorithmic}[1]
\Require{Initial $\bm{\beta}_0$, prior scale $\sigma$, data $\{ (y_i,\bm{x}_i) \}_{i=1}^m$, chain length $N$}
\State $\bm{X} \gets ((2y_1-1)\bm{x}_1,\dots,(2y_n-1)\bm{x}_n)^T$
\State $\Sigma \gets (\bm{X}^T\bm{X} + \sigma^{-2}I_k)^{-1}$
\For{$j=1,\dots,N$}
    \State Sample $\bm{z}|\bm{X},\bm{\beta}_{j-1}\sim \mathcal{TN}_{(\bm{0},\infty)}(\bm{X}\bm{\beta}_{j-1},I_n)$
    \State Sample $\bm{\beta}_j|\bm{X},\bm{z}\sim\mathcal{N}_k(\Sigma\bm{X}^T\bm{z}, \Sigma)$
\EndFor\\
\Return $\bm{\beta}_1,\dots,\bm{\beta}_N$
\end{algorithmic}
\end{algorithm}

Note that, since we only need to retain the $\bm{\beta}$'s, we do not need to store all the values of $\bm{z}$.

\subsection{Regenerative Rejection Sampling for Bayesian Inference}
The Regenerative Rejection Sampling method, in the same form of Algorithm \ref{algo:SequentialRRS}, does not require any modification to be applied to this specific case. However, we need to choose a suitable proposal distribution $g$.

Taking inspiration from what commonly happens in Bayesian Inference \cite{TierneyKadane:LaplaceApproximation,KassTierneyKadane:LaplaceApproximation}, in approximate Bayesian computations \cite{RueMartinoChopin:ABCINLA, GomezRubioRue:MCMCINLA}, and less often in MCMC contexts \cite{LewisRaftery:LaplaceMCMC, GomezRubioRue:MCMCINLA}, a promising candidate for proposal distribution is the Laplace Approximation of the posterior distribution $\pi(\bm{\beta}|\bm{X})$. It consists of a multivariate normal random variable, centered in the mode of the posterior distribution, i.e. the Maximum A Posteriori (MAP) $\mu_{Laplace}$, with covariance matrix given by $-\alpha^2 H^{-1}$, where $H$ is the Hessian matrix of the log-posterior density evaluated at the MAP and $\alpha$ is a given constant. The purpose of $\alpha$ is that of rescaling the proposal distribution, to obtain different step sizes.

For Probit regression, we have closed forms for both the gradient and the Hessian of the log-posterior density. Let us compute them explicitly. First, in a straightforward notation,
\[
\log \pi(\bm{\beta}|\bm{X}) = \sum_{i=1}^n \log\Phi(\bm{X}_i\bm{\beta}).
\]
From this, we easily compute
\[
\frac{\partial}{\partial\beta_j} \log\pi(\bm{\beta}|\bm{X}) = \sum_{i=1}^n \frac{\varphi(\bm{X}_i\bm{\beta})}{\Phi(\bm{X}_i\bm{\beta})}X_{ij}, \quad j=1,\dots, k.
\]
Hence, the gradient of the log-posterior density is
\begin{equation}
    \begin{split}
        \nabla\log\pi(\bm{\beta}|\bm{X}) &= \left( \sum_{i=1}^n \frac{\varphi(\bm{X}_i\bm{\beta})}{\Phi(\bm{X}_i\bm{\beta})}X_{i1}\,, \dots, \sum_{i=1}^n \frac{\varphi(\bm{X}_i\bm{\beta})}{\Phi(\bm{X}_i\bm{\beta})}X_{ik}\right)\\
        &=\eta(\bm{\beta})^T\bm{X},
    \end{split}
\end{equation}
where $\eta(\bm{\beta})$ is an $n$-dimensional vector such that $\eta(\bm{\beta})_i = \varphi(\bm{X}_i\bm{\beta})/\Phi(\bm{X}_i\bm{\beta})$. Additionally, since
\[
\frac{\partial}{\partial\beta_j}(\nabla\log\pi(\bm{\beta}|\bm{X}))_i = \sum_{m=1}^n \left[ -\bm{X}_m\bm{\beta}\frac{\varphi(\bm{X}_m\bm{\beta})}{\Phi(\bm{X}_m\bm{\beta})} - \left( \frac{\varphi(\bm{X}_m\bm{\beta})}{\Phi(\bm{X}_m\bm{\beta})} \right)^2 \right]X_{jm}X_{mi},
\]
the Hessian of the log-posterior density is
\begin{equation}
    \nabla^2\log\pi(\bm{\beta}|\bm{X}) = \bm{X}^TD\bm{X},
\end{equation}
where $D$ is an $n\times n$ diagonal matrix with $D_{ii} = -\bm{X}_i\bm{\beta}\frac{\varphi(\bm{X}_i\bm{\beta})}{\Phi(\bm{X}_i\bm{\beta})} - \left(\frac{\varphi(\bm{X}_i\bm{\beta})}{\Phi(\bm{X}_i\bm{\beta})}\right)^2$.
These closed formulas can be used for the maximization of the log-posterior density (to obtain the MAP), and to compute exactly the covariance matrix of the proposal distribution.

In this framework, the Regenerative Rejection Sampling's likelihood ratio takes the form
\begin{equation}
    W(\bm{\beta}) \propto e^{\xi - \frac{1}{2\alpha^2}(\bm{\beta}-\mu_{MAP})^TH(\bm{\beta}-\mu_{MAP})}\pi(\bm{\beta})\prod_{i=1}^n \Phi((2y_i-1)\bm{x}_i^T\bm{\beta}),
\end{equation}
where $\xi$ is a scaling constant that can be tuned to avoid having too small or too big cycle lengths without altering the proposal distribution. Since we only need to know the target distribution up to a multiplicative constant, scaling the likelihood ratio does not invalidate the procedure.

At this point, let us see a practical application of these two methods on a real dataset.

\subsection{The Latent Membranous Lupus Nephritis Dataset}
In this section, we will analyze a dataset taken from \cite{VanDykMeng:DataAugmentation}, that can be seen in Table \ref{tab:Lupus}.
\begin{table}[htbp]
\centering
\caption{The Latent Membranous Lupus Nephritis Dataset}
\label{tab:Lupus}
\begin{tabular}{cccccc}
\toprule
 & \multicolumn{5}{c}{\textit{IgA}} \\
\cmidrule(lr){2-6}
\textit{IgG3 -- IgG4} & 0 & .5 & 1 & 1.5 & 2 \\
\midrule
$-3.0$ & 0/1 & --- & --- & --- & --- \\
$-2.5$ & 0/3 & --- & --- & --- & --- \\
$-2.0$ & 0/7 & --- & --- & --- & 0/1 \\
$-1.5$ & 0/6 & 0/1 & --- & --- & --- \\
$-1.0$ & 0/6 & 0/1 & 0/1 & --- & 0/1 \\
$-0.5$ & 0/4 & --- & --- & 1/1 & --- \\
$0.0$  & 0/3 & --- & 0/1 & 1/1 & --- \\
$0.5$  & 3/4 & --- & 1/1 & 1/1 & 1/1 \\
$1.0$  & 1/1 & --- & 1/1 & 1/1 & 4/4 \\
$1.5$  & 1/1 & --- & --- & 2/2 & --- \\
\bottomrule
\end{tabular}
\end{table}
It displays the data of two covariates, which represent specific clinical measurements, that are used to predict the occurrence of \emph{Latent membranous lupus nephritis} (also called \emph{Lupus}) \cite{VanDykMeng:DataAugmentation}. The two covariates are the difference between IgG3 and IgG4, and IgA, where the acronyms stand for immunoglobulin G and A. The datasets comprises measurements from $n=55$ patients, 18 of which have been diagnosed with Latent membranous lupus nephritis. In Table \ref{tab:Lupus}, we see the number of Lupus cases (numerators) and the total number of cases (denominators) for any given combination of the two covariates. We perform the regression by adding an intercept term, bringing the number of parameters to $k=3$. For both Gibbs Sampling and RRS, we simulate $N=10^4$ samples from the posterior distribution $\pi(\bm{\beta}|\bm{X})$. Moreover, we use a flat improper prior, i.e. $\pi(\bm{\beta})\propto 1$.

Let us begin the analysis with Gibbs Sampling. Since we employ an improper prior, which corresponds to taking $\pi(\bm{\beta})= \mathcal{N}_k(\bm{0}, \sigma^2I_k)$ with $\sigma^2=\infty$, the parameters $\bm{\beta}$, given $\bm{z}$ and $\bm{X}$, follow a $\mathcal{N}(\Sigma \bm{X}^T\bm{z},\Sigma)$, where $\Sigma^{-1}=\bm{X}^T\bm{X}$. First of all, we show the boxplots of the marginal distribution of each parameter (Figure \ref{fig:LupusGibbsBoxplot}), together with a scatter-plot of the two covariates IgG3-IgG4 and IgA (Figure \ref{fig:LupusGibbsScatter}).
\begin{figure}[htbp]
    \centering
    \begin{minipage}{0.48\textwidth}
        \centering
        \includegraphics[width=\linewidth]{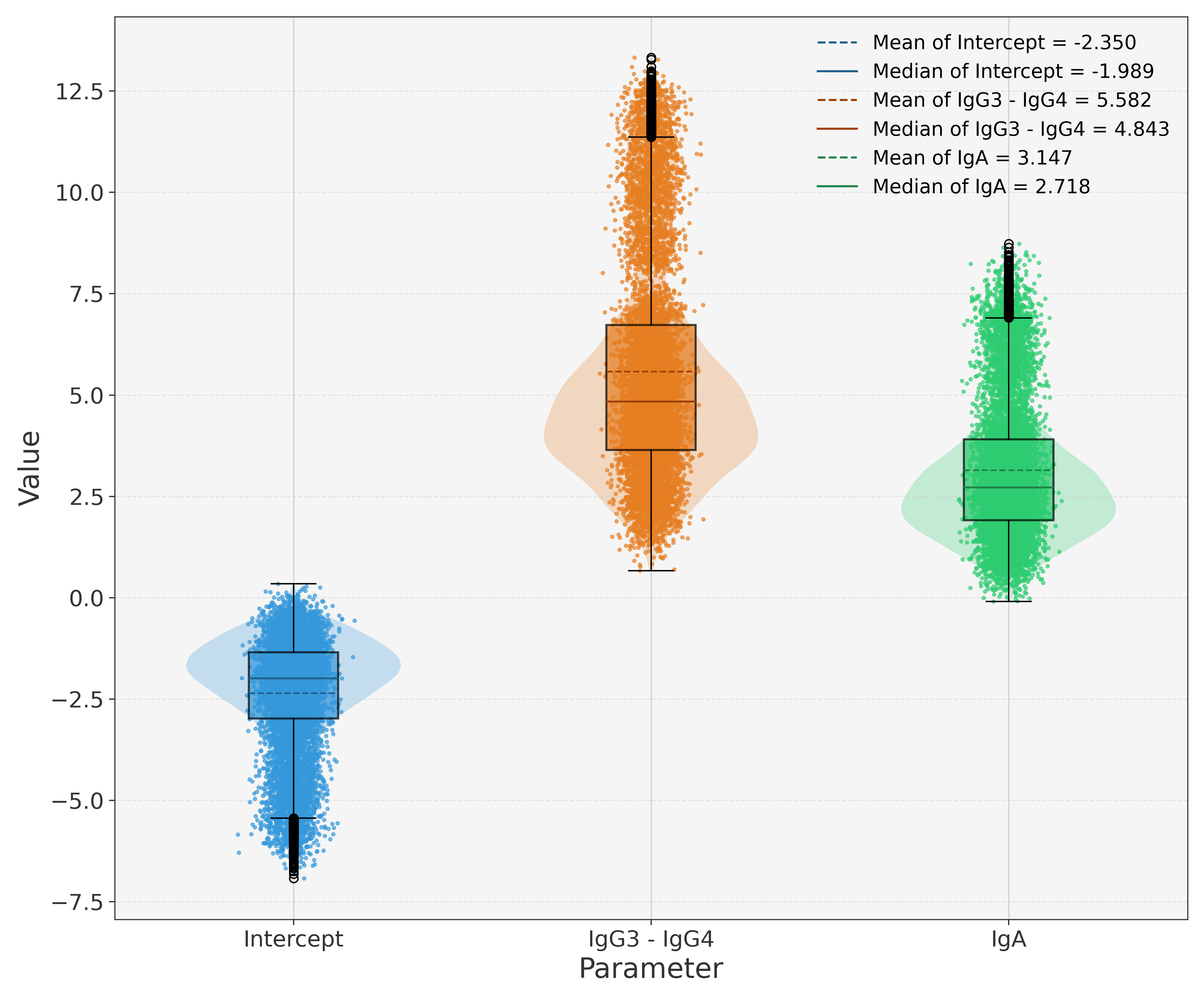}
        \caption{Gibbs - Boxplots of marginal distributions}
        \label{fig:LupusGibbsBoxplot}
    \end{minipage}
    \hfill
    \begin{minipage}{0.48\textwidth}
        \centering
        \includegraphics[width=\linewidth]{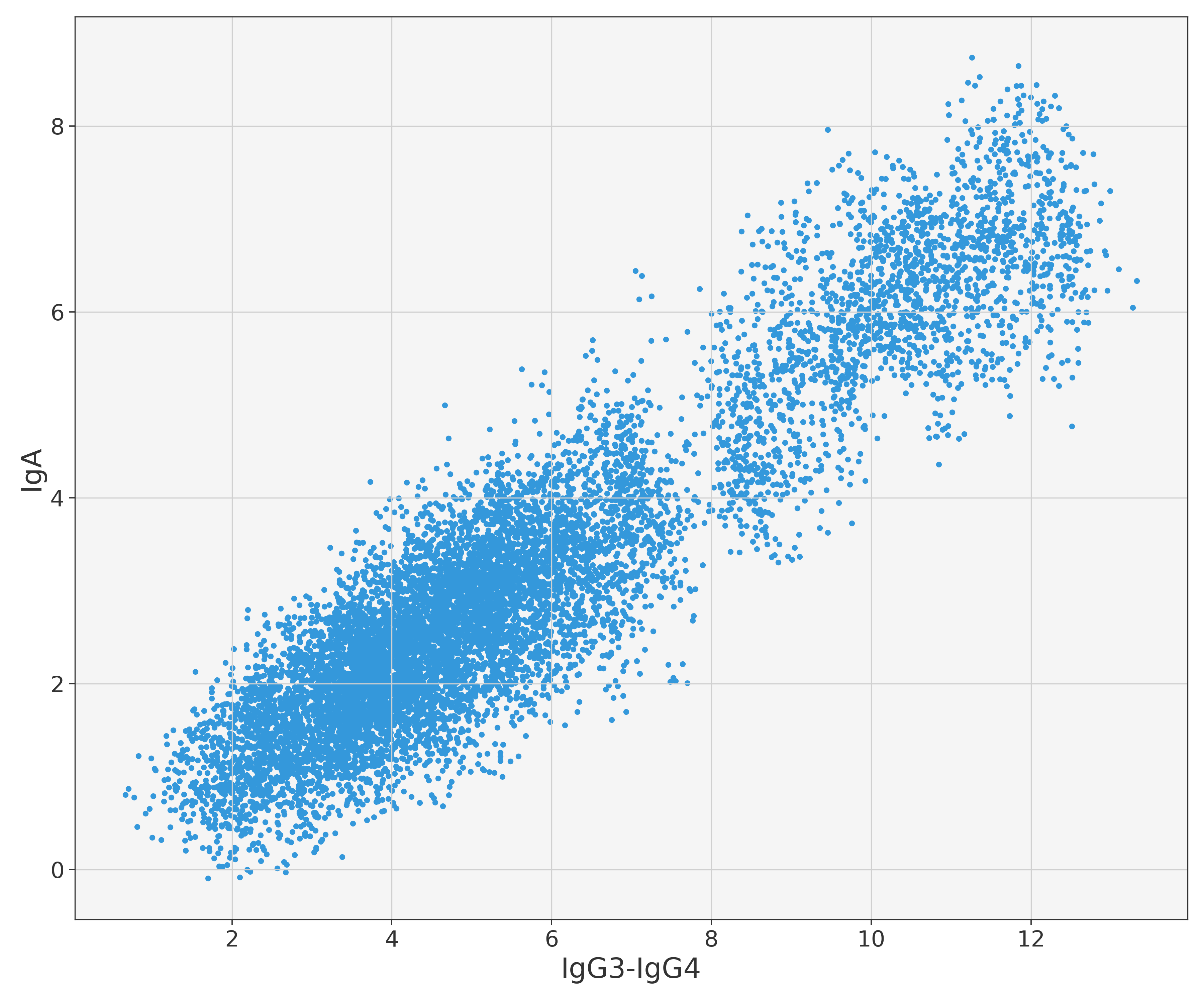}
        \caption{Gibbs - Scatter-plot of covariates}
        \label{fig:LupusGibbsScatter}
    \end{minipage}
\end{figure}

From the two Figures, it is clear that the performance of the Gibbs Sampler is not optimal. The scatter-plot highlights that some areas of the bi-dimensional distribution of the two covariates are not explored, and a similar behavior is also seen in the boxplots. Despite the presence of some outliers, they are not too extreme, meaning that the exploration does not move towards the more remote areas of the support of the marginals. The sub-optimal performance of Gibbs Sampling on this specific dataset is also confirmed by the autocorrelation functions of each component. Here, in Figure \ref{fig:LupusGibbsACF2}, we only show the one relative to the first covariate.
\begin{figure}[htpb]
    \centering
    \includegraphics[width=0.55\linewidth]{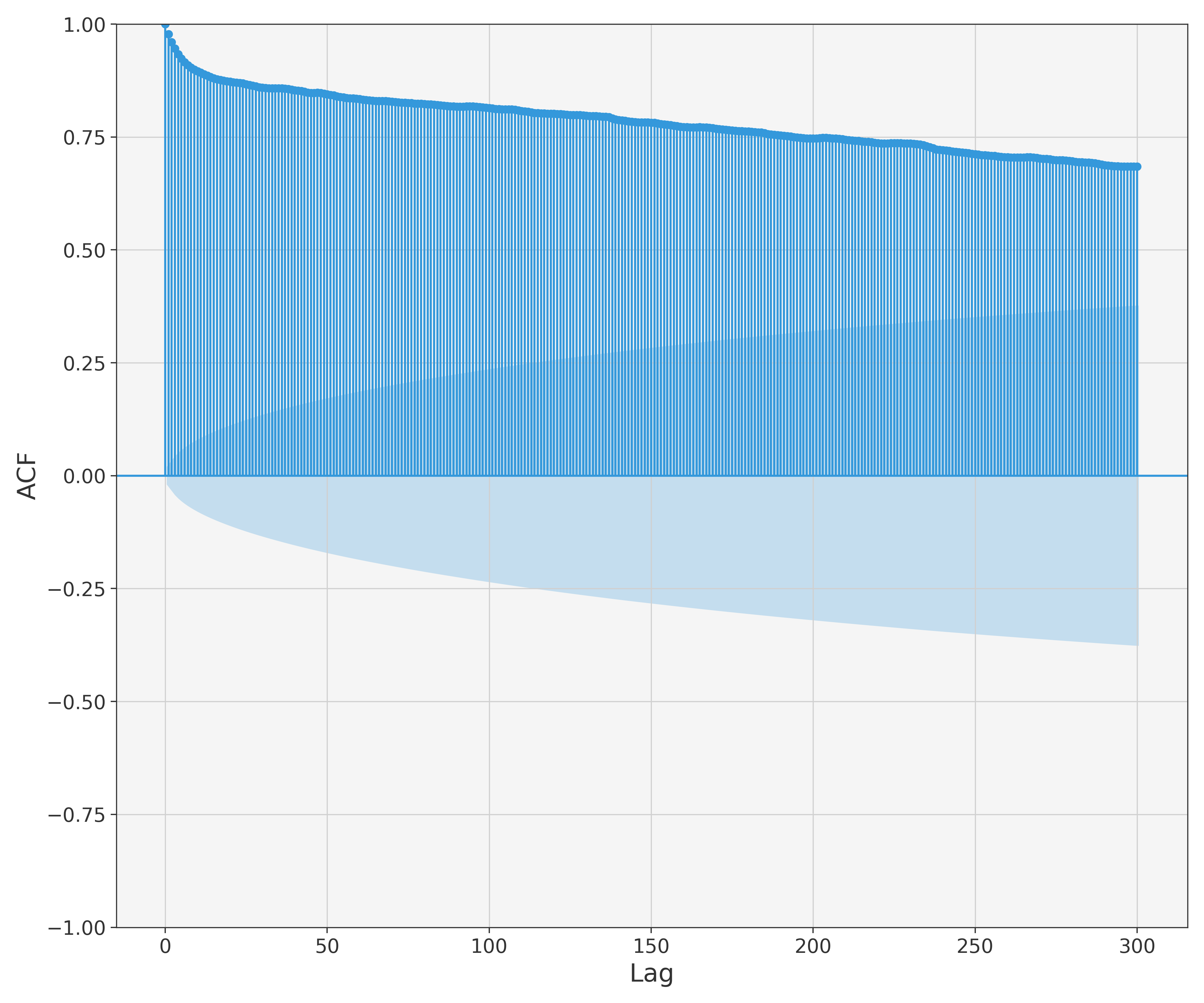}
    \caption{Gibbs - ACF plot of IgG3-IgG4}
    \label{fig:LupusGibbsACF2}
\end{figure}
It is clear that the sampling is particularly inefficient, as even after 300 lags the autocorrelation in the chain is still significantly high. Hence, we can conclude that, on this specific Lupus dataset, the Gibbs Sampling MCMC method does not perform well. 

At this point, we perform the same analysis using the Regenerative Rejection Sampling method, as detailed in the previous section. The maximization of the posterior pdf was done via the \texttt{scipy.optimize.minimize} function, using the Trust-Region Newton method of \cite{MoreSorensen:TrustRegion}, which employed the exact Hessian that was computed explicitly in the previous section.

Moreover, for this application we set $\xi=2$ and $\alpha^2=5$, obtaining the likelihood ratio (recall that $\pi(\bm{\beta})\propto 1$)
\[
    W(\bm{\beta}) \propto e^{2 - \frac{1}{2\cdot5}(\bm{\beta}-\mu_{MAP})^TH(\bm{\beta}-\mu_{MAP})}\prod_{i=1}^n \Phi((2y_i-1)\bm{x}_i^T\bm{\beta}).
\]
The next hyperparameter to set was the time-threshold $t$. Once again for the sake of completeness, we decided to simulate $10^6$ independent samples from the cycle length distribution (as they only require sampling from the proposal distribution), and plot their KDE, which can be seen in Figure \ref{fig:LupusRRSCycleLength}. A visual inspection of the empirical density suggests we could choose $t=1$, since the majority of the mass of the distribution is close to zero.
\begin{figure}[htpb]
    \centering
    \includegraphics[width=0.55\linewidth]{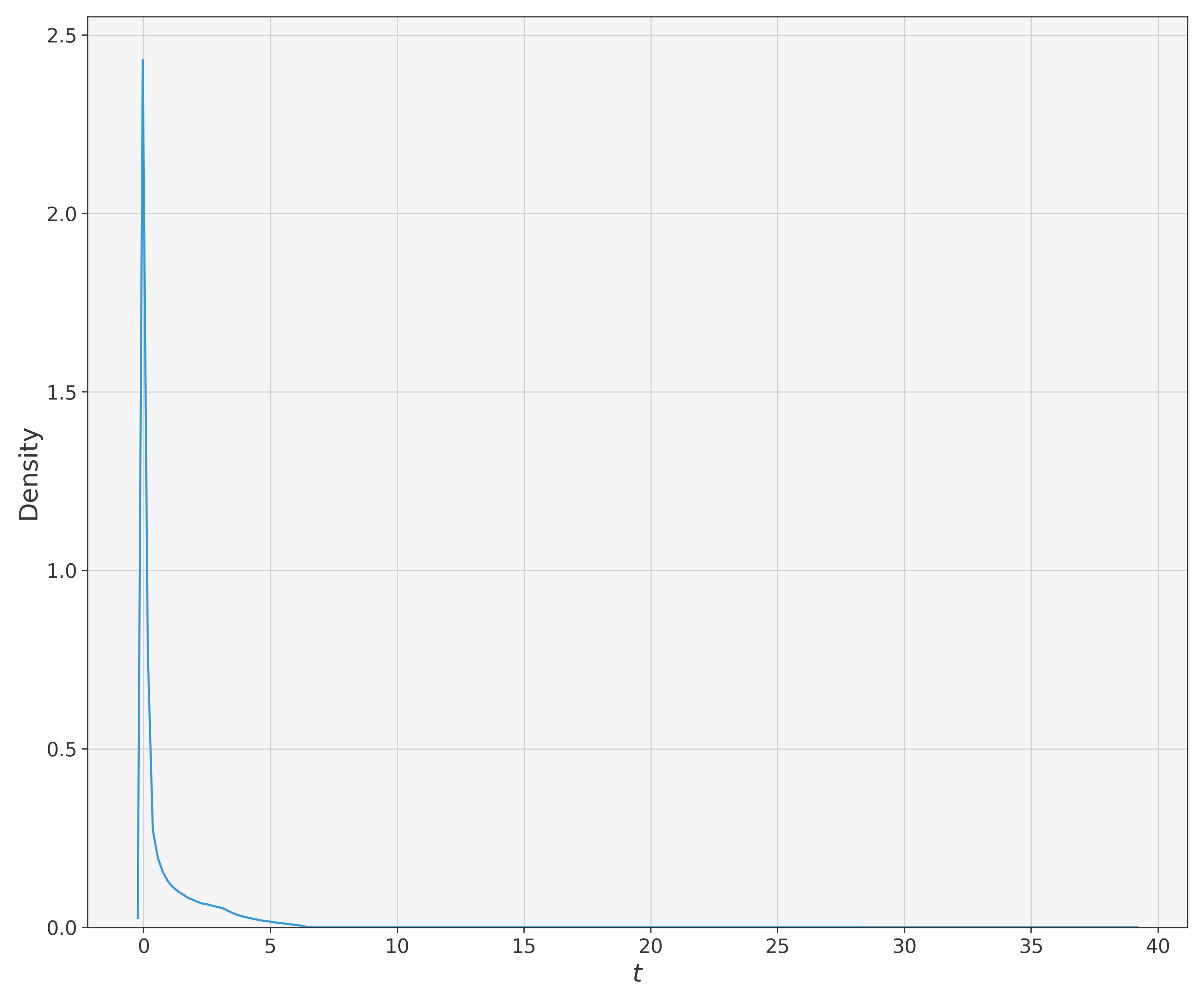}
    \caption{KDE of cycle length distribution}
    \label{fig:LupusRRSCycleLength}
\end{figure}

However, as in the previous synthetic examples, we choose $t=\frac{N_{RWM}+burnin}{N_{RRS}}\mathbb{E}[W]$, where $\mathbb{E}[W]$ can be estimated from the $10^6$ samples that were generated to compute the KDE. In this specific example, $t=0.7780$.

As was previously done for the Gibbs Sampler, we now show the boxplots of the marginal distributions of the three parameters (Figure \ref{fig:LupusRRSBoxplot}), and the scatter-plot of the two covariates (Figure \ref{fig:LupusRRSScatter}). Clearly, the algorithm now explores the sample space much more effectively, as seen in the scatter-plot. And the boxplots highlight that the sampling procedure is able to inspect the extreme tails of the marginal distributions of the parameters, as shown by the numerous outliers.
\begin{figure}[htbp]
    \centering
    \begin{minipage}{0.48\textwidth}
        \centering
        \includegraphics[width=\linewidth]{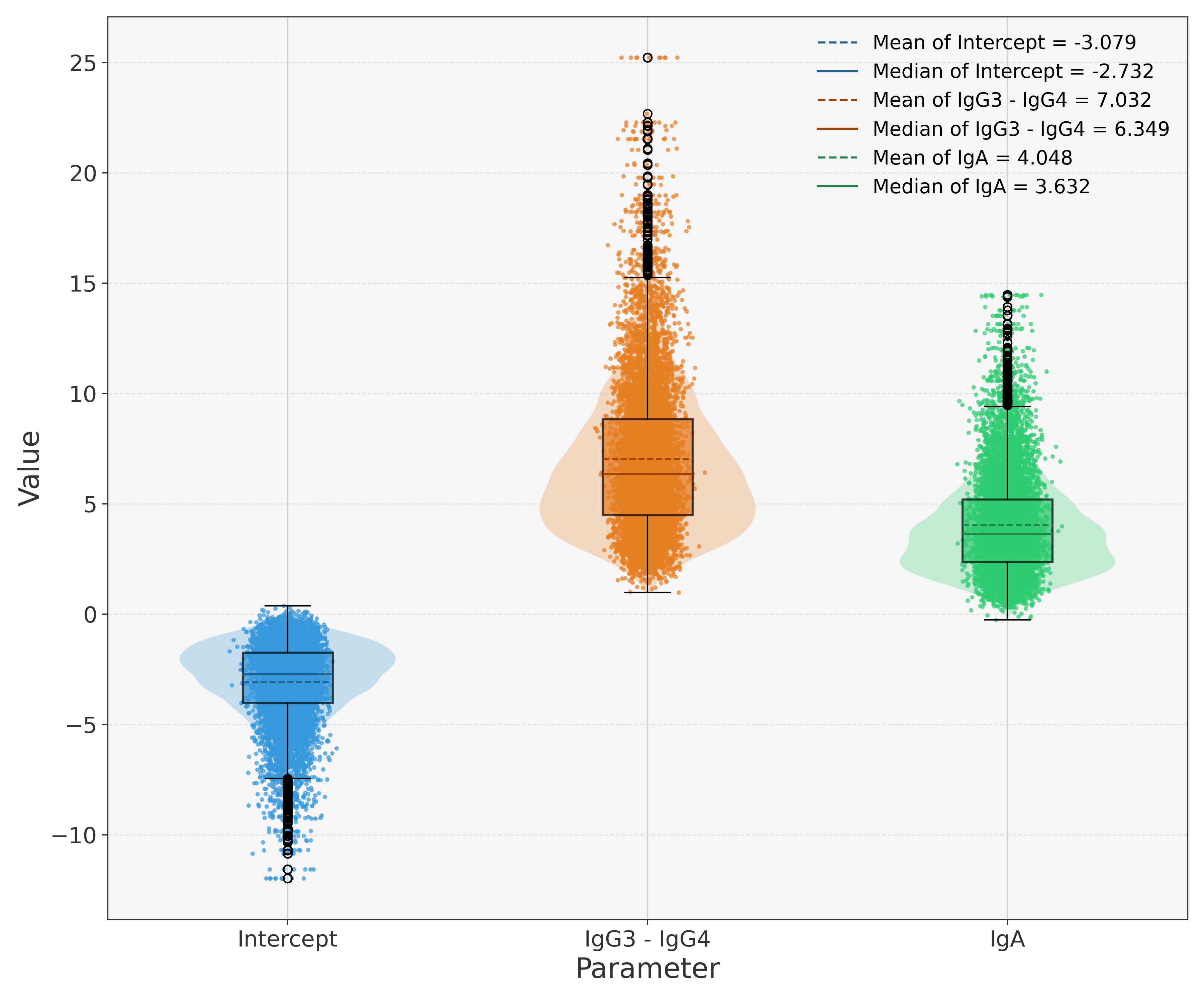}
        \caption{RRS - Boxplots of marginal distributions}
        \label{fig:LupusRRSBoxplot}
    \end{minipage}
    \hfill
    \begin{minipage}{0.48\textwidth}
        \centering
        \includegraphics[width=\linewidth]{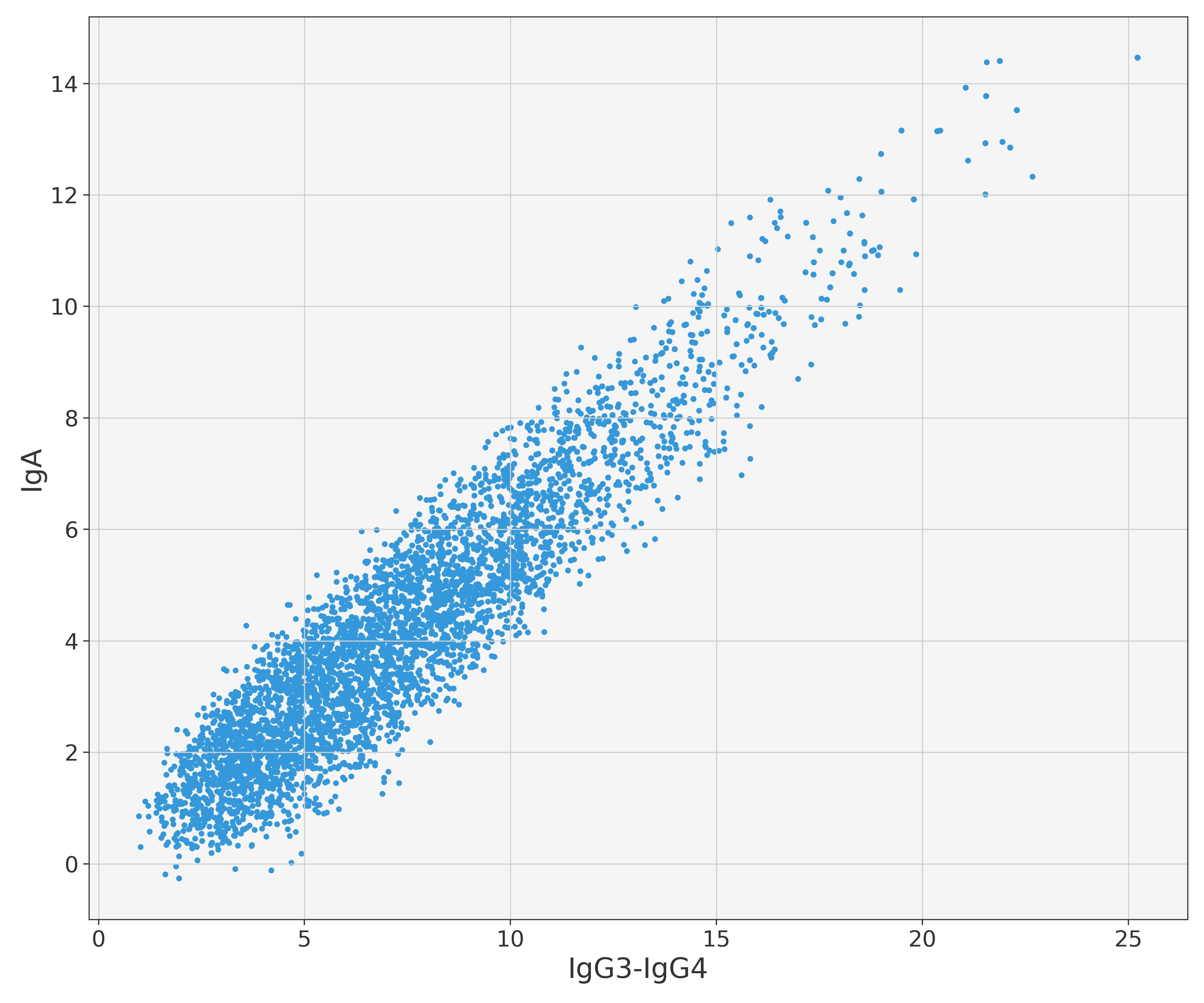}
        \caption{RRS - Scatter-plot of covariates}
        \label{fig:LupusRRSScatter}
    \end{minipage}
\end{figure}

To conclude, we also explore what happens at the autocorrelation level. 
\begin{figure}[htpb]
    \centering
    \includegraphics[width=0.55\linewidth]{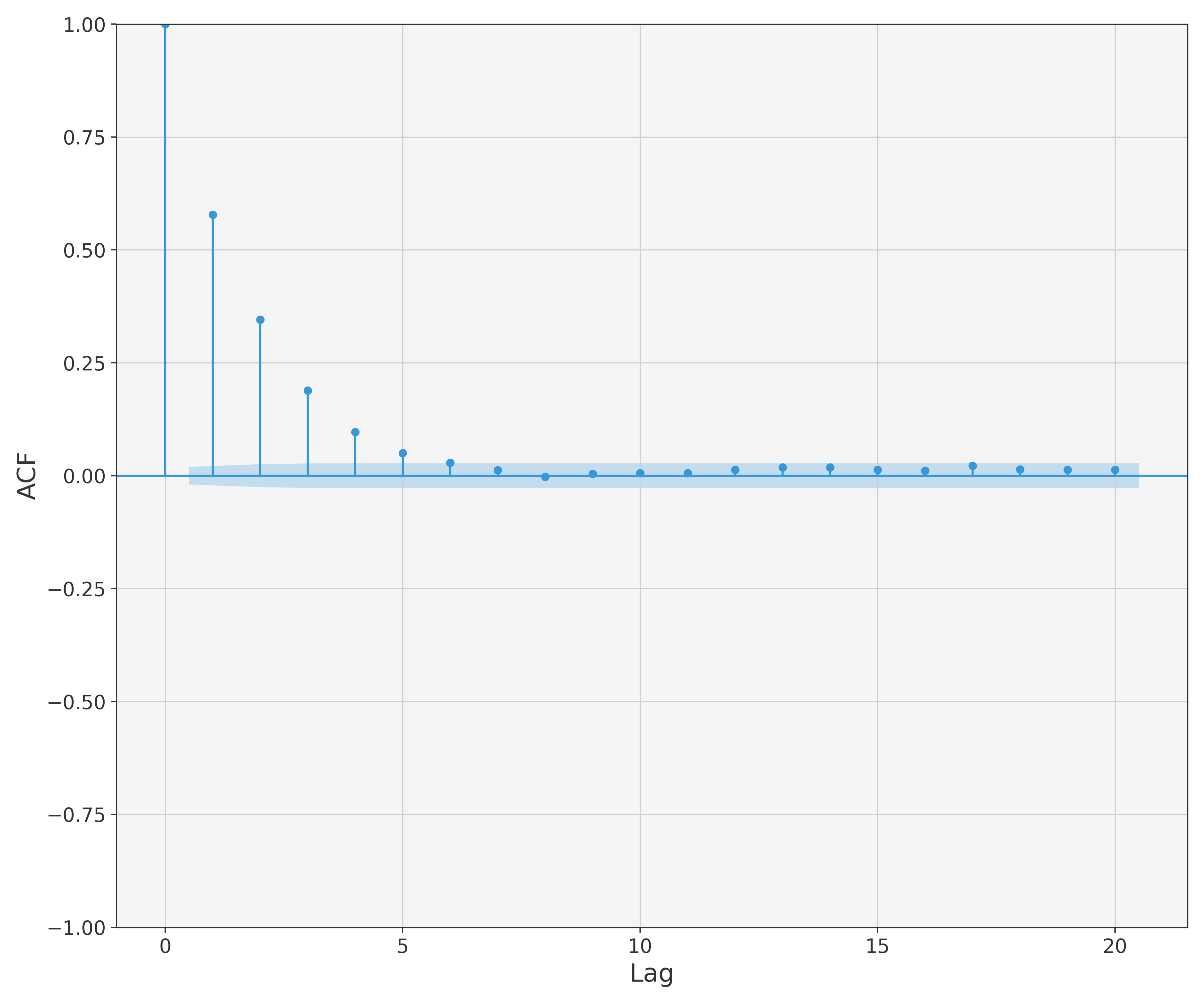}
    \caption{RRS - ACF plot of IgG3-IgG4}
    \label{fig:LupusRRSACF2}
\end{figure}
This time, as is clearly seen in Figure \ref{fig:LupusRRSACF2}, the autocorrelation of the IgG3-IgG4 chain decays very quickly. Even if we omitted the plots, we observe the same sharp autocorrelation decay also for the two other marginals, suggesting that the underlying regenerative process rapidly reaches its stationary distribution. Hence, we are able to obtain close-to-independence samples, which makes the procedure much more efficient compared to the Gibbs Sampler. To conclude, in this specific case, the Regenerative Rejection Sampling method behaves well in a multidimensional setting where a renowned MCMC method shows a sub-optimal performance.

\subsection{A Note on Efficiency}
As was done for the synthetic example, we compare the RRS method and the Gibbs sampler from an efficiency point of view, by computing the average number of samples generated in one second by the two algorithms. Similarly, we performed 100 independent runs of the two methods, calculating for each run the number of samples generated per second. In Table \ref{tab:SamplesSecondLupus} we report the average of the 100 values.

\begin{table}[htbp]
\centering
\caption{Average number of samples per second -- RRS vs Gibbs Sampler}
\label{tab:SamplesSecondLupus}
\begin{tabular}{p{4cm} p{3cm} p{3cm}}
\toprule
 & \textit{RRS} & \textit{Gibbs Sampler} \\
\midrule
\textit{Avg. Samples/Second} & 15893 & 8900 \\
\bottomrule
\end{tabular}
\end{table}

In this more complex setting, where each step of the methods is more expensive compared to the synthetic example, the RRS scheme shows a significantly lower computational cost, and, also for this reason, it should be the preferred method for this type of analyses.

\chapter{Conclusions}
\label{chap:8}
The main goal of this thesis was to present a novel approximate sampling method, recently introduced in \cite{Botev:MachineLearning}. It is based on regenerative processes, and it is called \emph{Regenerative Rejection Sampling}. The algorithm constructs a regenerative process with cycle lengths $W_i$ distributed as $W_i\sim f_\propto(X_i)/g(X_i),\: X_i\stackrel{iid}{\sim}g$, where $f_\propto$ is the target distribution, known up to a multiplicative constant, and $g$ is the proposal distribution. Such process converges to the required normalized distribution $f = f_\propto/\int f_\propto$. The method can be seen as an alternative to Rejection Sampling, to be used whenever the likelihood ratio $f_\propto/g$ is unbounded, or the bounding constant is too cumbersome to compute in practice.

After presenting the background theory on renewal and regenerative processes in Chapters \ref{chap:2} and \ref{chap:3}, in Chapter \ref{chap:4} we analyzed the coupling proof of the theorem concerning the exponential convergence of regenerative processes. As a result, we showed that the sufficient conditions for a regenerative process to converge exponentially fast to its stationary distribution mainly concern the cycle-length distribution: it must be spread-out, and it must have finite mean and exponential moments. A downside of this proof is that it does not give a practical recipe to compute the actual convergence rate.

Later, in Chapter \ref{chap:5}, we presented the RRS algorithm, and remarked that the theorem for regenerative processes can be straightforwardly applied to RRS to prove its convergence properties. However, since there is no way to obtain a general closed formula for the cycle-length distribution of RRS, the convergence analysis has to be performed on a case-to-case basis. Additionally, we compared the convergence rates of our RRS method to that of the Independent Metropolis-Hastings algorithm, and showed that whenever the likelihood ratio is bounded, both methods converge exponentially fast, even if the bounding constant is unknown. Furthermore, RRS converges exponentially fast also when the likelihood ratio is unbounded, provided that the cycle-length distribution satisfies the assumptions mentioned above. Under these conditions, the Independent Metropolis-Hastings algorithm can only have sub-geometric convergence. Even if the RRS method converges exponentially in more situations, it still inherits the Independence Sampler's limitations, together with those of Rejection Sampling. In fact, its performance is heavily dependent on the choice of proposal distribution \emph{and} time threshold $t$.

Moreover, we also showed, in Chapter \ref{chap:6}, that one can use the samples generated with RRS method to obtain estimates of quantities of the type $q = \mathbb{E}_F[h(Y)]$, where $F$ is the stationary distribution of the regenerative process. To this end, one builds the time-average \emph{ratio estimator} $\hat{q}(t)$. We have shown that the bias of $\hat{q}(t)$ decays with order $O(1/t^2)$, which is one order higher compared to the usual MCMC time-average estimator. However, we also remarked that the mean square error decay of the two estimators is of the same order, since the variance is the leading term of the expansion. Nonetheless, it is still favorable to use the low-bias estimator, if we want proper confidence interval coverage.\\
The novelty of Chapter \ref{chap:6} comes from the introduction of a computable non-asymptotic upper bound of $\hat{q}(t)$'s bias. Even if the bound decays with order $O(1/t^{3/2})$, it holds for all $t >0$, and it is completely estimable by simulation. Hence, it gives a practical way to control the bias during the simulation.

Lastly, we applied the RRS algorithm to one toy example and to one real dataset, and showed that RRS outperformed renowned MCMC methods.\\
In the former example, the goal was to sample from a synthetic bi-dimensional distribution first defined on a bounded domain, and then on the whole real plane $\mathbb{R}^2$. In both cases, the RRS method outperformed the Random Walk Metropolis algorithm, especially in terms of mixing/convergence.\\
In the latter, we performed a three-parameter Bayesian Probit regression on a real dataset containing data from $55$ patients, $18$ of which had been diagnosed with \emph{Latent membranous lupus nephritis}. For this example, we compared the performance of RRS and of the Gibbs Sampler, which exploited a latent variable structure of the model. In essence, by an inspection of the marginal distributions of the parameter, we showed that the Gibbs Sampler could not explore the space adequately, and that the chains mixing was significantly slow. On the other hand, the RRS method was able to move even in the tails of the marginal distributions, and to mix rapidly.\\
For both examples, we also showed, heuristically, that the RRS method's computational cost is lower compared to that of the MCMC methods that we analyzed. This implies that, for the same fixed amount of CPU time, the RRS method is able to compute, and hence return, more samples. 

The RRS algorithm showed good performances across different types of applications, hence it could be used together with Rejection Sampling to cover a larger basis of practical sampling situations. Furthermore, we theoretically proved that it converges exponentially fast for a wider class of situations compared to the Independent Metropolis-Hastings algorithm, which implies RRS could serve as a valuable alternative to the renowned MCMC method. For these reasons, and thanks to its simple and computationally-efficient implementation, we believe that the RRS method can be a very effective sampling method in a wide range of situations.

This introductory work paves the way for more advanced analyses of the RRS method, and for the development of modified (and hopefully better) versions of the algorithm. An interesting avenue could be to study the effect of dimensionality: in Chapter \ref{chap:7} we showed that, even in multi-dimensional settings, the RRS method seems to performs well. However, the number of dimensions was still limited to three. Hence, we suppose that with more complex applications in higher dimension, one could notice a decrease in performance.\\
Another promising direction could be to develop variants of the method that exploit in a more meaningful way the regenerative structure. For example, one could try to run i.i.d. Markov Chains, that have the same target distribution of the regenerative process, in each i.i.d. cycle, to further explore the sample space. This would obviously introduce more complexity, especially in the time-average ratio estimator setting, as it could not be possible to simplify the form of the general ratio estimator to obtain $\hat{q}(t)$. Furthermore, running a Markov Chain in each cycle would necessarily increase the computational cost of the algorithm, but if the correct type of transition kernel is chosen for the Markov Chain, we could observe a sharp improvement in the speed of convergence, meaning that the process could be ran for less time to obtain sufficiently good samples.\\
Lastly, we could apply the RRS method in the applicative situations where Rejection Sampling is employed. An example is that of \emph{Generative Adversarial Networks} (GANs), where Rejection Sampling has been used to improve their performance (see \cite{Azadi:DRS}). Since one of the practical problems that the authors of \cite{Azadi:DRS} encounter is the difficulty in computing the bounding constant of the likelihood ratio, RRS could serve as a valuable alternative.

\appendix
{
\pagestyle{appendixchapter}
\chapter{A Non-Trivial Example of Spread-Out Distribution}
\label{chap:NonTrivialSpreadOut}
Let us begin with a definition \cite{PeresSchlagSolomyak:BernoulliConvolutions,PeresSolomyak:BernoulliConvolutions}:
\begin{definition}[Infinite Bernoulli Convolution]
    We define the measure $\nu_\lambda$, for $\lambda\in(0,1)$, to be an \emph{Infinite Bernoulli Convolution} if it is the distribution of the random series
    \[
    \sum_{n=0}^\infty \pm \lambda^n,
    \]
    where the signs are given by $Ber(1/2)$ random variables.
\end{definition}

It is known that the measure $\nu_\lambda$ has compact support, and, for $\lambda<1/2$, is of Cantor-Lebesgue type, meaning that it is singular with respect to the Lebesgue measure \cite{PeresSchlagSolomyak:BernoulliConvolutions}. For this reason, it could serve as a non-trivial example of spread-out distribution, if we show that one of its convolution powers has an absolutely continuous component.

We report here a useful result from \cite[Corollary 1.6]{PeresSolomyak:BernoulliConvolutions}:
\begin{proposition}
    Let $n\geq1$, and denote $\tilde{\lambda}_n = \tbinom{2n}{n}2^{-2n}$. Then, the Fourier transform of $\nu_{\lambda}$ is in $L^{2n}(\mathbb{R})$ for almost every $\lambda\in(\tilde{\lambda}_n,1)$, and is not in $L^{2n}(\mathbb{R})$ for all $\lambda <\tilde{\lambda}_n$.
\end{proposition}
Let us take, for example, $n=2$. The Proposition implies that the Fourier transform of $\nu_\lambda$, $\widehat{\nu_{\lambda}}$, is in $L^4(\mathbb{R})$ for almost every $\lambda\in(3/8,1)$.

Moreover, we know that the Fourier transform of the $2^{nd}$ convolution power of $\nu_\lambda$ (i.e. $\widehat{\nu_{\lambda}\ast \nu_\lambda}$) is equal to $(\widehat{\nu_{\lambda}})^2$. Since $\widehat{\nu_\lambda}$ is in $L^4(\mathbb{R})$ for almost every $\lambda \in(3/8,1)$, we can conclude that $\widehat{\nu_{\lambda}\ast \nu_\lambda}=(\widehat{\nu_\lambda})^2\in L^2(\mathbb{R})$ for almost every $\lambda\in(3/8,1)$.

We can then use a result from Fourier Analysis \cite{Mattila:Fourier}, which states that if the Fourier transform of a finite Borel measure with compact support is in $L^2(\mathbb{R})$, then the measure is absolutely continuous with density in $L^2(\mathbb{R})$.
Hence, $\nu_\lambda\ast\nu_\lambda$, for $\lambda\in(3/8,1)$, is absolutely continuous.

Let us summarize everything in the following Proposition:
\begin{proposition}
    The measure $\nu_\lambda$ is \emph{singular} and \emph{spread-out} for almost every $\lambda\in(3/8,1/2)$, since $\nu_\lambda\ast\nu_\lambda$ is absolutely continuous.
\end{proposition}

\chapter{Nummelin Splitting}
\label{chap:NummelinSplitting}
This Appendix is devoted to the presentation of the technique of the \emph{Split Chain}, introduced independently and simultaneously by \cite{AthreyaNey:Splitting} and \cite{Nummelin:Splitting}, from two slightly different perspectives. The construction is used to embed a regenerative/renewal structure into a general state space Harris recurrent Markov Chain.
In the following pages we will concentrate on the derivation proposed by Nummelin \cite{Nummelin:Splitting}.

\section{Nummelin Splitting Technique}
\label{sec:NummelinSplittingTechinque}
When dealing with (recurrent) Markov Chains on a countable state space, a common strategy is to investigate the properties of the chain by fixing a state $i$ and using the independence of the paths between visits to $i$ \cite{Nummelin:Splitting}. However, if the state space is a general measurable space $(X,\mathcal{B}(X))$, this reasoning fails: in general, the chain does not visit with positive probability a single point $i\in X$. This is where the Splitting Technique acts: it introduces an artificial atom, which is visited with positive probability by the chain.

Let us assume that $\{\Phi_n\}_{n\in\mathbb{N}}$ is a Markov Chain on the general state space $(X,\mathcal{B}(X))$, with transition kernel $P(x,A)$, for $x\in X$ and $A\in\mathcal{B}(X)$. We also assume the chain to be $\varphi$-irreducible. To create the splitting, we need the following \emph{minorizing condition} \cite{MeynTweedie:MCSS}:
\begin{assumption}[Minorizing Condition]
\label{ass:Minorizing}
    For some $\delta>0$, some $C\in\mathcal{B}(X)$, and some probability measure $\nu$ concentrated on $C$, we have
    \[
    P(x,A)\geq \delta \nu(A),\quad A\in\mathcal{B}(X), x\in X.
    \]
\end{assumption}
The first step in the construction of the split chain is to split the state space and all the measures on $\mathcal{B}(X)$. The space is split by writing $\check{X}=X\times\{0,1\}$. $X_0=X\times\{0\}$ and $X_1=X\times\{1\}$ are copies of $X$, endowed with copies $\mathcal{B}(X_0),\mathcal{B}(X_1)$ of the $\sigma$-algebra $\mathcal{B}(X)$. The split state is then equipped with a split $\sigma$-algebra $\mathcal{B}(\check{X})$, generated by $\mathcal{B}(X_0),\mathcal{B}(X_1)$. In the rest of the discussion we use the following notation:
\begin{notation}
    For each $A\in\mathcal{B}(X)$, we write $A_0=A\times\{0\}$ and $A_1=A\times\{1\}$. Similarly, we write $x_0$ for the elements of $\check{X}$ living in the 0-level $X_0$, and $x_1$ for those in the 1-level $X_1$.
\end{notation}
Next, if $\lambda$ is a measure on $\mathcal{B}(X)$, we split it into two measures on $\mathcal{B}(X_0)$ and $\mathcal{B}(X_1)$, by defining a measure $\lambda^\ast$ on $\mathcal{B}(\check{X})$:
\begin{empheq}[left=\empheqlbrace]{align}
\lambda^\ast(A_0) &= \lambda(A\cap C)(1-\delta) + \lambda(A \cap C^c) \label{eq:lambdaA0}\\
\lambda^\ast(A_1) &= \lambda(A \cap C)\delta                                   \label{eq:lambdaA1}
\end{empheq}
where $\delta,C$ are defined in Assumption \ref{ass:Minorizing}. This definition is consistent, since $\lambda$ is the marginal measure induced by $\lambda^\ast$ \cite{MeynTweedie:MCSS}:
\[
\lambda^\ast(A_0\cup A_1) = \lambda(A), \quad A \in \mathcal{B}(X).
\]
Moreover, only subsets of $C$ are split by this construction, because $\lambda^\ast(A_0) = \lambda(A)$ for all $A\subseteq C^c$.

The last step is to create the split chain $\{\check{\Phi}_n\}_{n\in\mathbb{N}}=\{(\Phi_n,Y_n):\Phi_n\in X,Y_n\in \{0,1\}\}_{n\in\mathbb{N}}$, which is defined on $(\check{X}, \mathcal{B}(\check{X}))$. To do this, we construct the split transition kernel $\check{P}(x_i,A)$, for $x_i\in\check{X}$ and $A \in \mathcal{B}(\check{X})$:
\begin{empheq}[left=\empheqlbrace]{align}
\check{P}(x_0,\cdot) &= P(x,\cdot)^\ast, &\quad &x_0\in X_0\setminus C_0
  \label{eq:Px0-outC0}\\
\check{P}(x_0,\cdot) &= \dfrac{P(x,\cdot)^\ast-\delta \nu(\cdot)^\ast}{1-\delta} =: H(x,\cdot)^\ast, &\quad &x_0\in C_0
  \label{eq:Px0-inC0}\\
\check{P}(x_1,\cdot) &= \nu(\cdot)^\ast, &\quad &x_1\in X_1
  \label{eq:Px1}
\end{empheq}
where $\delta,C,\nu$ are defined in Assumption \ref{ass:Minorizing}. This is the central part of the construction, as the kernel defines how the chain moves in the state space. Let us explain more in detail how the movement works.\\
To this end, consider the set $C$. If the chain $\{\check{\Phi}_n\}_{n\in\mathbb{N}}$ is in $C^c$, it behaves exactly as $\{\Phi_n\}_{n\in\mathbb{N}}$, moving inside the 0-level $X_0$. This is because, by definition of the split kernel and of the split ($\ast$-) measure (cf. \eqref{eq:lambdaA0}, \eqref{eq:lambdaA1}), the chain moves to $C^c_1$ with probability 0. On the other hand, once the chain enters $C$, the situation becomes less straightforward. To ease the understanding, we compute explicitly the split kernel for a specific case: 
\begin{align}
        \check{P}(x_0,C_1) = \delta P(x,C), & \quad \check{P}(x_0,C_0) = (1-\delta)P(x,C), \quad& x_0 \in X_0\setminus C_0\\
        \check{P}(x_0,C_1) = \delta H(x,C), & \quad \check{P}(x_0,C_0) = (1-\delta)H(x,C), \quad& x_0 \in C_0\\
        \check{P}(x_0,C_1) = \delta, & \quad \check{P}(x_0,C_0) = (1-\delta), \quad& x_1 \in X_1
\end{align}
As we clearly see from these equations, as soon as the chain $\{\check{\Phi}_n\}_{n\in \mathbb{N}}$ enters $C$, it might jump to the 1-level with probability $\delta$, or remain in the 0-level with probability $1-\delta$, as if the level switch was controlled by a $Ber(s(x))$ random variable (independently), where $s(x) = \delta\mathbb{I}_C(x)$ \cite{Nummelin:IrreducibleMC}. When the chain moves to the 1-level, its next step will have the law $\nu$, while if it stays on the 0-level, the next step will follow the \emph{residual law} $H(x,A)$, which is non-negative due to the minorizing condition \ref{ass:Minorizing}.

The whole point of this construction was the creation of an atom. From the discussion above, it should be clear that such atom is the set $C_1$, which is reached with probability $\varphi^\ast(C_1) = \delta\varphi(C)>0$, whenever $\{\Phi_n\}_{n\in\mathbb{N}}$ is $\varphi$-irreducible (in reality, the atom is $X_1$, but, as was previously remarked, the chain can only reach the set $C_1$ in the 1-level).

To conclude, we mention the following theorem, which underlines the coherence and the importance of the splitting construction. We refer the reader to \cite[Theorem 5.1.3]{MeynTweedie:MCSS} for the details:
\begin{theorem}
    Consider the chain $\{\Phi_n\}_{n\in\mathbb{N}}$, and the associated split chain $\{\check{\Phi}_n\}_{n\in\mathbb{N}}$. We have:
    \begin{itemize}
        \item The chain $\{\Phi_n\}_{n\in\mathbb{N}}$ is the marginal chain of $\{\check{\Phi}_n\}_{n\in\mathbb{N}}$.
        \item The chain $\{\Phi_n\}_{n\in\mathbb{N}}$ is $\varphi$-irreducible if $\{\check{\Phi}_n\}_{n\in\mathbb{N}}$ is $\varphi^\ast$-irreducible, and if $\{\Phi_n\}_{n\in\mathbb{N}}$ is $\varphi$-irreducible with $\varphi(C)>0$ then $\{\check{\Phi}_n\}_{n\in\mathbb{N}}$ is $\nu^\ast$-irreducible, and $C_1$ is an accessible atom for the split chain.
    \end{itemize}
\end{theorem}

Any recurrent Markov Chain will then visit the atom infinitely often, and restart \emph{independently of all the past} after each visit with distribution $\nu$. This construction enables us to embed a regenerative structure (i.e. the subsequent visits to the atom) into the Markov Chain.

\chapter{Additional Proofs}
\label{chap:Proofs}
\section{Geometric Interpretation of Rejection Sampling}
\label{sec:ProofGeometricInterpretationRS}
In this Section, we present the proof of the Lemma on the geometric interpretation of the Rejection Sampling method (cf. Lemma \ref{lemma:geomInterp}). First, let us recall the statement:
\begin{lemma}
    Consider a non-negative integrable function $\tilde{h}\colon \mathbb{R}\rightarrow\mathbb{R}_+$, such that $\int_\mathbb{R} \tilde{h}\,\text{\emph{d}}y\neq 0$. Consider also the region $\mathcal{A}_{\tilde{h}}=\{(y,u)\in\mathbb{R}^2:0\leq u\leq \tilde{h}(y)\}$, and the normalized probability density function $h(y) = \tilde{h}(y)/\int_\mathbb{R} \tilde{h}(y)\text{\emph{d}}y$. Then, a pair of random variables $(Y,U)$ is uniformly distributed in $\mathcal{A}_{\tilde{h}}$ if and only if $Y\sim h$ and $U|Y\sim\mathcal{U}_{[0,\tilde{h}(Y)]}$.
\end{lemma}

Let us proceed with the proof.
\begin{proof}
    First, assume that $(Y,U)\sim\mathcal{U}_{\mathcal{A}_{\tilde{h}}}$. Note that $|\mathcal{A}_{\tilde{h}}|=\int_{\mathbb{R}}\tilde{h}(y)\,\text{d}y$. Hence, the density of the distribution of $(Y,U)$ is
    \[
    f_{(Y,U)}(y,u) = \frac{\mathbb{I}_{\mathcal{A}_{\tilde{h}}}(y,u)}{\int_{\mathbb{R}}\tilde{h}(y)\,\text{d}y}.
    \]
    Marginalize this density to obtain $f_Y(y)$:
    \[
    \begin{split}
    f_Y(y) = \int_{\mathbb{R}}f_{(Y,U)}(y,u) \, \text{d}u &= \frac{1}{\int_{\mathbb{R}}\tilde{h}(y)\,\text{d}y} \int_0^{\tilde{h}(y)} \mathbb{I}_{\mathbb{R}}(y)\,\text{d}u\\
    &= \frac{\tilde{h}(y)\mathbb{I}_{\mathbb{R}}(y)}{\int_{\mathbb{R}}\tilde{h}(y)\,\text{d}y}\\
    &= h(y)\mathbb{I}_{\mathbb{R}}(y),
    \end{split}
    \]
    which gives us $Y\sim h$.
    Moreover, for a given $y\in\mathbb{R}$, we also have
    \[
    f_{U|Y}(u|y) = \frac{f_{(Y,U)}(y,u)}{f_Y(y)} = \frac{\mathbb{I}_{\mathcal{A}_{\tilde{h}}}(y,u)}{\tilde{h}(y)\mathbb{I}_\mathbb{R}(y)} = \frac{\mathbb{I}_{[0,\tilde{h}(y)]}(u)}{\tilde{h}(y)},
    \]
    which proves $U|Y\sim\mathcal{U}_{[0,\tilde{h}(Y)]}$.

    Conversely, if $Y\sim h$ and $U|Y\sim \mathcal{U}_{[0,\tilde{h}(Y)]}$, we have
    \[
    f_{(Y,U)}(y,u) = f_{U|Y}(u|y) \, f_Y(y) = \frac{\mathbb{I}_{[0,\tilde{h}(y)]}(u)}{\tilde{h}(y)}h(y)\mathbb{I}_{\mathbb{R}}(y) = \frac{\mathbb{I}_{\mathcal{A}_{\tilde{h}}}(y,u)}{\int_{\mathbb{R}}\tilde{h}(y)\,\text{d}y},
    \]
    which concludes the proof.
\end{proof}

\section{Bias of Time-Average Estimator for Geometrically Ergodic MCMC}
\label{sec:BiasMCMC}

Given a Borel function $h$ such that $\mathbb{E}_F[|h(X)|]< \infty$, where $F$ is the target distribution of an MCMC method, the goal is to compute $\mu=\mathbb{E}_F[h(X)]$. Define the time-average estimator
\[
\hat{\mu}_N \colon = \frac{1}{N}\sum_{n=1}^N h(X_n),
\]
where the $X_n$'s are the samples generated with a run of the MCMC method.

The bias of this estimator satisfies the following Theorem.

\begin{theorem}
    Let $\{X_n\}_{n\in\mathbb{N}}$ be a Markov Chain (on the state space $\mathcal{X}$) with Metropolis-Hastings transition kernel $P$ and with initial distribution $\delta_x$, for a given $x\in\mathcal{X}$. Denote its invariant distribution by $\pi$. Assume that $\{X_n\}$ is geometrically ergodic, i.e. there exists $\gamma>0$ and $\varphi:\mathcal{X}\rightarrow\mathbb{R}_+$ such that $\|\pi^{n,\delta_x} - \pi\|_{TV} \leq \varphi(x)e^{-\gamma n}$, where $\pi^{n,\delta_x}$ represents the distribution of the chain with initial distribution $\delta_x$ at the $n$-th step. Then, for any bounded function $h:\mathcal{X}\rightarrow \mathbb{R}$, there exists $C_h$ such that 
    \[
    \left|\mathbb{E}\left[\hat{\mu}_N\right] - \mu\right| \leq \frac{C_h}{N}.
    \]
\end{theorem}

\begin{proof}
We have
\[
\begin{split}
    \left|\mathbb{E}\left[\hat{\mu}_N\right] - \mu\right| &= \left| \frac{1}{N} \sum_{n=1}^N \mathbb{E}[h(X_n)-\mu]\right|\\
    &\leq \frac{1}{N} \sum_{n=1}^N \left| \int_{\mathcal{X}} h(y) \left( \pi^{n,\delta_x}(\text{d}y)- \pi(\text{d}y) \right)\right|\\
    &\leq \frac{1}{N} \sum_{n=1}^N \sup_{y\in\mathcal{X}}|h(y)|\|\pi^{n,\delta_x} - \pi\|_{TV}\\
    &\leq \frac{1}{N} \sup_{y\in\mathcal{X}}|h(y)| \varphi(x)\frac{1}{1-e^{-\gamma}}.
\end{split}
\]
\end{proof}

}

\clearpage
\pagestyle{numberonly}
\printbibliography

\end{document}